\documentclass[11pt,reqno]{amsart}
\newtheorem{theorem}{Theorem}

\newtheorem{definition}{Definition}

\newtheorem{lemma}{Lemma}
\newtheorem{proposition}{Proposition}
\newtheorem{remark}{Remark}
\numberwithin{equation}{section}
\numberwithin{lemma}{section}
\numberwithin{theorem}{section}
\numberwithin{proposition}{section}
\numberwithin{remark}{section}
\numberwithin{corollary}{section}
\numberwithin{definition}{section}
\usepackage{amsmath}
\usepackage{amsthm}
\usepackage{amssymb}
\usepackage{amsfonts}
\usepackage{nicematrix}
\usepackage{graphicx}
\usepackage{subfigure}
\usepackage{stmaryrd}
\usepackage{caption}
\usepackage{import}
\usepackage{xifthen}
\usepackage{pdfpages}
\usepackage{transparent}

\usepackage{comment}
\usepackage{tikz}
\usepackage{url}
\usepackage[linktocpage]{hyperref}
\hypersetup{
    colorlinks=true,
    citecolor=blue,
    linkcolor=blue,
    filecolor=blue,
    urlcolor=blue,}
\calclayout 

\begin{document}

\title[Hamiltonian formulation for axisymmetric magnetohydrodynamics]{Hamiltonian formulation and matrix discretization for axisymmetric magnetohydrodynamics}

\author{Michael Roop}
\address{Michael Roop: Department of Mathematical Sciences, Chalmers University of Technology and University of Gothenburg, 412 96 Gothenburg, Sweden}
\email{michael.roop@chalmers.se}



\begin{abstract}
Equations of ideal magnetohydrodynamics (MHD) play an important role in the studies of turbulence, astrophysics, and plasma physics. These equations possess remarkable geometric structures and symmetries. Indeed, they admit a geodesic formulation in the sense of Arnold, as a Lie--Poisson flow on the dual of an infinite-dimensional Lie algebra. Zeitlin's model, previously developed for MHD on the flat torus and the two-sphere, is a matrix approximation of MHD consistent with the underlying geometric structures. In this paper, we derive the reduced model of axially symmetric magnetohydrodynamics on the three-sphere and give its Hamiltonian formulation. We further extend finite dimensional Zeitlin's matrix model for MHD from 2D to axially symmetric 3D flows of magnetized fluids, yielding the first discrete model for 3D magnetohydrodynamics compatible with the underlying Lie--Poisson structure.
\end{abstract}

\maketitle
\tableofcontents
\section{Introduction}
Since the discovery of the geodesic nature of fluid equations, originating from the seminal paper by V.I. Arnold \cite{Arn}, it has become evident that Hamiltonian structures play a crucial role in understanding fluid motion and hydrodynamical turbulence. Arnold discovered that the incompressible Euler equations constitute a geodesic equation on the infinite-dimensional group of volume preserving diffeomorphisms of a fluid domain with respect to the right-invariant Riemannian metric induced by kinetic energy of the fluid. From the Hamiltonian perspective, the Euler equations are a Lie--Poisson flow on the dual of the infinite-dimensional Lie algebra of divergence-free vector fields. Since then, many equations of mathematical physics have been shown to possess Hamiltonian formulations in terms of Lie--Poisson structures on duals of infinite-dimensional Lie algebras \cite{KhMisMod}. In particular, the equations of incompressible magnetohydrodynamics (MHD) share a similar formulation. These equations describe a conductive incompressible fluid moving in an ambient electromagnetic field. The presence of the electromagnetic field corrects the Euler equations by the Lorentz force, along with simultaneous transport of the magnetic field by the fluid. The corresponding configuration space for ideal MHD is the \textit{magnetic extension} of the group of volume preserving diffeomorphisms $\mathrm{Diff}_{\mu}(M)$, which is the \textit{semidirect product} of the group $\mathrm{Diff}_{\mu}(M)$ and the dual of its Lie algebra \cite{ArnKh,MoGr1980,VD}. In the two-dimensional case, the equations of magnetohydrodynamics reduce to a two scalar field model, called \textit{reduced MHD} \cite{Strauss1976}. This model admits numerous generalizations, such as Hazeltine's three-field model \cite{Haz,HazHolm,Holm}, a four-field model for tokamak dynamics \cite{HHM1987}, and an inclusive model for inertia MHD and Hall effect \cite{GTAM2017}. All of these have a semidirect product structure of the underlying Lie algebra.

The Hamiltonian nature of ideal hydrodynamics and magnetohydrodynamics implies symplecticity of the flow and existence of conservation laws called \textit{Casimirs}, for example, magnetic helicity and cross-helicity. These conserved quantities play a crucial role in the long term dynamics, as the corresponding flow remains confined to the \textit{coadjoint orbits} --- invariant submanifolds in the phase space where the Casimirs are constant \cite{MarsRat}. To study the dynamics numerically, one needs to aim for discrete models compatible with the rich geometry of the phase space for both Euler's equations and MHD. Such discrete models provide a better insight into the statistical behavior, compared to non-Lie--Poisson schemes, for example, finite elements or finite volumes. Matrix discretizations for the Euler equations and MHD were found in a series of works by V. Zeitlin, on the flat torus \cite{Ze1991,Ze2005}, and on the two-sphere \cite{Zeit}. The approximating sequence of Lie groups is given by finite-dimensional matrix Lie groups $\mathrm{SU}(N)$, where the dimension of the matrix $N$ serves as a truncation parameter in the Fourier decomposition of the fields. Zeitlin's matrix equations have recently been utilized to make important conclusions about hydrodynamical turbulence on the two-sphere \cite{CiViMo2023,ModViv1}. For reduced MHD on the two-sphere, a fully discrete approximation was developed in \cite{ModRoop} and used to study MHD turbulence in \cite{ModRoopJPP}.

In the present work, we extend the matrix approach to MHD from two dimensions to axially symmetric flows on the three-sphere $\mathbb{S}^{3}$. The key idea is to make use of the Hopf fibration and is summarized as follows: the quotient of the action of the Hopf field on $\mathbb{S}^{3}$ is the two-dimensional sphere $\mathbb{S}^{2}$, and therefore if one assumes the $\mathbb{S}^{1}$-symmetry, generated by the Hopf field, of solutions to MHD equations on $\mathbb{S}^{3}$, one gets the symmetry reduced version of three-dimensional MHD on the two-sphere. This idea has previously been used for the Euler equations on $\mathbb{S}^{3}$, see \cite{ModinPreston2025}, and is adopted here for the case of ideal MHD. The resulting reduced MHD equations on $\mathbb{S}^{2}$ are formulated in terms of four fields, and their Hamiltonian formulation in terms of semidirect product extensions of semidirect product extensions is given. This structure reveals a family of Casimir invariants parameterized by four arbitrary functions. The obtained system of equations can be approximated by matrix equations in a way similar to the two-dimensional case. Further, we address the problem of structure preserving time integration of the matrix equations and develop two time integration schemes fully compatible with the Lie--Poisson structure of the matrix equations.

The paper is organized as follows. In Sect.~\ref{sec2}, we give a brief overview of Hamiltonian structures for MHD models in two and three dimensions. Further, in Sect.~\ref{sec3} we derive the reduced version of three-dimensional axisymmetric MHD and provide its Hamiltonian formulation along with the description of Casimirs. In Sect.~\ref{sec4}, we give the matrix version of the derived equations, describe the Lie--Poisson formulation for the matrix equations, and discuss convergence results for Casimirs. Finally, in Sect.~\ref{sec5}, we discuss the Lie--Poisson time integration of the matrix equations.

\vspace{0.5cm}
\noindent
{\bf Acknowledgments.} This work was supported by the Knut and Alice Wallenberg Foundation, grant number WAF2019.0201, by the Swedish Research Council, grant number 2022-03453, the G\"oran Gustafsson Foundation for Research in Natural Sciences and Medicine, and by the Royal Swedish Academy of Sciences (project MA2025-0073). The author is grateful to Klas Modin for fruitful discussions.
\section{Hamiltonian formulation of MHD}
\label{sec2}
The system of self-consistent MHD describes the evolution of a divergence-free velocity field $u(t,x)$ and a magnetic field $B(t,x)$ on a three-dimensional Riemannian manifold $(M,g)\ni  x$:
\begin{equation}
\label{MHD}
\left\{
\begin{aligned}
&\dot u+\nabla_{u}u=-\nabla p+\mathrm{curl}B\times B, \\
&\dot B=\operatorname{curl}(u\times B),\\
&\operatorname{div}B=0,\\
&\operatorname{div}u=0,
\end{aligned}
\right.
\end{equation}
where $p(t,x)$ is the pressure function, $\nabla_{u}u$ is the covariant derivative of the vector field $u$ along itself. The term $\mathrm{curl}B\times B$ represents the Lorentz force, and the second equation in \eqref{MHD} reflects the frozenness of the magnetic field into the fluid. The total energy of a magnetic fluid is the sum of its kinetic and magnetic energies:
\begin{equation}
\label{energymetric}
E(u,B)=\frac{1}{2}\int_{M}\left(|u|^{2}+|B|^{2}\right)\mu,
\end{equation}
where $\mu$ is the Riemannian volume form on $M$ induced by the metric $g$, and $|\cdot|^{2}=g(\cdot,\cdot)$.
\subsection{Hamiltonian structures for 3D MHD}
Following \cite{ArnKh}, we introduce the configuration space $F$ for a magnetic fluid as a \textit{magnetic extension} of the group of volume preserving diffeomorphisms of $M$ given by the semidirect product:
\begin{equation*}
F=\mathrm{Diff}_{\mu}(M)\ltimes\mathfrak{X}_{\mu}^{*}(M),
\end{equation*}
where $\mathfrak{X}_{\mu}^{*}(M)\simeq\Omega^{1}(M)/\mathrm{d}\Omega^{0}(M)$ is the (smooth) dual of the Lie algebra $\mathfrak{X}_{\mu}(M)$ of divergence free vector fields on $M$. The group structure on $F$ is
\begin{equation*}
(\phi,a)\circ(\psi,b)=(\phi\circ\psi,\mathrm{Ad}^{*}_{\psi}a+b),\quad \phi,\psi\in\mathrm{Diff}_{\mu}(M),\,a,b\in\mathfrak{X}_{\mu}^{*}(M).
\end{equation*}
The Lie algebra $\mathfrak{f}=\mathfrak{X}_{\mu}(M)\ltimes\left(\Omega^{1}(M)/\mathrm{d}\Omega^{0}(M)\right)$ corresponding to the Lie group $F$ is a vector space of pairs $(v,a)\in\mathfrak{f}$, with $v\in\mathfrak{X}_{\mu}(M)$ and $a\in\Omega^{1}(M)/\mathrm{d}\Omega^{0}(M)$, and the Lie algebra structure is given by
\begin{equation}
\label{Liealgebrastructureonsemidir}
[(v,a),(w,b)]=([v,w],\mathrm{ad}^{*}_{w}a-\mathrm{ad}^{*}_{v}b)=([v,w],L_{v}b-L_{w}a),
\end{equation}
for $v,w\in\mathfrak{X}_{\mu}(M),\,a,b\in\Omega^{1}(M)/\mathrm{d}\Omega^{0}(M)$, and where $L_{v}a$ is the Lie derivative of a 1-form $a\in\Omega^{1}(M)/\mathrm{d}\Omega^{0}(M)$ along the vector field $v\in\mathfrak{X}_{\mu}(M)$.

The dual Lie algebra $\mathfrak{f}^{*}$ of the magnetic extension $F$ is $\mathfrak{f}^{*}=\mathfrak{X}_{\mu}(M)\times(\Omega^{1}(M)/\mathrm{d}\Omega^{0}(M))$, and the coadjoint action of the Lie algebra $\mathfrak{f}$ on its dual $\mathfrak{f}^{*}$ is 
\begin{equation}
\label{coadjopsemidir}
\mathrm{ad}^{*}_{(v,a)}(w,b)=\left([w,v],\mathrm{ad}^{*}_{v}b-\mathrm{ad}^{*}_{w}a\right)=([w,v],L_{w}a-L_{v}b).
\end{equation}

The Riemannian metric on $M$ \eqref{energymetric} induces an isomorphism $I\colon\mathfrak{g}\to\mathfrak{g}^{*}$, with $\mathfrak{g}=\mathfrak{X}_{\mu}(M)$, and
therefore the general form of the Euler--Arnold equation on $\mathfrak{f}^{*}$ is
\begin{equation}
\label{EuArnsemidir}
\frac{\mathrm{d}}{\mathrm{d}t}(w,b)=\mathrm{ad}^{*}_{(v,a)}(w,b)\Longleftrightarrow\dot w=[w,v],\quad\dot b=L_{w}a-L_{v}b,\quad a=I(w),\quad v=I^{-1}(b).
\end{equation}

Let us now make a connection between the abstract equations \eqref{EuArnsemidir} and the MHD system \eqref{MHD}. Let $u,B\in\mathfrak{X}_{\mu}(M)$ be the velocity field of a fluid and the magnetic field correspondingly. We observe that the inertia operator $I$ induced by the metric \eqref{energymetric} is assigning to a divergence-free vector field a coset in the dual space:
\begin{equation*}
\mathcal{B}=I(B)=[B^{\flat}]\in\mathfrak{X}^{*}_{\mu}(M),\quad m=I(u)=[u^{\flat}]\in\mathfrak{X}^{*}_{\mu}(M),
\end{equation*}
and therefore the induced Hamiltonian function on $\mathfrak{f}^{*}$ is
\begin{equation}
\label{hamiltonfstar}
H(m,B)=\frac{1}{2}\int_{M}\left(\langle m,I^{-1}(m)\rangle+\langle I(B),B\rangle\right)\mu,
\end{equation}
where $\langle\cdot,\cdot\rangle$ is the standard pairing between a Lie algebra and its dual.

Assigning in \eqref{EuArnsemidir} $v=u$, $w=B$, $a=\mathcal{B}$, $b=m$, we obtain the following system:
\begin{equation}
\label{almostMHD}
\dot B=[B,u],\quad \dot m=-L_{u}m+L_{B}\mathcal{B},\quad\mathrm{div}(u)=\mathrm{div}(B)=0,
\end{equation}
and finally using the identities
\begin{equation*}
\left(L_{u}u^{\flat}\right)^{\sharp}=\nabla_{u}u+\nabla P_{1},\quad\left(L_{B}B^{\flat}\right)^{\sharp}=\mathrm{curl}B\times B+\nabla|B|^{2},
\end{equation*}
we conclude that equations \eqref{almostMHD} are equivalent to \eqref{MHD}. We summarize the above discussion in the following theorem \cite{ArnKh,MarsRatWein,VD} 
\begin{theorem}
The equations of ideal magnetohydrodynamics \eqref{MHD} are a Lie--Poisson system on the dual $\mathfrak{f}^{*}$ of the semidirect product Lie algebra $\mathfrak{f}=\mathfrak{X}_{\mu}(M)\ltimes\left(\Omega^{1}(M)/\mathrm{d}\Omega^{0}\right(M))$. The Hamiltonian function is given by \eqref{hamiltonfstar}.
\end{theorem}

The Lie--Poisson nature of the flow \eqref{MHD} suggests that it has a number of conserved quantities. First, let $M$ be simply connected, then, any divergence-free vector field $B\in\mathfrak{X}_{\mu}(M)$ on $M$ has a vector potential $A=\mathrm{curl}^{-1}(B)$. The following quantities are known to be the \textit{Casimir invariants} for \eqref{MHD} (see, for example, \cite{VD}):
\begin{equation}
\label{casMHD}
\mathcal{C}(B)=\int_{M}g( A,B)\mu,\quad\mathcal{I}(u,B)=\int_{M}g(u,B)\mu.
\end{equation}
The Casimir $\mathcal{C}(B)$ is usually referred to as \textit{magnetic helicity}, and $\mathcal{I}(u,B)$ as \textit{cross-helicity}. They are known to be the only independent Casimirs (see \cite{KhPerYang2019} for details). The energy \eqref{hamiltonfstar} is also conserved, but is not a Casimir invariant. 
\subsection{Hamiltonian structures for 2D MHD}
The two-dimensional case is somewhat special in hydrodynamics and magnetohydrodynamics, as in this case the reduced version of the three-dimensional equations \eqref{MHD} admits an infinite collection of Casimirs. Here, we give a brief description of the Hamiltonian structures and infinite-dimensional Lie--Poisson brackets for reduced two-dimensional magnetohydrodynamics.

First, we reformulate equations \eqref{MHD} in terms of four scalar fields instead of two vector fields. Let $(M,g)$ be a Riemannian two-dimensional simply connected manifold. Let $\mu$ still be the Riemannian volume form on $M$. In this setting, divergence-free vector fields $u$ and $B$ are defined by their \textit{stream functions}, or \textit{Hamiltonians}:
\begin{equation*}
u=X_{\psi},\quad B=X_{\theta}.
\end{equation*}
The second pair of fields $(\omega,j)$ on $M$ is \textit{vorticities}:
\begin{equation*}
\omega\mu=(\mathrm{d}m)\mu,\quad j\mu=(\mathrm{d}\mathcal{B})\mu.
\end{equation*}
Then, applying the differential to the second equation in \eqref{almostMHD} and using in the first equation in \eqref{almostMHD} that the map $\psi\mapsto X_{\psi}$ is a Lie algebra isomorphism, we get the \textit{vorticity formulation} of 2D MHD:
\begin{equation}
\label{MHDvort}
\left\{
\begin{aligned}
&\dot\omega=\left\{\omega,\psi\right\}+\left\{\theta,j\right\}, \quad &\omega =\Delta\psi, \\
&\dot\theta=\left\{\theta,\psi\right\}, \quad &j=\Delta\theta, 
\end{aligned}
\right.
\end{equation}
where $\left\{\cdot,\cdot\right\}$ is the Poisson bracket on $M$, i.e. $\left\{f,h\right\}\mu=\mathrm{d}f\wedge\mathrm{d}h$ for $f,h\in C^{\infty}(M)$. The field $j$ is also called \textit{current density}.

The energy \eqref{hamiltonfstar} reduces to
\begin{equation}
\label{hamRMHD}
H=\frac{1}{2}\int_{M}\left(|\nabla\psi|^{2}+|\nabla\theta|^{2}
\right)\mu = -\frac{1}{2}\int_{M}(\omega\psi+\theta j)\mu.
\end{equation}

The system \eqref{MHDvort} admits a non-canonical Hamiltonian formulation in terms of the \textit{semidirect product Lie--Poisson bracket} \cite{HazMorr,HolmKuper1983,MoGr1980}:
\begin{equation}
\label{RMHDbracket}
\begin{split}
 \llbracket F,G\rrbracket&=\int_{M}\left[\omega\left\{\frac{\delta F}{\delta \omega},\frac{\delta G}{\delta \omega}\right\}+
 \theta\left(\left\{\frac{\delta F}{\delta \theta},\frac{\delta G}{\delta \omega}\right\}+\left\{\frac{\delta F}{\delta \omega},\frac{\delta G}{\delta \theta}\right\}\right)\right]\mu.
 \end{split}
\end{equation}
The bracket \eqref{RMHDbracket} turns equations \eqref{MHDvort} into
\begin{equation*}
\dot F=\llbracket F,H\rrbracket,
\end{equation*}
where $F$ is an observable of the fields $\omega$ and $\theta$.

One can readily check that the following quantities are Casimirs for the flow \eqref{MHDvort}:
\begin{equation}
\label{cascont}
\mathcal{C}_{f}=\int_{M}f(\theta)\mu,\quad\mathcal{I}_{\textsl{g}}=\int_{M}\omega \textsl{g}(\theta)\mu,
\end{equation}
for arbitrary smooth functions $f$ and $\textsl{g}$. Indeed, one has $\llbracket \mathcal{C}_{f},\mathcal{J}\rrbracket=\llbracket \mathcal{I}_{\textsl{g}},\mathcal{J}\rrbracket=0$ for any functional $\mathcal{J}$. The Casimir $\mathcal{I}_{\textsl{g}}$ descends from the cross-helicity Casimir $\mathcal{I}(u,B)$ in \eqref{casMHD}, and therefore will also be called cross-helicity. However, the Casimir $\mathcal{C}_{f}$ does not, but we will still call it magnetic helicity. The presence of the magnetic helicity Casimir reflects the transport property of the field $\theta$ by the field $\psi$.

This rich geometric structure of the phase space for MHD in both three and two dimensions determines important features of the MHD turbulence and the long time behavior of solutions to the MHD equations. Indeed, the flow defined by equations \eqref{MHD} or \eqref{MHDvort}, once started on a coadjoint orbit, remains confined to it forever. Since Casimir functions are constants on the coadjoint orbits, it is important to develop discrete approximations for \eqref{MHD} and \eqref{MHDvort} compatible with the aforementioned geometric structures, and in particular, Casimir-preserving. The necessity of such discretizations, in turn, follows from the fact that as of today, the only way to validate predictions of statistical theories for systems describing nonlinear chaotic behavior of fluids is via numerical simulations.

Such discrete models for two-dimensional hydrodynamics and MHD have been developed in the works of Zeitlin \cite{Ze1991,Zeit,Ze2005}. Zeitlin's approach to MHD has recently been taken over and developed for the spherical MHD in the works \cite{ModRoop,ModRoopJPP}. Even though it is the flat torus $\mathbb{T}^{2}$ that is a typical domain for MHD simulations, we underline that the matrix equations in both cases of the sphere and the torus coincide and therefore share the same Lie algebra structure. This is shown in detail in Sect.~\ref{Zeitlin2dMHDtorus} and \ref{Zeitlin2dMHDsphere}.

Until recently, matrix hydrodynamics approach has been thought to be feasible only for two-dimensional hydrodynamics, such as Euler's and MHD equations on the flat torus and the two-sphere. 
As noted by Zeitlin \cite{Ze1991}, to get 3D hydrodynamics as a limit of finite-mode truncations is impossible due to only a finite number of independent Casimirs in 3D, contrary to 2D with an infinite collection of Casimirs. This can be viewed as a ``no-go'' theorem for 3D hydrodynamics.
However, in the work \cite{ModinPreston2025} the matrix hydrodynamics approach was shown to be possible for the three-dimensional axisymmetric Euler equations on the three-sphere $\mathbb{S}^{3}$. We take this approach further and develop the discretization for the axisymmetric three-dimensional MHD on $\mathbb{S}^{3}$, which serves as a middle ground between two- and three-dimensional theories, the so-called ``two-and-a-half''-dimensional flows, where equations are formulated on $\mathbb{S}^{2}$, but the effects of three-dimensionality of the flow still retain. 
\section{Axisymmetric MHD equations}
\label{sec3}
In this section, we derive the reduced model for axisymmetric MHD flows on $M=\mathbb{S}^{3}$. The reduction is provided by the action of the Killing vector field $K$ on $\mathbb{S}^{3}$. Killing vector fields generate isometries, i.e. $L_{K}g=0$. On the three-sphere $\mathbb{S}^{3}$, such vector fields generate rotations around one of the axes, and invariance of the MHD flow with respect to such action means that the fields $B$ and $u$ commute with $K$, i.e. $[B,K]=[u,K]=0$. Since $K$ is an infinitesimal symmetry of \eqref{MHD}, one can perform the symmetry reduction and obtain the symmetry reduced version of equations \eqref{MHD}. This will be the main goal of this section.
\subsection{Derivation of reduced equations}
Let $\mathbb{S}^{3}$ be the three-sphere embedded in $\mathbb{R}^{4}(x,y,z,w)$:
\begin{equation*}
\mathbb{S}^{3}=\left\{(x,y,z,w)\in\mathbb{R}^{4}\mid x^{2}+y^{2}+z^{2}+w^{2}=1\right\}\subset\mathbb{R}^{4}.
\end{equation*}
The restriction of the Euclidean metric on $\mathbb{R}^{4}$ onto $\mathbb{S}^{3}$ makes the three-sphere a Riemannian manifold $(\mathbb{S}^{3},g)$. Let $\mu$ be the associated Riemannian volume form. The orthogonal basis of vector fields on $\mathbb{S}^{3}$ is given by
\begin{align*} 
E_{1} &= \frac{1}{2}(-x\partial_{w}+w\partial_{x}-z\partial_{y}+y\partial_{z}),\\ 
E_{2} &=  \frac{1}{2}(-y\partial_{w}+z\partial_{x}+w\partial_{y}-x\partial_{z}),\\
E_{3} &= \frac{1}{2}(-z\partial_{w}-y\partial_{x}+x\partial_{y}+w\partial_{z}),
\end{align*}
and the Lie algebra structure is given by
\begin{equation}
\label{LieAlgStr}
[E_{1},E_{2}]=-E_{3},\quad[E_{2},E_{3}]=-E_{1},\quad[E_{3},E_{1}]=-E_{2}.
\end{equation}
Then, if $E_{i}^{*}$, $i=1,2,3$, is the co-frame, i.e. $E_{i}^{*}(E_{j})=\delta_{ij}$, then the Riemannian volume form $\mu$ becomes $\mu=E_{1}^{*}\wedge E_{2}^{*}\wedge E_{3}^{*}$.

We define the three-dimensional vorticity fields $\omega$ and $j$ as follows:
\begin{equation}
\label{vort3d}
\iota_{\omega}\mu=\mathrm{d}m,\quad\iota_{j}\mu=\mathrm{d}\mathcal{B},
\end{equation}
where $\iota_{\omega}\mu$ is the contraction of the form $\mu$ with the vector field $\omega$.

Then, applying the differential to the equation $\dot m=-L_{u}m+L_{B}\mathcal{B}$ and using the definitions \eqref{vort3d}, we get
\begin{equation*}
\iota_{\dot\omega}\mu=-L_{u}(\iota_{\omega}\mu)+L_{B}(\iota_{j}\mu).
\end{equation*}
Let us recall now the following relation between the inner product and the Lie derivative:
\begin{equation*}
\iota_{[u,\omega]}\mu=[L_{u},\iota_{\omega}]\mu=L_{u}(\iota_{\omega}\mu)-\iota_{\omega}(L_{u}\mu)=L_{u}(\iota_{\omega}\mu),
\end{equation*}
where the last term vanishes, because $u$ is a volume preserving vector field. We thus conclude
\begin{equation*}
-L_{u}(\iota_{\omega}\mu)=\iota_{[\omega,u]}\mu,\quad L_{B}(\iota_{j}\mu)=\iota_{[B,j]}\mu,
\end{equation*}
and, since $\mu$ is non-degenerate, the vorticity formulation of \eqref{MHD} reads
\begin{equation}
\label{vectorfieldvortform}
\dot\omega=[\omega,u]+[B,j],\quad\dot B=[B,u].
\end{equation}

Let us choose the vector field $K=E_{1}$ to generate the rotational symmetry, which will be used to get the symmetry-reduced equations. In this case, the fields $B$ and $u$ are assumed to be $E_{1}$-invariant, i.e. $[B,E_{1}]=[u,E_{1}]=0$, and therefore they are generated by two $E_{1}$-invariant functions, $u\sim(\tilde\sigma,\tilde\psi)$ and $B\sim(\tilde\rho,\tilde\xi)$:
\begin{equation}
\begin{split}
\label{ubS3}
u&=\tilde\sigma E_{1}-(E_{3}\tilde\psi)E_{2}+(E_{2}\tilde\psi)E_{3},\quad E_{1}\tilde\sigma=E_{1}\tilde\psi=0,\\
B&=\tilde\rho E_{1}-(E_{3}\tilde\xi)E_{2}+(E_{2}\tilde\xi)E_{3}, \quad E_{1}\tilde\rho=E_{1}\tilde\xi=0.
\end{split}
\end{equation}

To get the system of equations for the fields $(\tilde\sigma,\tilde\psi)$ and $(\tilde\rho,\tilde\xi)$, one needs to express also the vorticities $(\omega,j)$ in terms of these functions. First, we observe that $u^{\flat}=\tilde\sigma E_{1}^{*}-(E_{3}\tilde\psi)E_{2}^{*}+(E_{2}\tilde\psi)E_{3}^{*}$. Further, 
\begin{equation}
\label{dufaltexpr}
\mathrm{d}u^{\flat}=\mathrm{d}\tilde\sigma \wedge E_{1}^{*}-\mathrm{d}(E_{3}\tilde\psi)\wedge E_{2}^{*}+\mathrm{d}(E_{2}\tilde\psi)\wedge E_{3}^{*},
\end{equation}
and since, for any function $f$, we have $\mathrm{d}f=E_{1}(f)E_{1}^{*}+E_{2}(f)E_{2}^{*}+E_{3}(f)E_{3}^{*}$, equation \eqref{dufaltexpr} transforms
\begin{equation}
\label{uflatfinal}
\mathrm{d}u^{\flat}=-(E_{2}\tilde\sigma+E_{2}\tilde\psi)E_{1}^{*}\wedge E_{2}^{*}-(E_{3}\tilde\sigma+E_{3}\tilde\psi)E_{1}^{*}\wedge E_{3}^{*}+(E_{2}^{2}\tilde\psi+E_{3}^{2}\tilde\psi)E_{2}^{*}\wedge E_{3}^{*}.
\end{equation}
Looking for the vorticity field in the form $\omega=W_{1}E_{1}+W_{2}E_{2}+W_{3}E_{3}$, we first find that 
\begin{equation}
\label{iomegamu}
\iota_{\omega}\mu=W_{1}E_{2}^{*}\wedge E_{3}^{*}-W_{2}E_{1}^{*}\wedge E_{3}^{*}+W_{3}E_{1}^{*}\wedge E_{2}^{*}.
\end{equation}
Finally, using the definition \eqref{vort3d}, we find from \eqref{uflatfinal} and \eqref{iomegamu} that
\begin{equation}
\label{omegaS3}
\omega=(E_{3}^{2}+E_{2}^{2})(\tilde\psi)E_{1}+E_{3}(\tilde\psi+\tilde\sigma)E_{2}-E_{2}(\tilde\psi+\tilde\sigma)E_{3}.
\end{equation}
Analogously, we get for the magnetic vorticity (current density)
\begin{equation}
\label{jS3}
j=(E_{3}^{2}+E_{2}^{2})(\tilde\xi)E_{1}+E_{3}(\tilde\xi+\tilde\rho)E_{2}-E_{2}(\tilde\xi+\tilde\rho)E_{3}.
\end{equation}

Further, inserting the expressions \eqref{ubS3}, \eqref{omegaS3}, and \eqref{jS3} into \eqref{vectorfieldvortform}, and using the identities \eqref{LieAlgStr}, we arrive (see the details in Appendix \ref{appendix1}) at the following system of equations for the fields $(\tilde\psi,\tilde\sigma,\tilde\rho,\tilde\xi)$:
\begin{equation}
\label{MHDtilde}
\left\{
\begin{aligned}
&\frac{\partial}{\partial t}(E_{3}^{2}+E_{2}^{2})(\tilde\psi)=\mathfrak{B}((E_{3}^{2}+E_{2}^{2})(\tilde\psi),\tilde\psi)+\mathfrak{B}(\tilde\xi,(E_{3}^{2}+E_{2}^{2})(\tilde\xi))+2\mathfrak{B}(\tilde q,\tilde\psi)+2\mathfrak{B}(\tilde \xi,\tilde\rho), \\
&\frac{\partial\tilde\rho}{\partial t}=\mathfrak{B}(\tilde \rho,\tilde\psi)+\mathfrak{B}(\tilde \xi,\tilde q)+2\mathfrak{B}(\tilde \psi,\tilde\xi),\\
&\frac{\partial\tilde q}{\partial t}=\mathfrak{B}(\tilde q,\tilde\psi)+\mathfrak{B}(\tilde\xi,\tilde\rho),\\
&\frac{\partial\tilde\xi}{\partial t}=\mathfrak{B}(\tilde\xi,\tilde\psi),
\end{aligned}
\right.
\end{equation}
where $\tilde q=\tilde\psi+\tilde\sigma$ is a new field introduced for convenience, and $\mathfrak{B}(\cdot,\cdot)$ is a bilinear, skew-symmetric form 
\begin{equation*}
\mathfrak{B}(\tilde f,\tilde h)=E_{2}(\tilde f)E_{3}(\tilde h)-E_{3}(\tilde f)E_{2}(\tilde h),
\end{equation*}
for arbitrary $E_{1}$-invariant functions $\tilde f$ and $\tilde h$.

Let now $\pi\colon \mathbb{S}^{3}\to \mathbb{S}^{2}\simeq \mathbb{S}^{3}/\mathbb{S}^{1}$ be the Hopf map, and the rotation action $\mathbb{S}^{1}$ be generated by the vector field $E_{1}$. This implies that $E_{1}$-invariant functions $\tilde f$ can be identified with the functions $f$ on $\mathbb{S}^{2}$ via a pullback $\pi^{*}(f)=\tilde f$. By direct computations, one can also verify that the bilinear form $\mathfrak{B}$ descends to the Poisson bracket on $\mathbb{S}^{2}$, i.e. $\pi^{*}(\left\{f,h\right\})=\mathfrak{B}(\tilde f,\tilde h)$, and that $\pi^{*}(\Delta\psi)=(E_{3}^{2}+E_{2}^{2})(\tilde\psi)$, i.e. the operator $E_{3}^{2}+E_{2}^{2}$ reduces to the Laplace operator $\Delta$ on $\mathbb{S}^{2}$ (for details, we refer to \cite{LichtMisPrest2022}). Since all the ingredients in \eqref{MHDtilde} are $E_{1}$-invariant, then \eqref{MHDtilde} descends to a system of equations on $\mathbb{S}^{2}$, the symmetry reduced system:
\begin{equation}
\label{LP3D}
\left\{
\begin{aligned}
&\dot\omega=\left\{\omega,\psi\right\}+\left\{\xi,\Delta\xi\right\}+2\left\{q,\psi\right\}+2\left\{\xi,\rho\right\}, \\
&\dot\rho=\left\{\rho,\psi\right\}+\left\{\xi,q\right\}+2\left\{\psi,\xi\right\},\\
&\dot q=\left\{q,\psi\right\}+\left\{\xi,\rho\right\},\\
&\dot\xi=\left\{\xi,\psi\right\},
\end{aligned}
\right.
\end{equation}
where $q,\psi,\omega=\Delta\psi,\rho,\xi\in C^{\infty}(\mathbb{S}^{2})$, and $\left\{\cdot,\cdot\right\}$ is the Poisson bracket on $\mathbb{S}^{2}$. 
\subsection{Hamiltonian formulation of reduced equations} Here, we develop the Euler--Arnold framework for equations \eqref{LP3D}. We will show that system \eqref{LP3D} is a Lie--Poisson flow on the dual of the Lie algebra that is the semidirect product extension of the semidirect product extension of the Lie algebra of divergence-free vector fields by the algebra of smooth functions on $\mathbb{S}^{2}$.

First, we construct the semidirect product extension $\mathfrak{f}=\mathfrak{X}_{\mu}(\mathbb{S}^{2})\ltimes C^{\infty}(\mathbb{S}^{2})$ of the algebra of divergence free vector fields $\mathfrak{X}_{\mu}(\mathbb{S}^{2})$ by the algebra of smooth functions $C^{\infty}(\mathbb{S}^{2})$. We identify the space $\mathfrak{f}$ as a set of pairs $(v,\sigma)$, with $v\in\mathfrak{X}_{\mu}(\mathbb{S}^{2})$ and $\sigma\in C^{\infty}(\mathbb{S}^{2})$, equipped with the Lie algebra structure
\begin{equation}
\label{LieAlgStructureSemidirExt}
[(v_{1},\sigma_{1}),(v_{2},\sigma_{2})]=([v_{1},v_{2}],\kappa(v_{1})\sigma_{2}-\kappa(v_{2})\sigma_{1}),
\end{equation}
where $[v_{1},v_{2}]$ is the commutator of vector fields, and $\kappa\colon\mathfrak{X}_{\mu}(\mathbb{S}^{2})\to\mathrm{End}(C^{\infty}(\mathbb{S}^{2}))$ is the action of the Lie algebra $\mathfrak{X}_{\mu}(\mathbb{S}^{2})$ on $C^{\infty}(\mathbb{S}^{2})$:
\begin{equation*}
\kappa(v)\sigma=\left\{\psi,\sigma\right\},\quad v=X_{\psi}.
\end{equation*}
With this choice of $\kappa$, the Lie algebra structure \eqref{LieAlgStructureSemidirExt} on $\mathfrak{f}$ becomes
\begin{equation}
\label{liebracketonpairs}
[(\psi_{1},\sigma_{1}),(\psi_{2},\sigma_{2})]=\mathrm{ad}_{(\psi_{1},\sigma_{1})}(\psi_{2},\sigma_{2})=\left(\left\{\psi_{1},\psi_{2}\right\},\left\{\psi_{1},\sigma_{2}\right\}-\left\{\psi_{2},\sigma_{1}\right\}\right).
\end{equation}

We identify the dual $\mathfrak{f}^{*}$ of $\mathfrak{f}$ with the space $\mathfrak{f}$ itself via the inner product coming from the $L_{2}$ inner product of divergence-free, $E_{1}$-invariant vector fields, which we denote as $\mathfrak{X}_{\mu}^{1}(\mathbb{S}^{3})$. Namely, let $u,v\in\mathfrak{X}_{\mu}^{1}(\mathbb{S}^{3})$. Then
\begin{equation*}
\begin{split}
u&=\tilde\sigma E_{1}-(E_{3}\tilde\psi)E_{2}+(E_{2}\tilde\psi)E_{3},\quad E_{1}\tilde\sigma=E_{1}\tilde\psi=0,\\
v&=\tilde\rho E_{1}-(E_{3}\tilde\xi)E_{2}+(E_{2}\tilde\xi)E_{3}, \quad E_{1}\tilde\rho=E_{1}\tilde\xi=0,
\end{split}
\end{equation*}
and therefore
\begin{equation}
\label{integralsonS3}
\begin{split}
\langle u,v\rangle_{L^{2}}&=\int_{\mathbb{S}^{3}}\left(\tilde\sigma\tilde\rho+(E_{3}\tilde\psi)(E_{3}\tilde\xi)+(E_{2}\tilde\psi)(E_{2}\tilde\xi)\right)\mu_{\mathbb{S}^{3}}{}\\&=4\pi\int_{\mathbb{S}^{2}}(\sigma\rho+\nabla\psi\cdot\nabla\xi)\mu_{\mathbb{S}^{2}}=4\pi\int_{\mathbb{S}^{2}}(\sigma\rho-(\Delta\psi)\xi)\mu_{\mathbb{S}^{2}}\overset{\text{def}}{=}\langle(\psi,\sigma),(\xi,\rho)\rangle,
\end{split}
\end{equation}
where $(\psi,\sigma)\in\mathfrak{f}^{*}$, $(\xi,\rho)\in\mathfrak{f}$, and where we used \cite[Eq. (20)-(21)]{ModinPreston2025}.

The following result gives the expression for the coadjoint operator on the dual $\mathfrak{f}^{*}\simeq\mathfrak{f}$ of the Lie algebra $\mathfrak{f}=\mathfrak{X}_{\mu}(\mathbb{S}^{2})\ltimes C^{\infty}(\mathbb{S}^{2})$.
\begin{lemma}
The coadjoint operator on $\mathfrak{f}^{*}$ has the following form:
\begin{equation}
\label{finalformulaforadstar}
\mathrm{ad}^{*}_{(\psi,\sigma)}(\xi,\rho)=\left(\Delta^{-1}\left(\left\{\Delta\xi,\psi\right\}+\left\{\sigma,\rho\right\}\right),\left\{\rho,\psi\right\}\right).
\end{equation}
\end{lemma}
\begin{proof}Let $(\xi,\rho)\in\mathfrak{f}^{*}$, $(\psi,\sigma),(a,b)\in\mathfrak{f}$. Then, using the definition of the coadjoint operator, we get
\begin{equation*}
\begin{split}
\left\langle\mathrm{ad}^{*}_{(\psi,\sigma)}(\xi,\rho),(a,b)\right\rangle&=\langle(\xi,\rho),\mathrm{ad}_{(\psi,\sigma)}(a,b)\rangle=\langle(\xi,\rho),(\left\{\psi,a\right\},\left\{\psi,b\right\}-\left\{a,\sigma\right\})\rangle{}\\&=4\pi\int_{\mathbb{S}^{2}}\rho(\left\{\psi,b\right\}-\left\{a,\sigma\right\})-4\pi\int_{\mathbb{S}^{2}}(\Delta\xi)\left\{\psi,a\right\}{}\\&=4\pi\int_{\mathbb{S}^{2}}\left\{\rho,\psi\right\}b+4\pi\int_{\mathbb{S}^{2}}(\left\{\rho,\sigma\right\}-\left\{\Delta\xi,\psi\right\})a{}\\&=
\left\langle\left(\Delta^{-1}\left(\left\{\Delta\xi,\psi\right\}+\left\{\sigma,\rho\right\}\right),\left\{\rho,\psi\right\}\right),(a,b)\right\rangle,
\end{split}
\end{equation*}
where we used \eqref{integralsonS3}, and which implies the assertion.
\end{proof}

The next step is to construct the magnetic extension of the semidirect product Lie algebra $\mathfrak{f}=\mathfrak{X}_{\mu}(\mathbb{S}^{2})\ltimes C^{\infty}(\mathbb{S}^{2})$, which we will denote by
\begin{equation}
\label{magneticextensionformhd}
\mathfrak{F}=(\mathfrak{X}_{\mu}(\mathbb{S}^{2})\ltimes C^{\infty}(\mathbb{S}^{2}))\ltimes(\mathfrak{X}_{\mu}(\mathbb{S}^{2})\ltimes C^{\infty}(\mathbb{S}^{2}))^{*}=\mathfrak{f}\ltimes\mathfrak{f}^{*}.
\end{equation}

To get the Euler--Arnold formulation for \eqref{LP3D}, we need to derive the coadjoint operator on the dual $\mathfrak{F}^{*}\simeq\mathfrak{F}$. Let $((\alpha,\beta),(\gamma,\delta))\in\mathfrak{F}$, and let $((\xi,\rho),(\psi,q))\in\mathfrak{F}^{*}$. Then, using formula \eqref{coadjopsemidir} for the coadjoint operator on a semidirect product, we get
\begin{equation}
\label{coadjforsemidirabelanext}
\mathrm{ad}^{*}_{((\alpha,\beta),(\gamma,\delta))}((\xi,\rho),(\psi,q))=\left([(\xi,\rho),(\alpha,\beta)],\mathrm{ad}^{*}_{(\alpha,\beta)}(\psi,q)-\mathrm{ad}^{*}_{(\xi,\rho)}(\gamma,\delta)\right).
\end{equation}
We consider the two components in \eqref{coadjforsemidirabelanext} separately. First, using \eqref{liebracketonpairs}, we get
\begin{equation*}
[(\xi,\rho),(\alpha,\beta)]=\left(\left\{\xi,\alpha\right\},\left\{\xi,\beta\right\}-\left\{\alpha,\rho\right\}\right),
\end{equation*}
and therefore the pair $(\xi,\rho)$ is evolving, according to the Euler--Arnold equations, as
\begin{equation}
\label{sigmapsieulerarnold}
\left\{
\begin{aligned}
&\dot\rho=\left\{\rho,\alpha\right\}+\left\{\xi,\beta\right\}, \\
&\dot\xi=\left\{\xi,\alpha\right\}.
\end{aligned}
\right.
\end{equation}

The Euler--Arnold equation for the pair $(\psi,q)$ is
\begin{equation*}
\frac{\mathrm{d}}{\mathrm{d}t}(\psi,q)=\mathrm{ad}^{*}_{(\alpha,\beta)}(\psi,q)-\mathrm{ad}^{*}_{(\xi,\rho)}(\gamma,\delta),
\end{equation*}
where the $\mathrm{ad}^{*}$ operator is defined by formula \eqref{finalformulaforadstar}. We finally get
\begin{equation}
\label{rhoxieulerarnold}
\left\{
\begin{aligned}
\dot\omega&=\left\{\omega,\alpha\right\}+\left\{\beta,q\right\}+\left\{\xi,\Delta\gamma\right\}+\left\{\delta,\rho\right\}, \\
\dot q&=\left\{q,\alpha\right\}+\left\{\xi,\delta\right\}.
\end{aligned}
\right.
\end{equation}
Altogether, equations \eqref{sigmapsieulerarnold} and \eqref{rhoxieulerarnold} constitute a general form of the Euler--Arnold equation on $\mathfrak{F}^{*}$, with $\mathfrak{F}=(\mathfrak{X}_{\mu}(\mathbb{S}^{2})\ltimes C^{\infty}(\mathbb{S}^{2}))\ltimes(\mathfrak{X}_{\mu}(\mathbb{S}^{2})\ltimes C^{\infty}(\mathbb{S}^{2}))^{*}$ provided an appropriate inertia operator $I\colon\mathfrak{F}\to\mathfrak{F}^{*}$ relating $((\alpha,\beta),(\gamma,\delta))$ and $((\xi,\rho),(\psi,q))$ is given. In essence, inertia operator defines the algebra fields $((\alpha,\beta),(\gamma,\delta))$ in terms of the co-algebra fields $((\xi,\rho),(\psi,q))$, and to make a connection between the general equations \eqref{sigmapsieulerarnold}-\eqref{rhoxieulerarnold} and \eqref{LP3D}, we will put
\begin{equation}
\label{choiceofthehamiltonian}
\beta=q-2\psi,\quad\delta=\rho+2\xi,\quad\alpha=\psi,\quad\gamma=\xi.
\end{equation}
With this choice, equations \eqref{sigmapsieulerarnold}-\eqref{rhoxieulerarnold} coincide with \eqref{LP3D}. Summarizing, we arrive at the first central result of the paper.
\begin{theorem}
The system of axisymmetric MHD equations on $\mathbb{S}^{3}$ defined in \eqref{LP3D} is a Lie--Poisson system
\begin{equation*}
\frac{\mathrm{d}}{\mathrm{d}t}((\xi,\rho),(\psi,q))=\mathrm{ad}^{*}_{((\alpha,\beta),(\gamma,\delta))}((\xi,\rho),(\psi,q))
\end{equation*}
on the dual $\mathfrak{F}^{*}$ of the Lie algebra
\begin{equation*}
\mathfrak{F}=(\mathfrak{X}_{\mu}(\mathbb{S}^{2})\ltimes C^{\infty}(\mathbb{S}^{2}))\ltimes(\mathfrak{X}_{\mu}(\mathbb{S}^{2})\ltimes C^{\infty}(\mathbb{S}^{2}))^{*}.
\end{equation*}
\end{theorem}
\subsection{Conservation laws for reduced equations}
Even though system \eqref{LP3D} does not reduce to a number of subsystems evolving on a semidirect product Lie algebra, like reduced equations \eqref{MHDvort}, which, as shown in the coming sections, typically happens in planar MHD, one can notice that some pairs of the fields in \eqref{LP3D} do form the ``semidirect product pairs'', in analogy with \eqref{MHDvort}. Therefore, one can guess that the following quantities are Casimir invariants for \eqref{LP3D}:
\begin{equation}
\label{cas3Dnotfull}
\mathcal{C}_{f}=\int_{\mathbb{S}^{2}}f(\xi)\mu,\quad J_{h}=\int_{\mathbb{S}^{2}}\rho h(\xi)\mu,\quad K_{\textsl{g}}=\int_{\mathbb{S}^{2}}q\textsl{g}(\xi)\mu,
\end{equation}
for arbitrary smooth functions $f,\textsl{g},h$. The fact that the field $\xi$ comes into the Casimirs \eqref{cas3Dnotfull} as an argument of arbitrary functions reflects the transport property of $\xi$, which, in turn, follows from the transport of the magnetic field $B$. There is, however, one more Casimir invariant:
\begin{equation}
\label{crosshelicityreduced}
\mathcal{I}_{k}=\int_{\mathbb{S}^{2}}(\omega k(\xi)-\rho q k^{\prime}(\xi))\mu,
\end{equation}
where $k$ is an arbitrary smooth function.

Given the formula for the $L_{2}$-pairing \eqref{integralsonS3}, one immediately conludes that \eqref{crosshelicityreduced}, in case $k(\xi)=\xi$, is the descendant of the three-dimensional cross-helicity Casimir $\mathcal{I}(u,B)$ in \eqref{casMHD}. Finally, the energy \eqref{energymetric} reduces to
\begin{equation}
\label{energyreduced3D}
H=-\frac{1}{2}\int_{\mathbb{S}^{2}}(\omega\psi-q^{2})\mu-\frac{1}{2}\int_{\mathbb{S}^{2}}(\xi\Delta\xi-\rho^{2})\mu,
\end{equation}
which is also conserved, but is not a Casimir. We arrive at the following result.
\begin{proposition}
The quantities \eqref{cas3Dnotfull}-\eqref{crosshelicityreduced} are Casimir invariants for the system of axisymmetric MHD equations on $\mathbb{S}^{3}$ defined in \eqref{LP3D}. The energy \eqref{energyreduced3D} is also conserved by the flow \eqref{LP3D}.
\end{proposition}
\begin{proof}
We will prove this result for the Casimir $\mathcal{I}_{k}$, and the Hamiltonian (others are treated analogously). To show that $\mathcal{I}_{k}$ is indeed a Casimir, we will calculate its evolution along trajectories of the general Euler--Arnold equation \eqref{sigmapsieulerarnold}-\eqref{rhoxieulerarnold}, where the choice of the inertia operator \eqref{choiceofthehamiltonian} is not yet made:
\begin{equation*}
\begin{split}
\frac{\mathrm{d}\mathcal{I}_{k}}{\mathrm{d}t}&=\int_{\mathbb{S}^{2}}\dot\omega k(\xi)+\int_{\mathbb{S}^{2}}\omega\frac{\mathrm{d}k(\xi)}{\mathrm{d}t}-\int_{\mathbb{S}^{2}}\dot\rho qk^{\prime}(\xi)-\int_{\mathbb{S}^{2}}\rho\dot qk^{\prime}(\xi)-\int_{\mathbb{S}^{2}}\rho q\frac{\mathrm{d}k^{\prime}(\xi)}{\mathrm{d}t}{}\\&=\underbrace{\int_{\mathbb{S}^{2}}\left\{\omega,\alpha\right\}k(\xi)+\int_{\mathbb{S}^{2}}\omega\frac{\mathrm{d}k(\xi)}{\mathrm{d}t}}_{=0}+\int_{\mathbb{S}^{2}}\left\{\beta,q\right\}k(\xi)+\int_{\mathbb{S}^{2}}\left\{\xi,\Delta\gamma\right\}k(\xi)+\int_{\mathbb{S}^{2}}\left\{\delta,\rho\right\}k(\xi){}\\&+\int_{\mathbb{S}^{2}}qk^{\prime}(\xi)\left(\left\{\alpha,\rho\right\}+\left\{\beta,\xi\right\}\right)+\int_{\mathbb{S}^{2}}\rho k^{\prime}(\xi)\left(\left\{\alpha,q\right\}+\left\{\delta,\xi\right\}\right)-\int_{\mathbb{S}^{2}}\rho q\frac{\mathrm{d}k^{\prime}(\xi)}{\mathrm{d}t}{}\\&=\underbrace{\int_{\mathbb{S}^{2}}\left(\left\{\beta,q\right\}k(\xi)+qk^{\prime}(\xi)\left\{\beta,\xi\right\}\right)}_{=0}+\underbrace{\int_{\mathbb{S}^{2}}\left(\left\{\delta,\rho\right\}k(\xi)+\rho k^{\prime}(\xi)\left\{\delta,\xi\right\}\right)}_{=0}{}\\&
+\int_{\mathbb{S}^{2}}\left(k^{\prime}(\xi)(q\left\{\alpha,\rho\right\}+\left\{\alpha,q\right\}\rho)\right)-\int_{\mathbb{S}^{2}}\rho q\frac{\mathrm{d}k^{\prime}(\xi)}{\mathrm{d}t},
\end{split}
\end{equation*}
where we used that $\left\{\xi,k(\xi)\right\}=0$, $k^{\prime}(\xi)\left\{q,\xi\right\}=\left\{q,k(\xi)\right\}$, and that if $\xi$ is transported according to the last equation in \eqref{sigmapsieulerarnold}, then $k(\xi)$ is as well. Consider the last integral in detail:
\begin{equation*}
\int_{\mathbb{S}^{2}}\rho q\frac{\mathrm{d}k^{\prime}(\xi)}{\mathrm{d}t}=\int_{\mathbb{S}^{2}}\rho q\left\{k^{\prime}(\xi),\alpha\right\}=\int_{\mathbb{S}^{2}}\left\{\alpha,\rho q\right\}k^{\prime}(\xi)=\int_{\mathbb{S}^{2}}k^{\prime}(\xi)\left(q\left\{\alpha,\rho\right\}+\rho\left\{\alpha,q\right\}\right),
\end{equation*}
which implies that $\dot{\mathcal{I}}_{k}=0$ along trajectories of \eqref{sigmapsieulerarnold}-\eqref{rhoxieulerarnold}, i.e. $\mathcal{I}_{k}$ is a Casimir.

Further, we calculate the evolution of the Hamiltonian \eqref{energyreduced3D} along trajectories of \eqref{LP3D}:
\begin{equation*}
\begin{split}
\frac{\mathrm{d}H}{\mathrm{d}t}&=-\frac{1}{2}\int_{\mathbb{S}^{2}}\dot\omega\psi-\frac{1}{2}\int_{\mathbb{S}^{2}}(\Delta\psi)\dot\psi+\int_{\mathbb{S}^{2}}q\dot q-\frac{1}{2}\int_{\mathbb{S}^{2}}\dot\xi\Delta\xi-\frac{1}{2}\int_{\mathbb{S}^{2}}\xi\Delta\dot\xi+\int_{\mathbb{S}^{2}}\rho\dot\rho{}\\&
=-\int_{\mathbb{S}^{2}}(\dot\omega\psi+\dot\xi\Delta\xi)+\int_{\mathbb{S}^{2}}(q\dot q+\rho\dot\rho){}\\&=-\int_{\mathbb{S}^{2}}\left(\psi(\left\{\omega,\psi\right\}+\left\{\xi,\Delta\xi\right\}+2\left\{q,\psi\right\}+2\left\{\xi,\rho\right\})+\Delta\xi\left\{\xi,\psi\right\}\right){}\\&+\int_{\mathbb{S}^{2}}(q(\left\{q,\psi\right\}+\left\{\xi,\rho\right\})+\rho(\left\{\rho,\psi\right\}+\left\{\xi,q\right\}+2\left\{\psi,\xi\right\}))=0,
\end{split}
\end{equation*}
where one repeatedly uses $\int\left\{f,g\right\}h=\int\left\{g,h\right\}f=\int \left\{h,f\right\}g$.
\end{proof}
\begin{remark}
\begin{enumerate}
\item One observes that if the magnetic field $B$ is trivial, which implies $\rho=\xi=0$, then system \eqref{LP3D} reduces to the axisymmetric Euler equations on $\mathbb{S}^{3}$ derived in \cite{ModinPreston2025}. Furthermore, the second component in the semidirect product Lie algebra \eqref{magneticextensionformhd} vanishes, and one gets the semidirect product Lie algebra $\mathfrak{f}=\mathfrak{X}_{\mu}(\mathbb{S}^{2})\ltimes C^{\infty}(\mathbb{S}^{2})$. In \cite{ModinPreston2025}, the reduced version of the axisymmetric Euler equations on $\mathbb{S}^{3}$ is modeled on the Abelian extension Lie algebra instead, which is a more general construction compared to the semidirect product extension. Here, we have shown that the semidirect product structure is enough to derive the Hamiltonian formulation. 
\item It would be beneficial to get an extension of the semidirect product bracket \eqref{RMHDbracket} to the case of semidirect products of semidirect product  extensions, which will be addressed in future works.
\end{enumerate}
\end{remark}
\subsection{Axisymmetric reduced equations in $\mathbb{R}^{3}$}
One can obtain the reduced version of the 3D MHD equations in planar geometry as well. Consider the system \eqref{MHD} in $\mathbb{R}^{3}(x,y,z)$ with solutions invariant with respect to vertical translations (in $z$ direction). In this case, the fields $u$ and $B$ take the form (see \cite{GTAM2017})
\begin{equation}
\label{uBaxisymmT2}
\begin{split}
u&=X_{\psi}+\sigma(t,x,y)\partial_{z}=-\psi_{y}\partial_{x}+\psi_{x}\partial_{y}+\sigma\partial_{z},\\
B&=X_{-\xi}+\rho(t,x,y)\partial_{z}=\xi_{y}\partial_{x}-\xi_{x}\partial_{y}+\rho\partial_{z}.
\end{split}
\end{equation}
Inserting the ansatz \eqref{uBaxisymmT2} into \eqref{MHD}, one gets the following system of equations for the fields $\omega=\Delta\psi,\rho,\sigma,\xi$:
\begin{equation}
\label{MHDaxisymmvortexT2}
\left\{
\begin{aligned}
&\dot\omega=\left\{\omega,\psi\right\}+\left\{\xi,\Delta\xi\right\}, \\
&\dot \rho=\left\{\rho,\psi\right\}+\left\{\sigma,\xi\right\},\\
&\dot \sigma=\left\{\sigma,\psi\right\}+\left\{\rho,\xi\right\},\\
&\dot\xi=\left\{\xi,\psi\right\}.
\end{aligned}
\right.
\end{equation}
We will assume doubly periodic boundary conditions, i.e. system \eqref{MHDaxisymmvortexT2} evolves on the flat torus $\mathbb{T}^{2}$.

Equations \eqref{MHDaxisymmvortexT2}, as well as their extensions, such as Hall MHD and L\"ust MHD, were derived and studied in \cite{GTAM2017} (see also \cite[Ch. 31.1]{DarrylBook}). We observe that system \eqref{MHDaxisymmvortexT2} splits into subsystems of type \eqref{MHDvort}, and the extended version of the Lie--Poisson bracket \eqref{RMHDbracket} for \eqref{MHDaxisymmvortexT2} takes the form
\begin{equation}
\label{RMHDbracketAxi}
\begin{split}
 \llbracket F,G\rrbracket_{\text{ax}}=\int_{\mathbb{T}^{2}}&\left[
 \xi\left(\left\{\frac{\delta F}{\delta \omega},\frac{\delta G}{\delta \xi}\right\}+\left\{\frac{\delta F}{\delta \xi},\frac{\delta G}{\delta \omega}\right\}+\left\{\frac{\delta F}{\delta \sigma},\frac{\delta G}{\delta \rho}\right\}+\left\{\frac{\delta F}{\delta \rho},\frac{\delta G}{\delta \sigma}\right\}\right)+\omega\left\{\frac{\delta F}{\delta \omega},\frac{\delta G}{\delta \omega}\right\}\right.{}\\&\left.+\rho\left(\left\{\frac{\delta F}{\delta \omega},\frac{\delta G}{\delta \rho}\right\}+\left\{\frac{\delta F}{\delta \rho},\frac{\delta G}{\delta \omega}\right\}\right)+\sigma\left(\left\{\frac{\delta F}{\delta \omega},\frac{\delta G}{\delta \sigma}\right\}+\left\{\frac{\delta F}{\delta \sigma},\frac{\delta G}{\delta \omega}\right\}\right)\right]\mathrm{d}x\mathrm{d}y.
 \end{split}
\end{equation}
Therefore, system \eqref{MHDaxisymmvortexT2} takes the Hamiltonian form $\dot F= \llbracket F,H\rrbracket_{\text{ax}}$, with the Hamltonian
\begin{equation*}
H=\frac{1}{2}\int_{\mathbb{T}^{2}}(|\nabla\psi|^{2}+|\nabla\xi|^{2}+\rho^{2}+\sigma^{2})\mathrm{d}x\mathrm{d}y,
\end{equation*}
and possesses the following set of Casimirs invariants, see \cite{GTAM2017}:
\begin{equation*}
\mathcal{C}_{f}=\int_{\mathbb{T}^{2}}f(\xi)\mathrm{d}x\mathrm{d}y,\quad J_{h}=\int_{\mathbb{T}^{2}}\rho h(\xi)\mathrm{d}x\mathrm{d}y,\quad K_{\textsl{g}}=\int_{\mathbb{T}^{2}}\sigma\textsl{g}(\xi)\mathrm{d}x\mathrm{d}y,\quad \mathcal{I}_{k}=\int_{\mathbb{T}^{2}}\omega k(\xi)\mathrm{d}x\mathrm{d}y,
\end{equation*}
for arbitrary smooth functions $f,g,h,k$.
\begin{remark}
We observe that the underlying Lie algebra for system \eqref{MHDaxisymmvortexT2} is a direct sum of semidirect product extensions. It becomes evident when writing equations \eqref{MHDaxisymmvortexT2} in the form
\begin{equation*}
\left\{
\begin{aligned}
&\dot\omega=\left\{\omega,\psi\right\}+\left\{\xi,\Delta\xi\right\}, \\
&\dot\xi=\left\{\xi,\psi\right\},
\end{aligned}
\right.
\quad
\left\{
\begin{aligned}
&\dot \rho=\left\{\rho,\psi\right\}+\left\{\xi,-\sigma\right\},\\
&\dot\xi=\left\{\xi,\psi\right\},
\end{aligned}
\right.
\quad
\left\{
\begin{aligned}
&\dot \sigma=\left\{\sigma,\psi\right\}+\left\{\xi,-\rho\right\},\\
&\dot\xi=\left\{\xi,\psi\right\},
\end{aligned}
\right.
\end{equation*}
which yields three semidirect product subsystems of the type \eqref{MHDvort}. We conclude that the Lie algebra $\mathfrak{F}$ is
\begin{equation}
\label{LieAlgStructureforT2}
\mathfrak{F}=(\mathfrak{X}_{\mu}(\mathbb{T}^{2})\ltimes C^{\infty}(\mathbb{T}^{2}))\oplus(\mathfrak{X}_{\mu}(\mathbb{T}^{2})\ltimes C^{\infty}(\mathbb{T}^{2}))\oplus(\mathfrak{X}_{\mu}(\mathbb{T}^{2})\ltimes C^{\infty}(\mathbb{T}^{2})).
\end{equation}
Hence, in the case of non-Euclidean geometry, such as axisymmetric flows on $\mathbb{S}^{3}$, the underlying Lie algebra has a more complicated structure (semidirect products) compared to the Euclidean case, such as translationally invariant flows in $\mathbb{R}^{3}$ (direct sum). The same effect is known to hold for the Euler equations. Indeed, in the $\partial_{z}$-invariant Euler equations, the horizontal dynamics is completely decoupled from the vertical dynamics, whereas in the axisymmetric Euler equations on $\mathbb{S}^{3}$ the coupling occurs via the semidirect product action.
\end{remark}
\section{Matrix discretizations for axisymmetric MHD}
\label{sec4}
In this section, we review the key concepts underlying Zeitlin's matrix equations of hydrodynamics and magnetohydrodynamics and extend it to the 3D axisymmetric equations \eqref{LP3D}. As of today, matrix approximations developed by V. Zeitlin in the works \cite{Ze1991,Zeit,Ze2005} remain the only consistent structure-preserving way of discretization of the Euler and MHD equations on the flat torus and the two-sphere. Namely, Zeitlin's equations are themselves Lie--Poisson equations for the truncated and modified finite-dimensional Lie--Poisson structure that converges to the infinite-dimensional one.
\subsection{Zeitlin's model for 2D MHD on the torus $\mathbb{T}^{2}$}
\label{Zeitlin2dMHDtorus}
We start with the overview of Zeitlin's model, which historically appeared first on the flat torus $\mathbb{T}^{2}$ in \cite{Ze1991}. 

The concept of obtaining a consistent truncation for \eqref{MHDvort} begins with an attempt to write equations \eqref{MHDvort} in the Fourier space. Let $\mathbf{x}\in\mathbb{T}^{2}$, let $\mathbf{m}=(m_{1},m_{2})\in\mathbb{Z}^{2}$ be the lattice, and $\mathbb{Z}_{0}^{2}=\mathbb{Z}^{2}\setminus(0,0)$. Then, decomposing the fields $\omega$ and $\psi$ in the basis $\exp(\mathrm{i}(\mathbf{m}\cdot\mathbf{x}))$ as
\begin{equation}
\begin{aligned}
\label{fourierontorus}
\omega(t,\mathbf{x})&=\sum\limits_{\mathbf{m}\in\mathbb{Z}^{2}_{0}}\omega_{\mathbf{m}}(t)\exp(\mathrm{i}(\mathbf{m}\cdot\mathbf{x})),\\\psi(t,\mathbf{x})&=-\sum\limits_{\mathbf{m}\in\mathbb{Z}^{2}_{0}}\omega_{\mathbf{m}}(t)|\mathbf{m}|^{-2}\exp(\mathrm{i}(\mathbf{m}\cdot\mathbf{x})),\\
\end{aligned}
\quad
\begin{aligned}
\theta(t,\mathbf{x})&=\sum\limits_{\mathbf{m}\in\mathbb{Z}^{2}_{0}}\theta_{\mathbf{m}}(t)\exp(\mathrm{i}(\mathbf{m}\cdot\mathbf{x})),\\j(t,\mathbf{x})&=-\sum\limits_{\mathbf{m}\in\mathbb{Z}^{2}_{0}}\theta_{\mathbf{m}}(t)|\mathbf{m}|^{2}\exp(\mathrm{i}(\mathbf{m}\cdot\mathbf{x})),
\end{aligned}
\end{equation}
and inserting them into \eqref{MHDvort}, we get the following dynamical system for the Fourier coefficients $\omega_{\mathbf{m}}(t)$ and $\theta_{\mathbf{m}}(t)$:
\begin{equation}
\label{MHDEqsFourier}
\left\{
\begin{aligned}
\dot\omega_{\mathbf{m}}&=\sum\limits_{\mathbf{k}\in\mathbb{Z}^{2}_{0}}|\mathbf{k}|^{-2}(\mathbf{k}\times\mathbf{m})\omega_{\mathbf{k}+\mathbf{m}}\omega_{-\mathbf{k}}+\sum\limits_{\mathbf{k}\in\mathbb{Z}^{2}_{0}}|\mathbf{k}|^{2}(\mathbf{k}\times\mathbf{m})\theta_{\mathbf{k}+\mathbf{m}}\theta_{-\mathbf{k}},\\
\dot\theta_{\mathbf{m}}&=\sum\limits_{\mathbf{k}\in\mathbb{Z}^{2}_{0}}|\mathbf{k}|^{-2}(\mathbf{k}\times\mathbf{m})\theta_{\mathbf{k}+\mathbf{m}}\omega_{-\mathbf{k}},
\end{aligned}
\right.
\end{equation}
where $\mathbf{k}\times\mathbf{m}=k_{1}m_{2}-k_{2}m_{1}$. Summation in \eqref{fourierontorus} and \eqref{MHDEqsFourier} runs over $\mathbb{Z}^{2}_{0}$, rather than $\mathbb{Z}^{2}$, as due to Stokes' theorem $\int_{\mathbb{T}^{2}}\omega=0$, which implies $\omega_{00}=0$. Furthermore, without loss of generality, one can assume $\int_{\mathbb{T}^{2}}\theta=0$, and therefore $\theta_{00}=0$.

Equations \eqref{MHDEqsFourier} are Hamiltonian with respect to the Fourier transformed Lie--Poisson bracket \eqref{RMHDbracket}:
\begin{equation}
\label{LiePoissonFourier}
\llbracket F,G\rrbracket=\frac{1}{4\pi^{2}}\sum_{\mathbf{k},\mathbf{m}\in\mathbb{Z}^{2}}\omega_{\mathbf{k}+\mathbf{m}}(\mathbf{k}\times\mathbf{m})\frac{\partial F}{\partial\omega_{\mathbf{m}}}\frac{\partial G}{\partial\omega_{\mathbf{k}}}+\frac{1}{4\pi^{2}}\sum_{\mathbf{k},\mathbf{m}\in\mathbb{Z}^{2}}\theta_{\mathbf{k}+\mathbf{m}}(\mathbf{k}\times\mathbf{m})\left(\frac{\partial F}{\partial\theta_{\mathbf{m}}}\frac{\partial G}{\partial\omega_{\mathbf{k}}}+\frac{\partial F}{\partial\omega_{\mathbf{m}}}\frac{\partial G}{\partial\theta_{\mathbf{k}}}\right),
\end{equation}
for $F=F(\omega_{\mathbf{m}})$, $G=G(\omega_{\mathbf{m}})$, $\mathbf{m}\in\mathbb{Z}^{2}$, and with the Hamiltonian being
\begin{equation}
\label{HamiltonianFourier}
H=\frac{1}{2}\sum\limits_{\mathbf{k}\in\mathbb{Z}^{2}_{0}}\mathbf{k}^{-2}\omega_{\mathbf{k}}\omega_{-\mathbf{k}}+\frac{1}{2}\sum\limits_{\mathbf{k}\in\mathbb{Z}^{2}_{0}}\mathbf{k}^{2}\theta_{\mathbf{k}}\theta_{-\mathbf{k}}.
\end{equation}

The next natural step is to truncate equations \eqref{MHDEqsFourier} by running the frequency indices $\mathbf{k}$ and $\mathbf{m}$ over the box $\mathcal{B}=([-N,N]\times[-N,N])\setminus(0,0)\ni\mathbf{k},\mathbf{m}$ rather than over $\mathbb{Z}^{2}$. However, this destroys the Hamiltonian structure of the infinite-dimensional equations \eqref{MHDEqsFourier}. In particular, all the Casimirs are lost, except for the quadratic ones, and the truncated bracket \eqref{LiePoissonFourier} fails to fulfill the Jacobi identity.

The way to overcome the mentioned difficulties of the described above ``naive truncation" was suggested by Zeitlin \cite{Ze1991}, based on the ideas of geometric quantization \cite{FZ1989,FFZ1989,Bord1991,Bord1994}. The key idea of Zeitlin's method is as follows. In any Euler--Arnold equation, there are two independent ingredients that define it: the structure constants of the underlying Lie algebra (or, the coadjoint operator) and the metric (or, equivalently, the inertia operator). Therefore, if one introduces appropriate approximations for them, determined by some truncation parameter $N$, converging to the original structure constants and the metric, as $N\to\infty$, then, the discrete analogue of \eqref{MHDvort} will still be an Euler--Arnold system for the modified finite-dimensional Lie--Poisson structure approximating the original infinite-dimensional one.

Let $\mathcal{L}$ be an abstract Lie algebra, i.e., a vector space with a Lie bracket $[\cdot,\cdot]$. If one fixes a basis $L_{1},L_{2},\ldots$ in $\mathcal{L}$, then the corresponding structure constants $C_{ij}^{k}$ are defined as
\begin{equation*}
[L_{i},L_{j}]=\sum_{k}C_{ij}^{k}L_{k}.
\end{equation*}

Let $\mathcal{L}^{*}$ be the dual of $\mathcal{L}$. The Lie--Poisson structure on $\mathcal{L}^{*}$ is given by
\begin{equation*}
\llbracket f,g\rrbracket=\sum_{i,j,k}C_{ij}^{k}\omega_{k}\frac{\delta f}{\delta\omega_{i}}\frac{\delta g}{\delta\omega_{j}},\quad f,g\in C^{\infty}(\mathcal{L}^{*}).
\end{equation*}

The inertia operator $I\colon\mathcal{L}\to\mathcal{L}^{*}$ is usually provided by the metric tensor $g^{lj}$ and defines the energy
\begin{equation*}
H=\frac{1}{2}\sum_{i,j}g^{ij}\omega_{i}\omega_{j}.
\end{equation*}

For the given $g^{lj}$ and $C_{ij}^{k}$, the general Euler--Arnold equation takes the form
\begin{equation}
\label{abstractEulerArnold}
\dot\omega_{i}=\sum_{j,k,l}g^{lj}C_{ij}^{k}\omega_{l}\omega_{k}.
\end{equation}

To put the Fourier transformed MHD equations \eqref{MHDEqsFourier} into the Euler--Arnold framework, we specify $\mathcal{L}$ to be the Lie algebra of smooth functions on $\mathbb{T}^{2}$ with the basis given by $\phi_{\mathbf{n}}=\exp(\mathrm{i}(\mathbf{n}\cdot\mathbf{x}))$ and the Lie algebra structure given by the Poisson bracket $\left\{\cdot,\cdot\right\}$, so that the structure constants are
\begin{equation}
\label{structureconstantsMHDT2}
\left\{\phi_{\mathbf{n}},\phi_{\mathbf{m}}\right\}=(\mathbf{m}\times\mathbf{n})\phi_{\mathbf{n}+\mathbf{m}}=\sum_{\mathbf{k}\in\mathbb{Z}^{2}}C_{\mathbf{n}\mathbf{m}}^{\mathbf{k}}\phi_{\mathbf{k}}\Longrightarrow C_{\mathbf{n}\mathbf{m}}^{\mathbf{k}}=(\mathbf{m}\times\mathbf{n})\delta(\mathbf{k}-\mathbf{n}-\mathbf{m}).
\end{equation}
The metric tensor consists of two blocks, $g_{1}^{\mathbf{k}\mathbf{l}}$ and $g_{2}^{\mathbf{k}\mathbf{l}}$ acting on the vorticity part and the magnetic part respectively:
\begin{equation}
\label{metrictensorMHD}
g_{1}^{\mathbf{k}\mathbf{l}}=\delta(\mathbf{k}+\mathbf{l})|\mathbf{k}|^{-2},\quad g_{2}^{\mathbf{k}\mathbf{l}}=\delta(\mathbf{k}+\mathbf{l})|\mathbf{k}|^{2},
\end{equation}
so that equations \eqref{MHDEqsFourier} take the form of the Euler--Arnold equation \eqref{abstractEulerArnold}:
\begin{equation}
\label{EulerArnoldMHDwithCandG}
\left\{
\begin{aligned}
\dot\omega_{\mathbf{m}}&=\sum\limits_{\mathbf{n},\mathbf{k},\mathbf{l}\in\mathbb{Z}^{2}_{0}}g_{1}^{\mathbf{l}\mathbf{k}}C_{\mathbf{m}\mathbf{k}}^{\mathbf{n}}\omega_{\mathbf{l}}\omega_{\mathbf{n}}+\sum\limits_{\mathbf{n},\mathbf{k},\mathbf{l}\in\mathbb{Z}^{2}_{0}}g_{2}^{\mathbf{l}\mathbf{k}}C_{\mathbf{m}\mathbf{k}}^{\mathbf{n}}\theta_{\mathbf{l}}\theta_{\mathbf{n}},\\
\dot\theta_{\mathbf{m}}&=\sum\limits_{\mathbf{n},\mathbf{k},\mathbf{l}\in\mathbb{Z}^{2}_{0}}g_{1}^{\mathbf{l}\mathbf{k}}C_{\mathbf{m}\mathbf{k}}^{\mathbf{n}}\omega_{\mathbf{l}}\theta_{\mathbf{n}}.
\end{aligned}
\right.
\end{equation}

Zeitlin's idea of truncating \eqref{EulerArnoldMHDwithCandG} suggested in \cite{Ze1991,Ze2005} is to keep a finite number of modes $(\omega_{\mathbf{m}},\theta_{\mathbf{m}})$ while simultaneously modifying the structure constants $C_{\mathbf{n}\mathbf{m}}^{\mathbf{k}}$. The metric tensor \eqref{metrictensorMHD} is suggested to be kept the same.

Let $N$ be some fixed odd number, and consider the box $\mathcal{B}_{0}=[-(N-1)/2,(N-1)/2]^{2}\setminus(0,0)$. Since the vorticity field $\omega(t,\mathbf{x})$ and the magnetic potential field $\theta(t,\mathbf{x})$ are real-valued, the complex Fourier coefficients $\omega_{\mathbf{m}}$ and $\theta_{\mathbf{m}}$ fulfill the \textit{reality condition}. Indeed, $\omega(t,\mathbf{x})=\bar{\omega}(t,\mathbf{x})$ implies $\omega_{-\mathbf{m}}=\bar{\omega}_{\mathbf{m}}$, and also $\theta(t,\mathbf{x})=\bar{\theta}(t,\mathbf{x})$ implies $\theta_{-\mathbf{m}}=\bar{\theta}_{\mathbf{m}}$. Consider a sequence of matrix Lie algebras $\mathfrak{sl}(N,\mathbb{C})$. There exists a basis in $\mathfrak{sl}(N,\mathbb{C})$ with structure constants approximating \eqref{structureconstantsMHDT2}, see, e.g., \cite{FFZ1989,FZ1989}.

Indeed, consider the following two matrices in $\mathfrak{sl}(N,\mathbb{C})$:
\begin{equation*}
F=\mathrm{diag}(1,\varepsilon,\ldots,\varepsilon^{N-1}),\quad H=\begin{pNiceMatrix}[left-margin]
0 & 1 & 0 & \ldots & 0  \\
0  & 0 &  1 & \ldots &  0  \\
\vdots  & \vdots  &   & \ddots   &   \\
0 &  0 &  \ldots  & \ldots  & 1  \\
1 &  0 &  \ldots  & \ldots  & 0  \\
\end{pNiceMatrix},
\end{equation*}
where $\varepsilon$ is a primitive root of $N$-th degree of unity, e.g. $\varepsilon=e^{\mathrm{i}4\pi/N}$.

Further, one defines $(N^{2}-1)$ matrices $J_{\mathbf{m}}$, $\mathbf{m}=(m_{1},m_{2})\in\mathcal{B}_{0}$:
\begin{equation*}
J_{\mathbf{m}}=\varepsilon^{m_{1}m_{2}/2}F^{m_{1}}H^{m_{2}},
\end{equation*}
spanning $\mathfrak{sl}(N,\mathbb{C})$ and satisfying the commutation relations
\begin{equation*}
[J_{\mathbf{n}},J_{\mathbf{m}}]=2\mathrm{i}\cdot\sin\left(\frac{2\pi(\mathbf{n}\times\mathbf{m})}{N}\right)J_{\mathbf{n}+\mathbf{m}},
\end{equation*}
and also $J^{\dag}_{\mathbf{m}}=J_{-\mathbf{m}}$. Rescaling $J_{\mathbf{m}}$ as $L_{\mathbf{m}}=(\mathrm{i}N/4\pi)J_{\mathbf{m}}$, we get a basis with structure constants approximating \eqref{structureconstantsMHDT2}, as $N\to+\infty$:
\begin{equation}
\label{sineStructureConst}
[L_{\mathbf{n}},L_{\mathbf{m}}]=\frac{N}{2\pi}\sin\left(\frac{2\pi(\mathbf{m}\times\mathbf{n})}{N}\right)L_{\mathbf{m}+\mathbf{n}}\Longrightarrow C_{\mathbf{n}\mathbf{m}}^{\mathbf{k}}=\frac{N}{2\pi}\sin\left(\frac{2\pi(\mathbf{m}\times\mathbf{n})}{N}\right)\delta(\mathbf{k}-\mathbf{n}-\mathbf{m}).
\end{equation}

Finally, taking equations \eqref{EulerArnoldMHDwithCandG} with the structure constants \eqref{sineStructureConst}, we get the \textit{sine-MHD equations}:
\begin{equation}
\label{sineMHD}
\left\{
\begin{aligned}
\dot\omega_{\mathbf{m}}&=\sum\limits_{\mathbf{k}\in\mathcal{B}_{0}}|\mathbf{k}|^{-2}\frac{N}{2\pi}\sin\left(\frac{2\pi(\mathbf{k}\times\mathbf{m})}{N}\right)\omega_{(\mathbf{k}+\mathbf{m})|_{\operatorname{mod}N}}\omega_{-\mathbf{k}}+\sum\limits_{\mathbf{k}\in\mathcal{B}_{0}}|\mathbf{k}|^{2}\frac{N}{2\pi}\sin\left(\frac{2\pi(\mathbf{k}\times\mathbf{m})}{N}\right)\theta_{(\mathbf{k}+\mathbf{m})|_{\operatorname{mod}N}}\theta_{-\mathbf{k}},\\
\dot\theta_{\mathbf{m}}&=\sum\limits_{\mathbf{k}\in\mathcal{B}_{0}}|\mathbf{k}|^{-2}\frac{N}{2\pi}\sin\left(\frac{2\pi(\mathbf{k}\times\mathbf{m})}{N}\right)\theta_{(\mathbf{k}+\mathbf{m})|_{\operatorname{mod}N}}\omega_{-\mathbf{k}}.
\end{aligned}
\right.
\end{equation}

Introducing the matrices
\begin{equation*}
\begin{split}
&W=\sum\limits_{\mathbf{m}\in\mathcal{B}_{0}}\omega_{\mathbf{m}}L_{\mathbf{m}},\quad P=\Delta_{N}^{-1}(W)=-\sum_{\mathbf{k}\in\mathcal{B}_{0}}\omega_{\mathbf{k}}|\mathbf{k}|^{-2}L_{\mathbf{k}},\quad {}\\&\Theta=\sum\limits_{\mathbf{m}\in\mathcal{B}_{0}}\theta_{\mathbf{m}}L_{\mathbf{m}},\quad J=\Delta_{N}(\Theta)=-\sum_{\mathbf{k}\in\mathcal{B}_{0}}\theta_{\mathbf{k}}|\mathbf{k}|^{2}L_{\mathbf{k}},
\end{split}
\end{equation*}
one can turn \eqref{sineMHD} into a compact matrix form:
\begin{equation}
\label{MHDZeitlinTorusZeit}
\dot W=[W,P]+[\Theta,J],\quad\dot\Theta=[\Theta,P],\quad W,\Theta,J,P\in\mathfrak{su}(N).
\end{equation}
The matrices $W,\Theta,J,P$ are indeed skew-hermitian, which follows from the reality condition:
\begin{equation*}
W^{\dag}=\sum\limits_{\mathbf{m}\in\mathcal{B}_{0}}\bar{\omega}_{\mathbf{m}}L^{\dag}_{\mathbf{m}}=\sum\limits_{\mathbf{m}\in\mathcal{B}_{0}}\omega_{-\mathbf{m}}\left(\frac{\mathrm{i}N}{4\pi}J_{\mathbf{m}}\right)^{\dag}=\sum\limits_{\mathbf{m}\in\mathcal{B}_{0}}\omega_{-\mathbf{m}}\left(-\frac{\mathrm{i}N}{4\pi}J_{-\mathbf{m}}\right)=-W.
\end{equation*}

The Hamiltonian nature of the sine-MHD equations becomes evident when written in the matrix form \eqref{MHDZeitlinTorusZeit}. Namely, we have the following result, see for details \cite{ModRoop}:
\begin{theorem}
The sine-MHD equations \eqref{MHDZeitlinTorusZeit} are a Lie--Poisson flow on the dual $\mathfrak{f}^{*}$ of the Lie algebra $\mathfrak{f}=\mathfrak{su}(N)\ltimes\mathfrak{su}(N)^{*}$. The quantities \begin{equation*}
\mathcal{C}_{f}^{N}=\mathrm{tr}(f(\Theta)),\quad \mathcal{I}_{m}^{N}=\mathrm{tr}(W\textsl{g}(\Theta)),
\end{equation*}
for arbitrary smooth functions $f$ and $\textsl{g}$, are Casimirs for \eqref{MHDZeitlinTorusZeit}, and the Hamiltonian
\begin{equation*}
H^{N}=-\frac{1}{2}\mathrm{tr}(WP)-\frac{1}{2}\mathrm{tr}(\Theta J).
\end{equation*} is conserved by the flow.
\end{theorem}
\subsection{Zeitlin's model for 3D axisymmetric MHD on $\mathbb{R}^{3}$}
Translationally invariant reduced equations \eqref{MHDaxisymmvortexT2} on $\mathbb{T}^{2}$ can also be discretized analogously to \eqref{MHDvort}. Namely, we observe that the matrix equations \eqref{MHDZeitlinTorusZeit} resemble the structure of their continuous counterparts \eqref{MHDvort}, and are obtained from \eqref{MHDvort} by a formal replacement of functions by matrices defined by their Fourier modes, and Poisson brackets by a matrix commutator, hence one may call this way of discretizing the Euler and MHD equations \textit{quantization}. Following the same strategy, we get the discrete version of \eqref{MHDaxisymmvortexT2}, the \textit{axisymmetric sine-MHD equations}:
\begin{equation}
\label{MHDaxisymmvortexT2matrix}
\left\{
\begin{aligned}
&\dot W=[W,P]+[\Xi,J], \\
&\dot R=[R,P]+[\Sigma,\Xi],\\
&\dot \Sigma=[\Sigma,P]+[R,\Xi],\\
&\dot\Xi=[\Xi,P],
\end{aligned}
\right.
\end{equation}
where the matrices are defined as follows:
\begin{equation*}
\begin{split}
&W=\sum\limits_{\mathbf{m}\in\mathcal{B}_{0}}\omega_{\mathbf{m}}L_{\mathbf{m}},\quad P=-\sum_{\mathbf{k}\in\mathcal{B}_{0}}\omega_{\mathbf{k}}|\mathbf{k}|^{-2}L_{\mathbf{k}},\quad \Xi=\sum\limits_{\mathbf{m}\in\mathcal{B}_{0}}\xi_{\mathbf{m}}L_{\mathbf{m}},\quad J=-\sum_{\mathbf{k}\in\mathcal{B}_{0}}\xi_{\mathbf{k}}|\mathbf{k}|^{2}L_{\mathbf{k}},{}\\& R=\sum\limits_{\mathbf{m}\in\mathcal{B}_{0}}\rho_{\mathbf{m}}L_{\mathbf{m}},\quad \Sigma=\sum\limits_{\mathbf{m}\in\mathcal{B}_{0}}\sigma_{\mathbf{m}}L_{\mathbf{m}},
\end{split}
\end{equation*}
with the respective Fourier coefficients $\omega_{\mathbf{m}},\xi_{\mathbf{m}},\rho_{\mathbf{m}},\sigma_{\mathbf{m}}$.

Equations \eqref{MHDaxisymmvortexT2matrix} possess a Lie--Poisson formulation on the dual of the matrix version of \eqref{LieAlgStructureforT2}. 
\begin{theorem}
The axisymmetric sine-MHD equations \eqref{MHDaxisymmvortexT2matrix} are a Lie--Poisson flow on the dual $\mathfrak{f}^{*}$ of the Lie algebra $\mathfrak{f}=(\mathfrak{su}(N)\ltimes\mathfrak{su}(N)^{*})\oplus(\mathfrak{su}(N)\ltimes\mathfrak{su}(N)^{*})\oplus(\mathfrak{su}(N)\ltimes\mathfrak{su}(N)^{*})$. The quantities
\begin{equation*}
\mathcal{C}_{f}^{N}=\mathrm{tr}(f(\Xi)),\quad J_{h}^{N}=\mathrm{tr}(R h(\Xi)),\quad K_{\textsl{g}}^{N}=\mathrm{tr}(\Sigma\textsl{g}(\Xi)),\quad \mathcal{I}_{k}^{N}=\mathrm{tr}(W k(\Xi)),
\end{equation*}
for arbitrary smooth functions $f$, $\textsl{g}$, $h$, $k$, are Casimirs for \eqref{MHDaxisymmvortexT2matrix}, and the Hamiltonian
\begin{equation*}
H^{N}=-\frac{1}{2}\mathrm{tr}(WP)-\frac{1}{2}\mathrm{tr}(\Theta J)+\mathrm{tr}(R^{2})+\mathrm{tr}(\Sigma^{2})
\end{equation*}
is conserved by the flow.
\end{theorem}
\subsection{Zeitlin's model for 2D MHD on the sphere $\mathbb{S}^{2}$}
\label{Zeitlin2dMHDsphere}

For the two-sphere $\mathbb{S}^{2}$, quantization was developed in the works of Hoppe and Yau \cite{Hopp,Hopp1,HoppYau}. The starting point of the quantization theory for the two-sphere in the quantized version $\Delta_{N}\colon\mathfrak{u}(N)\to\mathfrak{u}(N)$ of the Laplace--Beltrami operator $\Delta$. The operator $\Delta_{N}$ is called the \textit{Hoppe--Yau Laplacian}. Its explicit form is given in \cite{HoppYau}:
\begin{equation*}
\Delta_{N}(\cdot)=\frac{1}{\hbar^2}\left([X_{1}^{N},[X_{1}^{N},\cdot]]+[X_{2}^{N},[X_{2}^{N},\cdot]]+[X_{3}^{N},[X_{3}^{N},\cdot]]\right),\quad\hbar=\frac{2}{\sqrt{N^{2}-1}},
\end{equation*}
where $X_{a}^{N}$, $a=1,2,3$ are generators of a unitary irreducible ``spin $(N-1)/2$" representation of $\mathfrak{so}(3)$, i.e.,
\begin{equation*}
\frac{1}{\hbar}[X_{a}^{N},X_{b}^{N}]=\mathrm{i}\varepsilon_{abc} X_{c},\quad (X_{1}^{N})^{2}+(X_{2}^{N})^{2}+(X_{3}^{N})^{2}=I,
\end{equation*}
where $\varepsilon_{abc}$ is the Levi--Civita symbol, $I$ is the identity matrix, and summation is assumed over repeated indices.

Hoppe and Yau also showed that the spectrum of the quantized Laplacian $\Delta_{N}$ coincides with the spectrum of the continuous Laplacian $\Delta$:
\begin{equation*}
\Delta_{N}T_{lm}^{N}=-l(l+1)T_{lm}^{N},
\end{equation*}
where the matrices $T_{lm}^{N}$, the eigenmatrices of the Hoppe--Yau Laplacian $\Delta_{N}$, are called \textit{matrix harmonics} (in analogy with spherical harmonics, the eigenfunctions of the Laplace operator $\Delta$) and are explicitly given by
\begin{equation}
\label{definitionTLMwigner}
\left(T^{N}_{lm}\right)_{m_{1}m_{2}}=(-1)^{[(N-1)/2]-m_{1}}\sqrt{2l+1}
\begin{pmatrix}
\frac{N-1}{2}&l&\frac{N-1}{2}\\
-m_{1} & m &m_{2}
\end{pmatrix},
\end{equation}
with $(:::)$ being the Wigner 3j-symbol. 

The projection $p_{N}\colon C^{\infty}(\mathbb{S}^{2})\to\mathfrak{u}(N)$ can be defined as follows for $\omega\in C^{\infty}(\mathbb{S}^{2})$:
\begin{equation*}
p_{N}(\omega)=\sum\limits_{l=0}^{N-1}\sum\limits_{m=-l}^{l}\mathrm{i}\omega^{lm}T_{lm}^{N}=W,
\end{equation*}
and given a matrix $W\in\mathfrak{u}(N)$, one reconstructs the first $N^{2}$ Fourier coefficients $\omega^{lm}$ of $\omega\in C^{\infty}(\mathbb{S}^{2})$, using the Frobenius orthogonality of $T_{lm}^{N}$.

Convergence results with the rate $O(1/N)$ for the matrix approximations, i.e. for the functions and Poisson brackets, have been established by Charles and Polterovich in \cite{ChPolt2017}:
\begin{theorem}[Charles, Polterovich]
For every $f,g\in C^{3}(\mathbb{S}^{2})$ there exist $c_{0}, c_{1}>0$, such that
\vspace{0.3cm}
\begin{itemize}
\item 
$
\displaystyle\|f\|_{L^{\infty}}-c_{0}\|f\|_{C^{2}}\hbar\le\|p_{N}(f)\|_{L_{N}^{\infty}}\le\|f\|_{L^{\infty}}
$
\end{itemize}
\begin{itemize}
\item
$
\displaystyle\left\|\frac{1}{\hbar}[p_{N}(f),p_{N}(g)]-p_{N}\left(\left\{f,g\right\}\right)\right\|_{L_{N}^{\infty}}\le c_{1}\hbar 
\left(\|f\|_{C^{1}}\|g\|_{C^{3}}+\|f\|_{C^{2}}\|g\|_{C^{2}}+\|f\|_{C^{3}}\|g\|_{C^{1}}\right),
$
\end{itemize}
where \begin{equation*}
\hbar=\frac{2}{\sqrt{N^{2}-1}},\quad\|f\|_{C^{k}}=\max\limits_{i\le k}\,\sup|\nabla^{i} f|,\quad \|A\|_{L_{N}^{\infty}}=\sup\limits_{\|x\|=1}\|Ax\|,\quad A\in\mathfrak{u}(N).
\end{equation*}
\end{theorem}

We have now all the ingredients to discretize the reduced MHD equations \eqref{MHDvort} on the two-dimensional sphere $\mathbb{S}^{2}$. We well refer to the equations below as \textit{MHD--Zeitlin equations}:
\begin{equation}
\label{qMHD}
\left\{
\begin{aligned}
\displaystyle\dot W&=\frac{1}{\hbar}[W,\Delta^{-1}_{N}W]+\frac{1}{\hbar}[\Theta,\Delta_{N}\Theta], \\
\displaystyle\dot\Theta&=\frac{1}{\hbar}[\Theta,\Delta^{-1}_{N}W].
\end{aligned}
\right.
\end{equation}

The system of equations \eqref{qMHD} is itself a Lie--Poisson system on the dual of a semidirect product matrix Lie algebra, as shown in \cite{ModRoop}:
\begin{theorem}
System \eqref{qMHD} is a Lie--Poisson flow on the dual $\mathfrak{f}^{*}$ of the Lie algebra $\mathfrak{f}=\mathfrak{su}(N)\ltimes\mathfrak{su}(N)^{*}$, with the Hamiltonian
\begin{equation}
\label{hamiltforqMHD2d}
H^{N}(W,\Theta)=-\frac{2\pi}{N}\left(\mathrm{tr}(W\Delta^{-1}_{N}W)+\mathrm{tr}(\Theta\Delta_{N}\Theta)\right).
\end{equation}
The Casimir invariants for \eqref{qMHD} are
\begin{equation}
\label{casimirsforqMHD2d}
\mathcal{C}^N_{f}(\Theta)=\frac{4\pi}{N}\mathrm{tr}\left(f(\Theta)\right),\quad \mathcal{I}^N_{\textsl{g}}(W,\Theta)=\frac{4\pi}{N}\mathrm{tr}(W\textsl{g}(\Theta)),
\end{equation}
for arbitrary smooth functions $f$ and $\textsl{g}$. 
\end{theorem}
Vanishing trace of the matrices $W,\Theta$ follows from the vanishing mean value of the fields $\omega,\theta$, which itself follows from the Stokes theorem. Convergence of the Casimirs \eqref{casimirsforqMHD2d} and the Hamiltonian \eqref{hamiltforqMHD2d} to the corresponding continuous counterparts \eqref{cascont} as $N\to\infty$ is also established in \cite{ModRoop}.
\begin{remark}
We observe that Zeitlin's models for reduced MHD equations \eqref{MHDvort} both on the torus \eqref{MHDZeitlinTorusZeit} and on the sphere \eqref{qMHD} have the same form and evolve on the same matrix Lie algebra $\mathfrak{f}=\mathfrak{su}(N)\ltimes\mathfrak{su}(N)^{*}$.
\end{remark}
\subsection{Zeitlin's model for 3D axisymmetric MHD on $\mathbb{S}^{3}$} Formally replacing functions with matrices, as well as Poisson brackets with scaled matrix commutators $[\cdot,\cdot]_{N}=\frac{1}{\hbar}[\cdot,\cdot]$ in \eqref{LP3D}, we get the following matrix form of equations \eqref{LP3D}:
\begin{equation}
\label{LP3Dquant}
\left\{
\begin{aligned}
\displaystyle\dot W&=\frac{1}{\hbar}[W,\Psi]+\frac{1}{\hbar}[\Xi,J]+\frac{2}{\hbar}[Q,\Psi]+\frac{2}{\hbar}[\Xi,P], \\
\displaystyle\dot P&=\frac{1}{\hbar}[P,\Psi]+\frac{1}{\hbar}[\Xi,Q]+\frac{2}{\hbar}[\Psi,\Xi],\\
\displaystyle\dot Q&=\frac{1}{\hbar}[Q,\Psi]+\frac{1}{\hbar}[\Xi,P],\\
\displaystyle\dot\Xi&=\frac{1}{\hbar}[\Xi,\Psi],
\end{aligned}
\right.
\end{equation}
where $\Psi,W=\Delta_{N}\Psi,J=\Delta_{N}\Xi, Q,P,\Xi\in\mathfrak{su}(N)$, and $\Delta_{N}\colon\mathfrak{su}(N)\to\mathfrak{su}(N)$ is the Hoppe--Yau Laplacian. We will refer to equations \eqref{LP3Dquant} as \textit{axisymmetric MHD--Zeitlin equations}.  

To get the Euler--Arnold formulation of \eqref{LP3Dquant}, we first recall the semidirect product extension $\mathfrak{f}=\mathfrak{su}(N)\ltimes\mathfrak{su}(N)^{*}$ with the Lie algebra structure (we will use the notation $[\cdot,\cdot]_{N}$ for the scaled bracket), analogously to \eqref{Liealgebrastructureonsemidir}:
\begin{equation}
\label{adonmatrices}
[(P,B),(U,V)]=\mathrm{ad}_{(P,B)}(U,V)=([P,U]_{N},\mathrm{ad}^{*}_{U}B-\mathrm{ad}^{*}_{P}V),
\end{equation}
where $P,U\in\mathfrak{su}(N)$, $B,V\in\mathfrak{su}(N)^{*}$.

With the identification of $\mathfrak{f}^{*}$ with $\mathfrak{f}$ via the pairing
\begin{equation*}
\langle(P_{1},B_{1}),(P_{2},B_{2})\rangle=\langle B_{2},P_{1}\rangle_{F}+\langle B_{1},P_{2}\rangle_{F}=\mathrm{tr}(B_{2}^{\dag},P_{1})+\mathrm{tr}(B_{1}^{\dag},P_{2}),
\end{equation*}
where $\langle\cdot,\cdot\rangle_{F}$ is the Frobenius inner product inducing the pairing between $\mathfrak{su}(N)$ and $\mathfrak{su}(N)^{*}$, we get the coadjoint operator on $\mathfrak{f}^{*}$:
\begin{equation}
\begin{split}
\label{coadjointoperatorforabelinext}
\mathrm{ad}^{*}_{(P,B)}(U,V)&=\left([U,P]_{N},\mathrm{ad}^{*}_{P}V-\mathrm{ad}^{*}_{U}B\right){}\\&=([U,P]_{N},[P^{\dag},V]_{N}-[U^{\dag},B]_{N})=([U,P]_{N},[V,P]_{N}+[U,B]_{N}).
\end{split}
\end{equation}

Next, we construct the magnetic extension $\mathfrak{F}$ of $\mathfrak{f}$ in the standard way, $\mathfrak{F}=\mathfrak{f}\ltimes\mathfrak{f}^{*}=(\mathfrak{su}(N)\ltimes\mathfrak{su}(N)^{*})\ltimes(\mathfrak{su}(N)\ltimes\mathfrak{su}(N)^{*})^{*}$. Using formula \eqref{coadjopsemidir} for the coadjoint operator on a semidirect product, we get
\begin{equation}
\label{coadjopforextensonextension}
\mathrm{ad}^{*}_{((\alpha,\beta),(\gamma,\delta))}((\Xi,-Q),(P,W))=\left([(\Xi,-Q),(\alpha,\beta)],\mathrm{ad}^{*}_{(\alpha,\beta)}(P,W)-\mathrm{ad}^{*}_{(\Xi,-Q)}(\gamma,\delta)\right),
\end{equation}
for $((\alpha,\beta),(\gamma,\delta))\in\mathfrak{F}$, and where the bracket $[\cdot,\cdot]$ is given by \eqref{adonmatrices}, and $\mathrm{ad}^{*}$ is defined in \eqref{coadjointoperatorforabelinext}.

Using \eqref{adonmatrices}, \eqref{coadjointoperatorforabelinext}, and \eqref{coadjopforextensonextension}, the Euler--Arnold equation for the pair $(\Xi,Q)$ is
\begin{equation}
\label{matrixequationforPsigma}
\left\{
\begin{aligned}
&\dot Q=[Q,\alpha]_{N}+[\beta,\Xi]_{N}, \\
&\dot\Xi=[\Xi,\alpha]_{N}.
\end{aligned}
\right.
\end{equation}
For the pair $(P,W)$, the equations are
\begin{equation*}
\frac{\mathrm{d}}{\mathrm{d}t}(P,W)=\mathrm{ad}^{*}_{(\alpha,\beta)}(P,W)-\mathrm{ad}^{*}_{(\Xi,-Q)}(\gamma,\delta),
\end{equation*}
and using \eqref{coadjointoperatorforabelinext}, we finally get
\begin{equation}
\label{rhoxieulerarnoldmatrix}
\left\{
\begin{aligned}
\dot W&=[W,\alpha]_{N}+[\Xi,\delta]_{N}+[P,\beta]_{N}+[\gamma,Q]_{N}, \\
\dot P&=[P,\alpha]_{N}+[\Xi,\gamma]_{N}.
\end{aligned}
\right.
\end{equation}
Choosing the algebra variables $\alpha,\beta,\gamma,\delta$ to be
\begin{equation*}
\alpha=\Psi,\quad\beta=-P-2\Xi,\quad\gamma=Q-2\Psi,\quad\delta=J,
\end{equation*}
we observe that equations \eqref{matrixequationforPsigma}-\eqref{rhoxieulerarnoldmatrix} coincide with \eqref{LP3Dquant}. Thus, we arrive at the second central result of this paper.
\begin{theorem}
The axisymmetric MHD--Zeitlin equations \eqref{LP3Dquant} have the form of a Lie--Poisson equation
\begin{equation*}
\frac{\mathrm{d}}{\mathrm{d}t}((\Xi,-Q),(P,W))=\mathrm{ad}^{*}_{((\alpha,\beta),(\gamma,\delta))}((\Xi,-Q),(P,W))
\end{equation*}
on the dual $\mathfrak{F}^{*}$ of the Lie algebra $\mathfrak{F}=(\mathfrak{su}(N)\ltimes\mathfrak{su}(N))\ltimes(\mathfrak{su}(N)\ltimes\mathfrak{su}(N))^{*}$, with the Hamiltonian function given by
\begin{equation}
\label{energyreduced3Dmatrix}
H^{N}=-\frac{2\pi}{N}\mathrm{tr}\left(W\Psi-Q^{2}\right)-\frac{2\pi}{N}\mathrm{tr}\left(\Xi J-P^{2}\right).
\end{equation}
\end{theorem}
\subsection{Casimir functions for the matrix equations}
As a standard implication of the Lie--Poisson nature of the flow \eqref{LP3Dquant}, there exist matrix analogues of the Casimirs \eqref{cas3Dnotfull}-\eqref{crosshelicityreduced}. Since polynomials are dense in $C^{\infty}$, we can restrict the arbitrary function in \eqref{cas3Dnotfull}-\eqref{crosshelicityreduced} to be monomials, i.e. by Casimirs for the system \eqref{LP3Dquant} we will mean
\begin{equation}
\label{casimirsquantizedmonomials}
\begin{split}
&\mathcal{C}_{m}^{N}=\frac{4\pi}{N}\mathrm{tr}(\Xi^{m}),\quad J_{m}^{N}=\frac{4\pi}{N}\mathrm{tr}(P \Xi^{m}),\quad K_{m}^{N}=\frac{4\pi}{N}\mathrm{tr}(Q\Xi^{m}),{}\\&
\mathcal{I}_{m}^{N}=\frac{4\pi}{N}\mathrm{tr}\left(W\Xi^{m}-\sum_{j=0}^{m-1}\Xi^{j}P\Xi^{m-1-j}Q\right).
\end{split}
\end{equation}
which are matrix analogues for
\begin{equation}
\label{continuousmonomialsformahd3d}
\mathcal{I}_{m}=\int_{\mathbb{S}^{2}}(\omega\xi^{m}-m\rho q\xi^{m-1})\mu,\quad\mathcal{C}_{m}=\int_{\mathbb{S}^{2}}\xi^{m}\mu,\quad J_{m}=\int_{\mathbb{S}^{2}}\rho \xi^{m}\mu,\quad K_{m}=\int_{\mathbb{S}^{2}}q\xi^{m}\mu.
\end{equation}

The formula for the Casimir $\mathcal{I}_{m}^{N}$ in \eqref{casimirsquantizedmonomials} is a counterintuitive analogue for $\mathcal{I}_{m}$ in \eqref{continuousmonomialsformahd3d}. Indeed, replacing matrices in \eqref{casimirsquantizedmonomials} by functions, and the trace operation by integration, one recovers \eqref{continuousmonomialsformahd3d}.
\begin{proposition}
The quantities \eqref{casimirsquantizedmonomials} are Casimir invariants for \eqref{LP3Dquant}. The Hamiltonian \eqref{energyreduced3Dmatrix} is also conserved by the flow.
\end{proposition}
\begin{proof}
We will prove the statement for the cross-helicity Casimir $\mathcal{I}_{m}^{N}$, others are treated analogously. To that end, we calculate the evolution of the Casimir $\mathcal{I}_{m}^{N}$ along trajectories of the general matrix Euler--Arnold equation \eqref{matrixequationforPsigma}-\eqref{rhoxieulerarnoldmatrix}.  We get
\begin{equation*}
\begin{split}
\frac{\mathrm{d}\mathcal{I}_{m}^{N}}{\mathrm{d}t}&=\frac{4\pi}{N}\mathrm{tr}\left(\dot W\Xi^{m}+W\frac{\mathrm{d}}{\mathrm{d}t}\left(\Xi^{m}\right)\right.{}\\&-\left.\sum_{j=0}^{m-1}\left(\frac{\mathrm{d}}{\mathrm{d}t}\left(\Xi^{j}\right)P\Xi^{m-1-j}Q+\Xi^{j}\dot P\Xi^{m-1-j}Q+\Xi^{j}P\frac{\mathrm{d}}{\mathrm{d}t}\left(\Xi^{m-1-j}\right)Q+\Xi^{j}P\Xi^{m-1-j}\dot Q\right)\right).
\end{split}
\end{equation*}
Further, for the first term, using \eqref{rhoxieulerarnoldmatrix}, we get
\begin{equation}
\label{firsttermincasimir}
\begin{split}
\mathrm{tr}\left(\dot W\Xi^{m}+W\frac{\mathrm{d}}{\mathrm{d}t}\left(\Xi^{m}\right)\right)&=\mathrm{tr}\left([W,\alpha]\Xi^{m}+[\Xi,\delta]\Xi^{m}+[P,\beta]\Xi^{m}+[\gamma, Q]\Xi^{m}+W[\Xi^{m},\alpha]\right){}\\&
=\mathrm{tr}\left(P\beta\Xi^{m}-\beta P\Xi^{m}+\gamma Q\Xi^{m}-Q\gamma\Xi^{m}\right).
\end{split}
\end{equation}
Using \eqref{matrixequationforPsigma}, we further obtain
\begin{equation}
\label{manytermsincasimir}
\begin{split}
&\mathrm{tr}\left(\frac{\mathrm{d}}{\mathrm{d}t}\left(\Xi^{j}\right)P\Xi^{m-1-j}Q\right)=\mathrm{tr}([\Xi^{j},\alpha]P\Xi^{m-1-j}Q)=\mathrm{tr}(\Xi^{j}\alpha P\Xi^{m-1-j}Q-\alpha\Xi^{j}P\Xi^{m-1-j}Q),{}\\&
\mathrm{tr}(\Xi^{j}\dot P\Xi^{m-1-j}Q)=\mathrm{tr}\left(\Xi^{j}P\alpha\Xi^{m-1-j}Q-\Xi^{j}\alpha P\Xi^{m-1-j}Q+\Xi^{j+1}\gamma\Xi^{m-1-j}Q-\Xi^{j}\gamma\Xi^{m-j}Q\right),{}\\&
\mathrm{tr}\left(\Xi^{j}P\frac{\mathrm{d}}{\mathrm{d}t}\left(\Xi^{m-1-j}\right)Q\right)=\mathrm{tr}(\Xi^{j}P\Xi^{m-1-j}\alpha Q-\Xi^{j}P\alpha\Xi^{m-1-j}Q),{}\\&
\mathrm{tr}\left(\Xi^{j}P\Xi^{m-1-j}\dot Q\right)=\mathrm{tr}(\Xi^{j}P\Xi^{m-1-j}Q\alpha-\Xi^{j}P\Xi^{m-1-j}\alpha Q+\Xi^{j}P\Xi^{m-1-j}\beta\Xi-\Xi^{j}P\Xi^{m-j}\beta),
\end{split}
\end{equation}
Summing up all the terms in \eqref{manytermsincasimir} and \eqref{firsttermincasimir}, we find
\begin{equation}
\label{almostfinalwithsums}
\begin{split}
\frac{\mathrm{d}\mathcal{I}_{m}^{N}}{\mathrm{d}t}=\frac{4\pi}{N}\mathrm{tr}&\left(P\beta\Xi^{m}-\beta P\Xi^{m}-\sum_{j=0}^{m-1}(\Xi^{j+1}P\Xi^{m-1-j}\beta-\Xi^{j}P\Xi^{m-j}\beta)\right.{}\\&\left.+\gamma Q\Xi^{m}-Q\gamma\Xi^{m}-\sum_{j=0}^{m-1}(\Xi^{j+1}\gamma\Xi^{m-1-j}Q-\Xi^{j}\gamma\Xi^{m-j}Q)\right).
\end{split}
\end{equation}
Rearranging the sums in \eqref{almostfinalwithsums} we obtain
\begin{equation}
\label{formulawithPbeta}
\begin{split}
\mathrm{tr}\left(\sum_{j=0}^{m-1}(\Xi^{j+1}P\Xi^{m-1-j}\beta-\Xi^{j}P\Xi^{m-j}\beta)\right)&=\mathrm{tr}\left(\sum_{j^{\prime}=1}^{m}(\Xi^{j^{\prime}}P\Xi^{m-j^{\prime}}\beta-\sum_{j=0}^{m-1}\Xi^{j}P\Xi^{m-j}\beta\right){}\\&=\mathrm{tr}(\Xi^{m}P\beta-P\Xi^{m}\beta)=\mathrm{tr}(P\beta\Xi^{m}-\beta P\Xi^{m}),
\end{split}
\end{equation}
and in a similar way we also find
\begin{equation}
\label{formulawithgammaQ}
\mathrm{tr}\left(\sum_{j=0}^{m-1}(\Xi^{j+1}\gamma\Xi^{m-1-j}Q-\Xi^{j}\gamma\Xi^{m-j}Q)\right)=\mathrm{tr}(\gamma Q\Xi^{m}-Q\gamma\Xi^{m}).
\end{equation}
From \eqref{formulawithPbeta} and \eqref{formulawithgammaQ}, we conclude that the right hand side of \eqref{almostfinalwithsums} vanishes, which finalizes the proof.
\end{proof}
Let us establish convergence results for the discretized versions of Casimirs \eqref{casimirsquantizedmonomials} and the Hamiltonian \eqref{energyreduced3Dmatrix}. 

We observe that convergence of $\mathcal{C}_{m}^{N}$, $J_{m}^{N}$, and $K_{m}^{N}$ to the corresponding quantities $\mathcal{C}_{m}$, $J_{m}$, and $K_{m}$, as $N\to\infty$, follows from the results obtained in \cite{ModRoop}. Namely, the following holds:
\begin{theorem}
Let $m>2$. There exist a constant $c_{1}>0$ such that for all $\xi,\rho,q\in C^2(\mathbb{S}^2)$  we have
  \begin{equation*}
    \lvert \mathcal{C}^N_m - \mathcal{C}_m \rvert \leq \beta_{1}(\xi,m)\hbar,\quad
   \left|J^{N}_{m}-J_{m}\right|\le4\pi\|\rho\|_{\infty}\beta_{2}(\xi,m)\hbar,\quad\left|K^{N}_{m}-K_{m}\right|\le4\pi\|q\|_{\infty}\beta_{2}(\xi,m)\hbar,
  \end{equation*}
  where
  \begin{equation*}
  \begin{split}
&\beta_{1}(\xi,m)=c_{1}\sum\limits_{j=2}^{m-1}(4\pi\|\xi\|_{\infty})^{m-j}\left(\|\xi\|_{\infty}\|\xi^{j-1}\|_{C^{2}}+\|\xi\|_{C^{1}}\|\xi^{j-1}\|_{C^{1}}\right),{}\\&
\beta_{2}(\xi,m)=c_{1}\sum\limits_{j=2}^{m}(4\pi\|\xi\|_{\infty})^{m-j}\left(\|\xi\|_{\infty}\|\xi^{j-1}\|_{C^{2}}+\|\xi\|_{C^{1}}\|\xi^{j-1}\|_{C^{1}}\right)
\end{split}
\end{equation*}
  In particular, $\lvert \mathcal{C}^N_m - \mathcal{C}_m \rvert = \mathcal{O}(1/N)$, $\lvert J^N_m - J_m \rvert = \mathcal{O}(1/N)$, $\lvert K^N_m - K_m \rvert = \mathcal{O}(1/N)$ as $N\to\infty$, for $m>2$.
\end{theorem}
For quadratic Casimirs, i.e. for $\mathcal{C}_{2}^{N}$ and $\mathcal{I}_{1}^{N}$, $J_{1}^{N}$, $K_{1}^{N}$, as well as for the Hamiltonian $H^{N}$, one can get a sharper estimate giving $\mathcal{O}(1/N^{s})$ convergence, where $s$ is the regularity parameter. Indeed, the scaled Frobenius inner product $\langle A,B\rangle_{F}=(4\pi/N)\operatorname{tr}(A^{\dag}B)$ on $\mathfrak{su}(N)$ converges to the $L^{2}$ inner product with the rate depending on the regularity of the approximated fields, i.e. (see \cite[lemma 5]{MoPe2024})
\begin{equation}
  \label{estim2CfN}
  \left|\langle p_{N}(f),p_{N}(g)\rangle_{F}-\langle f,g\rangle_{L^{2}}\right|\le2\hbar^{s}\|f\|_{H^{s}(\mathbb{S}^{2})}\|g\|_{H^{s}(\mathbb{S}^{2})}=\mathcal{O}(1/N^{s}),
\end{equation}
for every $s>1$.

From estimate \eqref{estim2CfN}, we deduce the following result:
\begin{theorem}
Let $\rho,q\in H^{s}(\mathbb{S}^{2})$, $\xi,\psi\in H^{s+2}(\mathbb{S}^{2})$. Then, $\mathcal{I}^{N}_{1}$ and $H^{N}$ in \eqref{casimirsquantizedmonomials} and \eqref{energyreduced3Dmatrix} respectively, converge to $\mathcal{I}_{1}$ and $H$, as $N\to\infty$, and the following estimate holds:
\begin{equation*}
\begin{split}
&|\mathcal{I}^{N}_{1}-\mathcal{I}_{1}|\le2\hbar^{s}\left(\|\xi\|_{H^{s}(\mathbb{S}^{2})}\|\psi\|_{H^{s+2}(\mathbb{S}^{2})}+\|\rho\|_{H^{s}(\mathbb{S}^{2})}\|q\|_{H^{s}(\mathbb{S}^{2})}\right),{}\\&
|H^{N}-H|\le2\hbar^{s}(\|\psi\|_{H^{s}(\mathbb{S}^{2})}\|\psi\|_{H^{s+2}(\mathbb{S}^{2})}+\|q\|^{2}_{H^{s}(\mathbb{S}^{2})}+\|\rho\|^{2}_{H^{s}(\mathbb{S}^{2})}+\|\xi\|_{H^{s}(\mathbb{S}^{2})}\|\xi\|_{H^{s+2}(\mathbb{S}^{2})}).
\end{split}
\end{equation*}
\end{theorem}
\section{Time discretization of matrix equations}
\label{sec5}
Now that we have obtained the spatially discretized version \eqref{LP3Dquant} of equations \eqref{LP3D} via geometric quantization, the following question arises: how does one integrate equations \eqref{LP3Dquant} in time in a way that preserves their geometric structure, which in particular means the preservation of Casimirs \eqref{casimirsquantizedmonomials} and Hamiltonian \eqref{energyreduced3Dmatrix}? Since both the Casimirs and the Hamiltonian are preserved by the exact flow \eqref{LP3Dquant}, it is desirable to achieve the same properties when discretizing the flow in time, in order to assure qualitatively correct long time behavior.

As was shown in the previous sections, system \eqref{LP3Dquant} is a Lie--Poisson system, which means that the flow given in \eqref{LP3Dquant} is confined to the coadjoint orbits of the $\mathfrak{F}$-action on $\mathfrak{F}^{*}$, and that the flow is a symplectic map on the corresponding orbit. 

Let $0<t_{1}<t_{2}<\cdots<t_{n}<\cdots<T$ be the equidistant time mesh, i.e., $h(n)=h=t_{n}-t_{n-1}$, and let us use the subscript $n$ to denote the matrices in \eqref{LP3Dquant} at the time moment $t_{n}$, e.g., $\Xi(t_{n})=\Xi_{n}$. We will rescale the time appropriately to avoid the appearance of the multiplier $1/\hbar$ in \eqref{LP3Dquant}.
\begin{definition}
\label{definitionofLPintegr}
The integration map $\phi_{h}\colon\mathfrak{F}^{*}\to\mathfrak{F}^{*},\,\phi_{h}\colon(W_{n},P_{n},Q_{n},\Xi_{n})\mapsto(W_{n+1},P_{n+1},Q_{n+1},\Xi_{n+1})$ is called a \textit{Lie--Poisson integrator}, if it satisfies the following properties:
\begin{itemize}
\item The map $\phi_{h}$ preserves the coadjoint orbits, in particular, preserves the Casimirs \eqref{casimirsquantizedmonomials}:
\begin{equation}
\label{casimirspreservation}
\begin{split}
&\mathcal{I}^{N}_{m}(W_{n},P_{n},Q_{n},\Xi_{n})=\mathcal{I}^{N}_{m}(W_{n+1},P_{n+1},Q_{n+1},\Xi_{n+1}),\quad\mathcal{C}^{N}_{m}(\Xi_{n})=\mathcal{C}_{m}^{N}(\Xi_{n+1}),{}\\& J^{N}_{m}(P_{n},\Xi_{n})=J_{m}^{N}(P_{n+1},\Xi_{n+1}),\quad K^{N}_{m}(Q_{n},\Xi_{n})=K_{m}^{N}(Q_{n+1},\Xi_{n+1}).
\end{split}
\end{equation}
\item The map $\phi_{h}$ is a symplectomorphism of the coadjoint orbit.
\end{itemize}
\end{definition}
\begin{remark}
\begin{enumerate}
\item The equalities in \eqref{casimirspreservation} must be understood up to machine precision. It means that the only contribution to the numerical errors in Casimirs can come from round-off errors when implementing on a computer.
\item Preservation of the Casimirs $\mathcal{C}_{m}^{N}$ is equivalent to the preservation of the spectrum of the matrix $\Xi$. Indeed, the last equation in \eqref{LP3Dquant} prescribes the dynamics for $\Xi$ to be isospectral.
\item 
The flow provided by a Lie--Poisson integrator is a Hamiltonian flow on the coadjoint orbit, but for a modified Hamiltonian function. The closeness of the modified Hamiltonian to the exact Hamiltonian is defined by the convergence order of the integrator, and therefore one can only expect near preservation of the Hamiltonian, $
H^{N}(W_{n},P_{n},Q_{n},\Xi_{n})\approx H^{N}(W_{n+1},P_{n+1},Q_{n+1},\Xi_{n+1})
$ in the sense of backward error analysis. For a more detailed discussion, we refer to \cite{HaiLubWan}.
\end{enumerate}
\end{remark}

For canonical Hamiltonian systems on $T^{*}\mathbb{R}^{n}$, there exists a collection of symplectic integrators called symplectic Runge--Kutta schemes \cite{Sa1988}. An example of such a scheme is the implicit midpoint method. However, symplectic Runge--Kutta schemes do not yield Lie--Poisson integrators when applied directly to a Lie--Poisson system, and the problem of finding Lie--Poisson integrators is therefore more involved.

There exist several approaches to constructing Lie--Poisson (or, more generally, Poisson) integrators. If the Hamiltonian can be split into a sum of integrable Hamiltonians, one can use splitting methods developed in \cite{McQui1,McQui2}. In \cite{ChanScov,MarsPekSh}, the discrete Lie--Poisson flow is constructed from discrete invariant Lagrangians. The constrained symplectic integrator RATTLE is used in \cite{ChanScov1,Jay,McLModVerdWil}. Most of the mentioned approaches result in numerical schemes relying on group to algebra maps, such as exponential maps, and therefore become computationally expensive if the dimension of the corresponding Lie--Poisson system is large, which typically is the case for the Lie--Poisson flows originating from matrix equations for fluids.

A collection of computationally cheap implicit time integrators for isospectral flows (both Hamiltonian and non-Hamiltonian) was developed in \cite{ModViv,Milo}. The  motivation for such integrators came in part from the Zeitlin model of the incompressible Euler equations, obtained from \eqref{qMHD} by neglecting the magnetic field matrix $\Theta$, i.e. by putting $\Theta=0$, and therefore making the flow for the vorticity matrix $W$ isospectral.  The methods are formulated explicitly on the Lie algebra $\mathfrak{g}=\mathfrak{su}(N)$ and avoid computationally expensive exponential maps. The methods developed in \cite{ModViv,Milo} are referred to as \textit{isospectral symplectic Runge--Kutta methods (IsoSyRK)}. The IsoSyRK methods were later applied to the Euler--Zeitlin equations to simulate two-dimensional hydrodynamical turbulence on the sphere \cite{ModViv1,CiViMo2023}. The \textit{magnetic extension} of the IsoSyRK methods was developed in \cite{ModRoop}. The extended version of the IsoSyRK methods is a Lie--Poisson integrator for the Lie--Poisson flows on the dual of the magnetic extension Lie algebra $\mathfrak{f}=\mathfrak{su}(N)\ltimes\mathfrak{su}(N)^{*}$. The developed methods were applied to the MHD--Zeitlin equations \eqref{qMHD} in \cite{ModRoopJPP} to study the MHD turbulence on the sphere. The key idea behind the derivation of IsoSyRK and their magnetic extension is the discrete Lie--Poisson reduction. Namely, a Lie--Poisson flow on $\mathfrak{g}^{*}$ can be lifted to a $G$-invariant Hamiltonian system on $T^{*}G$, and a symplectic $G$-equivariant method $\Phi_{h}\colon T^{*}G\to T^{*}G$ descends to a Lie--Poisson integrator $\phi_{h}\colon\mathfrak{g}^{*}\to\mathfrak{g}^{*}$ via the momentum map $\mu\colon T^{*}G\to\mathfrak{g}^{*}$, see \cite{ModViv,ModRoop}, and also \cite[Ch. VII.5.5]{HaiLubWan}. In this section, we derive two Lie--Poisson integrators for \eqref{LP3Dquant} and demonstrate their preservation properties via numerical simulations.
\subsection{Lie--Poisson integrator for semidirect products}
One observes that the key structure that unifies all the versions of matrix MHD models, both on the sphere and on the torus, is the semidirect product Lie algebra structure. The Lie--Poisson time integrator for the MHD--Zeitlin system \eqref{qMHD} was developed in \cite{ModRoop} and has the following form:
\begin{equation}
\label{integratorforsemidirect}
    \begin{split}
      W_{n}
        &=\widetilde W-\frac{h}{2}[\widetilde W,\widetilde{P}]-\frac{h}{2}[\widetilde\Theta,\widetilde J]-\frac{h^{2}}{4}\left(\widetilde{P}\widetilde{W}\widetilde{P}+\widetilde J\widetilde\Theta\widetilde{P}+\widetilde{P}\widetilde\Theta\widetilde J\right),\\
      \Theta_{n}
        &= \widetilde\Theta-\frac{h}{2}[\widetilde\Theta,\widetilde{P}]-\frac{h^{2}}{4}\widetilde{P}\widetilde\Theta\widetilde{P},
    \\
      W_{n+1}
        &= W_{n}+h[\widetilde W,\widetilde P]+h[\widetilde\Theta,\widetilde J],
    \\
      \Theta_{n+1}
        &= \Theta_{n}+h[\widetilde\Theta,\widetilde{P}],
    \end{split}
\end{equation}
where $\widetilde W=\Delta_{N}\widetilde P$, $\widetilde J=\Delta_{N}\widetilde\Theta$.

Since the structure of the Lie algebra underlying the sine-MHD equations \eqref{sineMHD} is the same, the integrator \eqref{integratorforsemidirect} is also applicable to the reduced MHD equations on the torus $\mathbb{T}^{2}$. 

Furthermore, using that the structure of the axisymmetric sine-MHD equations \eqref{MHDaxisymmvortexT2matrix} suggests their splitting into a number of subsystems evolving on semidirect products, namely,
\begin{equation*}
\left\{
\begin{aligned}
&\dot W=[W,P]+[\Xi,J], \\
&\dot\Xi=[\Xi,P],
\end{aligned}
\right.
\quad
\left\{
\begin{aligned}
&\dot R=[R,P]+[\Xi,-\Sigma],\\
&\dot\Xi=[\Xi,P],
\end{aligned}
\right.
\quad
\left\{
\begin{aligned}
&\dot \Sigma=[\Sigma,P]+[\Xi,-R],\\
&\dot\Xi=[\Xi,P].
\end{aligned}
\right.
\end{equation*}
one can adopt the integrator \eqref{integratorforsemidirect} to the matrix equations \eqref{MHDaxisymmvortexT2matrix}:
\begin{equation}
\label{integratorforsemidirectT2}
\left\{
\begin{aligned}
      &W_{n}
        =\widetilde W-\frac{h}{2}[\widetilde W,\widetilde{P}]-\frac{h}{2}[\widetilde\Xi,\widetilde J]-\frac{h^{2}}{4}\left(\widetilde{P}\widetilde{W}\widetilde{P}+\widetilde J\widetilde\Xi\widetilde{P}+\widetilde{P}\widetilde\Xi\widetilde J\right),\\
      &R_{n}
        =\widetilde R-\frac{h}{2}[\widetilde R,\widetilde{P}]+\frac{h}{2}[\widetilde\Xi,\widetilde \Sigma]-\frac{h^{2}}{4}\left(\widetilde{P}\widetilde{R}\widetilde{P}-\widetilde \Sigma\widetilde\Xi\widetilde{P}-\widetilde{P}\widetilde\Xi\widetilde \Sigma\right),\\
    &\Sigma_{n}
        =\widetilde\Sigma-\frac{h}{2}[\widetilde \Sigma,\widetilde{P}]+\frac{h}{2}[\widetilde\Xi,\widetilde R]-\frac{h^{2}}{4}\left(\widetilde{P}\widetilde{\Sigma}\widetilde{P}-\widetilde {R}\widetilde\Xi\widetilde{P}-\widetilde {P}\widetilde\Xi\widetilde{R}\right),\\
      &\Xi_{n}
        = \widetilde\Xi-\frac{h}{2}[\widetilde\Xi,\widetilde{P}]-\frac{h^{2}}{4}\widetilde{P}\widetilde\Xi\widetilde{P},
   \end{aligned}
   \right.
\end{equation}
\begin{equation}
\label{integratorforsemidirectT22}
\left\{
   \begin{aligned}
      &W_{n+1}
        = W_{n}+h[\widetilde W,\widetilde P]+h[\widetilde\Xi,\widetilde J],
    \\
      &R_{n+1}
        = R_{n}+h[\widetilde R,\widetilde P]+h[\widetilde\Sigma,\widetilde \Xi],
    \\
      &\Sigma_{n+1}
        = \Sigma_{n}+h[\widetilde \Sigma,\widetilde P]+h[\widetilde R,\widetilde\Xi],
    \\
      &\Xi_{n+1}
        = \Xi_{n}+h[\widetilde\Xi,\widetilde{P}].
   \end{aligned}
   \right.
\end{equation}
Both time integration schemes \eqref{integratorforsemidirect} and \eqref{integratorforsemidirectT2}-\eqref{integratorforsemidirectT22} are Lie--Poisson integrators for the respective Lie--Poisson matrix systems \eqref{qMHD} and \eqref{sineMHD}.
\subsection{Lie--Poisson integrator for axisymmetric MHD--Zeitlin equations}
Let $\mathfrak{f}\subset\mathfrak{gl}(N,\mathbb{C})$ be a matrix Lie subalgebra defined by a $\mathcal{J}$-quadratic constraint:
\begin{equation}
\label{JquadraticConstr}
A\in\mathfrak{f}\Longleftrightarrow A^{\dag}\mathcal{J}+\mathcal{J}A=0,
\end{equation}
with $\mathcal{J}$ being such a matrix that $\mathcal{J}^{2}=cI_{N}$, where $I_{N}$ is the identity matrix, $c\in\mathbb{R}\setminus\left\{0\right\}$. A Lie subalgebra $\mathfrak{f}\subset\mathfrak{gl}(N,\mathbb{C})$ is called \textit{reductive}, if $[\mathfrak{f}^{\dag},\mathfrak{f}]\subset\mathfrak{f}$.

Consider a Hamiltonian isospectral flow on $\mathfrak{f}$
\begin{equation}
\label{isospectrHamilt}
\dot A=[A,B(A)],\quad A\in\mathfrak{f},
\end{equation}
where $B(A)=\nabla H(A)^{\dag}$ for some Hamiltonian function $H(A)$. The following result holds (see \cite{ModViv,Milo}):
\begin{theorem}[Modin,Viviani]
\label{thmMiloKlas}
Let $\mathfrak{f}\subset\mathfrak{gl}(N,\mathbb{C})$ be a reductive, $\mathcal{J}$-quadratic Lie subalgebra. Then, the map $\phi\colon\mathfrak{f}^{*}\to\mathfrak{f}^{*},\,\phi_{h}\colon A_{n}\mapsto A_{n+1}$ defined by the equations
\begin{equation}
\label{LPintegrMilo}
\begin{aligned}
A_{n}&=\left(I_{N}+\frac{h}{2}B(\widetilde A)\right)\widetilde A\left(I_{N}-\frac{h}{2}B(\widetilde A)\right)=\widetilde A-\frac{h}{2}[\widetilde A,B(\widetilde A)]-\frac{h^{2}}{4}B(\widetilde A)\widetilde A B(\widetilde A),\\
A_{n+1}&=\left(I_{N}-\frac{h}{2}B(\widetilde A)\right)\widetilde A\left(I_{N}+\frac{h}{2}B(\widetilde A)\right)=\widetilde A+\frac{h}{2}[\widetilde A,B(\widetilde A)]-\frac{h^{2}}{4}B(\widetilde A)\widetilde A B(\widetilde A),
\end{aligned}
\end{equation}
is a Lie--Poisson integrator for the flow \eqref{isospectrHamilt}. In particular, the integrator \eqref{LPintegrMilo} preserves the spectrum of the matrix $A$, or, equivalently,
\begin{equation}
\label{isospectralityMiloKlas}
\mathrm{tr}(A_{n}^{m})=\mathrm{tr}(A_{n+1}^{m}),\quad m=1,2,\ldots,N.
\end{equation}
\end{theorem}

Let us introduce the Lie subalgebra $\mathfrak{f}\subset\mathfrak{gl}(4N,\mathbb{C})$ consisting of the following block lower triangular matrices:
\begin{equation}
\label{definitionofAB1}
A=\begin{pmatrix}
A_{1} & 0\\
A_{2} & A_{1}
\end{pmatrix}
,\quad
B(A)=\begin{pmatrix}
B_{1} & 0\\
B_{2} & B_{1}
\end{pmatrix}
,
\end{equation}
where the blocks $A_{1},A_{2}$, and $B_{1},B_{2}$ have the following structure:
\begin{equation}
\label{definitionofAB2}
A_{1}=\begin{pmatrix}
\Xi & 0\\
-Q & \Xi
\end{pmatrix}
,\quad
A_{2}=\begin{pmatrix}
P & 0\\
W & P
\end{pmatrix}
,\quad
B_{1}=\begin{pmatrix}
a & 0\\
b & a
\end{pmatrix}
,\quad
B_{2}=\begin{pmatrix}
c & 0\\
d & c
\end{pmatrix}
,
\end{equation}
with $a=\Psi$, $b=-2\Xi-P$, $c=-2\Psi+Q$, $d=J$. By direct computations, one arrives at the following result:
\begin{proposition}
The system of axisymmetric MHD--Zeitlin equations \eqref{LP3Dquant} is given by the isospectral flow \eqref{isospectrHamilt} with the matrices $A$ and $B(A)$ defined in \eqref{definitionofAB1}-\eqref{definitionofAB2}.
\end{proposition}
By applying the scheme \eqref{LPintegrMilo} to the matrices \eqref{definitionofAB1}-\eqref{definitionofAB2}, one gets the following integrator for equations \eqref{LP3Dquant}:
\begin{subequations}
  \label{integrationMethodFor3DMHD}
  \begin{align}
    \begin{split}
      W_{n}
        &= \begin{aligned}[t]
            \widetilde W&-\frac{h}{2}[\widetilde W,\widetilde a]-\frac{h}{2}[\widetilde P,\widetilde b]-\frac{h}{2}[\widetilde \Xi,\widetilde d]-\frac{h}{2}[\widetilde c,\widetilde Q]
           \\&
             -\frac{h^{2}}{4}\left(\widetilde d\,\widetilde\Xi\,\widetilde a-\widetilde c\,\widetilde Q\,\widetilde a+\widetilde c\,\widetilde \Xi\,\widetilde b+\widetilde b\,\widetilde P\,\widetilde a+\widetilde a\,\widetilde W\,\widetilde a+\widetilde a\,\widetilde P\,\widetilde b+\widetilde b\,\widetilde \Xi\,\widetilde c-\widetilde a\,\widetilde Q\,\widetilde c+\widetilde a\,\widetilde \Xi\,\widetilde d\right)
           \end{aligned}
    \end{split}
    \label{eqWn3DMHD}
    \\
    \begin{split}
      P_{n}
        &= \widetilde P-\frac{h}{2}[\widetilde P,\widetilde a]-\frac{h}{2}[\widetilde\Xi,\widetilde c]-\frac{h^{2}}{4}\left(\widetilde a\,\widetilde P\,\widetilde a+\widetilde c\,\widetilde \Xi\,\widetilde a+\widetilde a\,\widetilde \Xi\,\widetilde c\right)
    \end{split}
    \label{eqPn3DMHD} 
    \\
    \begin{split}
      Q_{n}
        &= \widetilde Q-\frac{h}{2}[\widetilde Q,\widetilde a]-\frac{h}{2}[\widetilde b,\widetilde\Xi]-\frac{h^{2}}{4}\left(\widetilde a\,\widetilde Q\,\widetilde a-\widetilde b\,\widetilde \Xi\,\widetilde a-\widetilde a\,\widetilde \Xi\,\widetilde b\right)
    \end{split}
    \label{eqQn3DMHD}
    \\
    \begin{split}
      \Xi_{n}
        &= \widetilde \Xi-\frac{h}{2}[\widetilde \Xi,\widetilde a]-\frac{h^{2}}{4}\widetilde a\,\widetilde \Xi\,\widetilde a,
    \end{split}
    \label{eqSigman3DMHD}
    \\
    \begin{split}
      W_{n+1}
        &= W_{n}+h[\widetilde W,\widetilde a]+h[\widetilde P,\widetilde b]+h[\widetilde\Xi,\widetilde d]+h[\widetilde c,\widetilde Q],
    \end{split}
    \label{eqWnplus13DMHD}
    \\
    \begin{split}
      P_{n+1}
        &= P_{n}+h[\widetilde P,\widetilde a]+h[\widetilde \Xi,\widetilde c],
    \end{split}
    \label{eqPnplus13DMHD}
    \\
    \begin{split}
      Q_{n+1}
        &= Q_{n}+h[\widetilde Q,\widetilde a]+h[\widetilde b,\widetilde\Xi],
    \end{split}
    \label{eqQnplus13DMHD}
    \\
    \begin{split}
      \Xi_{n+1}
        &= \Xi_{n}+h[\widetilde \Xi,\widetilde a],
    \end{split}
    \label{eqSigmanplus13DMHD}
  \end{align}
\end{subequations}
where $\widetilde a=\widetilde\Psi$, $\widetilde b=-2\widetilde\Xi-\widetilde P$, $\widetilde c=-2\widetilde\Psi+\widetilde Q$, $\widetilde d=\Delta_{N}\widetilde\Xi$, and $\widetilde\Psi=\Delta_{N}^{-1}(\widetilde W)$.

One can conclude that the integrator given by \eqref{integrationMethodFor3DMHD} constitutes a Lie--Poisson integrator for \eqref{LP3Dquant} provided that the conditions of Theorem~\ref{thmMiloKlas} are satisfied. Namely, that the Lie subalgebra $\mathfrak{f}\subset\mathfrak{gl}(4N,\mathbb{C})$ is reductive and $\mathcal{J}$-quadratic for some $\mathcal{J}\in\mathfrak{gl}(4N,\mathbb{C})$.
\begin{lemma}
\label{lemJquadr}
Let $\mathfrak{g}\subset\mathfrak{gl}(N,\mathbb{C})$ be $\mathcal{J}$-quadratic. Then, a Lie algebra $\widetilde{\mathfrak{g}}\subset\mathfrak{gl}(2N,\mathbb{C})$ consisting of block matrices of the form \eqref{definitionofAB2} is a Lie subalgebra of a $\widetilde{\mathcal{J}}$-quadratic Lie algebra with 
\begin{equation}
\label{defofwidetildeJ}
\widetilde{\mathcal{J}}=\begin{pmatrix}
0 & \mathcal{J}\\
\mathcal{J} & 0
\end{pmatrix}
.
\end{equation}
\end{lemma}
\begin{proof}
Let $c,d\in\mathfrak{g}$, and let $B\in\widetilde{\mathfrak{g}}$, i.e.
\begin{equation*}
B=\begin{pmatrix}
c & 0\\
d & c
\end{pmatrix}
.
\end{equation*}
First, we observe that $\widetilde{\mathcal{J}}^{2}=cI_{2N}$, because $\mathcal{J}^{2}=cI_{N}$. We need to check that $B^{\dag}\widetilde{\mathcal{J}}+\widetilde{\mathcal{J}}B=0$. 
\begin{equation*}
B^{\dag}\widetilde{\mathcal{J}}+\widetilde{\mathcal{J}}B=\begin{pmatrix}
c^{\dag} & d^{\dag}\\
0 & c^{\dag}
\end{pmatrix}
\begin{pmatrix}
0 & \mathcal{J}\\
\mathcal{J} & 0
\end{pmatrix}+\begin{pmatrix}
0 & \mathcal{J}\\
\mathcal{J} & 0
\end{pmatrix}
\begin{pmatrix}
c & 0\\
d & c
\end{pmatrix}
=\begin{pmatrix}
d^{\dag}\mathcal{J}+\mathcal{J}d & c^{\dag}\mathcal{J}+\mathcal{J}c\\
c^{\dag}\mathcal{J}+\mathcal{J}c & 0
\end{pmatrix}
=0,
\end{equation*}
which follows from \eqref{JquadraticConstr}, because $c,d$ are elements of a $\mathcal{J}$-quadratic Lie algebra by assumption.
\end{proof}
Clearly, $\mathfrak{g}=\mathfrak{su}(N)$ is $\mathcal{J}$-quadratic for $\mathcal{J}=I_{N}$. However, we note that the Lie algebra $\widetilde{\mathfrak{g}}\subset\mathfrak{gl}(2N,\mathbb{C})$ in Lemma~\ref{lemJquadr} is not defined by the $\widetilde{\mathcal{J}}$-quadratic constraint. Indeed, the $\widetilde{\mathcal{J}}$-quadratic Lie algebra with $\widetilde{\mathcal{J}}$ defined in \eqref{defofwidetildeJ} contains also matrices of the form \eqref{definitionofAB2} with the upper right block not being trivial. In other words, only a direct implication $\Rightarrow$ in \eqref{JquadraticConstr} is fulfilled, but not the opposite one $\Leftarrow$, in general.
\begin{proposition}
\label{proptildetildeJquadr}
Let $\mathfrak{g}\subset\mathfrak{gl}(N,\mathbb{C})$ be $\mathcal{J}$-quadratic. Then, a Lie algebra $\mathfrak{f}\subset\mathfrak{gl}(4N,\mathbb{C})$ consisting of block matrices of the form \eqref{definitionofAB1}-\eqref{definitionofAB2} is a Lie subalgebra of a $\mathbf{J}$-quadratic Lie algebra with 
\begin{equation*}
\mathbf{J}=\begin{pmatrix}
0 & \widetilde{\mathcal{J}}\\
\widetilde{\mathcal{J}} & 0
\end{pmatrix}
,
\end{equation*}
with $\widetilde{\mathcal{J}}$ defined in \eqref{defofwidetildeJ}.
\end{proposition}
The proof of Proposition~\ref{proptildetildeJquadr} is analogous to the proof of Lemma~\ref{lemJquadr}, and the results of Lemma~\ref{lemJquadr} are used for the corresponding blocks in the matrices in $\mathfrak{f}$, in analogy with the blocks $c,d$. 

By the same arguments, the Lie subalgebra $\mathfrak{f}\subset\mathfrak{gl}(4N,\mathbb{C})$ is not defined by the $\mathbf{J}$-quadratic constraint, but rather is a subalgebra of a larger $\mathbf{J}$-quadratic Lie algebra. This reasoning is already sufficient to state that the results of Theorem~\ref{thmMiloKlas} cannot be used to justify the Lie--Poisson property of the numerical scheme \eqref{integrationMethodFor3DMHD} when applied to \eqref{LP3Dquant}. In particular, it does not follow from Theorem~\ref{thmMiloKlas} that the integrator \eqref{integrationMethodFor3DMHD} preserves the specific structure of matrices defined in \eqref{definitionofAB1}-\eqref{definitionofAB2}. Furthermore, the Lie subalgebra $\mathfrak{f}\subset\mathfrak{gl}(4N,\mathbb{C})$ is not reductive, i.e., the condition $[\mathfrak{f}^{\dag},\mathfrak{f}]\subset\mathfrak{f}$ does not hold.

Even if the conditions for the Lie algebra $\mathfrak{f}\subset\mathfrak{gl}(N,\mathbb{C})$ in Theorem~\ref{thmMiloKlas} to be reductive and $\mathcal{J}$-quadratic, happen to be sufficient but not necessary, and the integrator \eqref{integrationMethodFor3DMHD} still preserves the structure defined by equations \eqref{definitionofAB1}-\eqref{definitionofAB2}, it is not clear how the isospectral property \eqref{isospectralityMiloKlas} of the scheme \eqref{LPintegrMilo} implies the preservation of Casimirs $\mathcal{I}^{N}_{m}$, $J_{m}^{N}$, and $K_{m}^{N}$. Indeed, taking the traces of powers of the matrix $A$ defined by \eqref{definitionofAB1}-\eqref{definitionofAB2}, we deduce the preservation of the quantities $\mathrm{tr}(\Xi^{m})$, and therefore of the spectrum of the matrix $\Xi$, but not the preservation of other Casimirs. This discussion drives us to the following conclusion: the results of Theorem~\ref{thmMiloKlas} and the derivation of the scheme \eqref{integrationMethodFor3DMHD} from \eqref{LPintegrMilo} cannot be used to prove that \eqref{integrationMethodFor3DMHD} is a Lie--Poisson integrator for \eqref{LP3Dquant}. 

The conservation properties of the scheme \eqref{integrationMethodFor3DMHD} can potentially be explained in the framework of a more recent theory developed in \cite{VivianiStern2026}. Namely, the work \cite{VivianiStern2026} addresses the problem of reduction of symplectic integrators on a space $Y$ to Lie--Poisson integrators on a space $Z$ under the action of a quadratic map $F\colon Y\to Z$. Computations along the same lines as in Section 4.1 of \cite{VivianiStern2026} seem promising in proving the Lie--Poisson property of the scheme \eqref{integrationMethodFor3DMHD}, but remain outside of the scope of the present paper. Alternatively, in the forthcoming Section~\ref{sectdiscrLPreduction} we derive a different integrator for \eqref{LP3Dquant}, and its derivation lies outside the setting developed in \cite{VivianiStern2026}.

As we see further, numerical simulations suggest that the method \eqref{integrationMethodFor3DMHD} indeed has all the desired preservation properties listed in Definition~\ref{definitionofLPintegr}.
\subsection{Numerical simulations}
\label{sectnumersim1}
Here, we demonstrate the preservation properties of the integrator \eqref{integrationMethodFor3DMHD}. The integrator \eqref{integrationMethodFor3DMHD} is an implicit scheme and requires finding the intermediate variables $(\widetilde{W},\widetilde{P},\widetilde{Q},\widetilde{\Xi})$. This is achieved by solving equations \eqref{eqWn3DMHD}-\eqref{eqSigman3DMHD} given the state $(W_{n},P_{n},Q_{n},\Xi_{n})$ at the time instance $t_{n}$. Further, the updated state at the moment $t_{n+1}$ is calculated by means of \eqref{eqWnplus13DMHD}-\eqref{eqSigmanplus13DMHD}. The integrator \eqref{integrationMethodFor3DMHD} preserves the Casimirs exactly under the condition that the intermediate state $(\widetilde{W},\widetilde{P},\widetilde{Q},\widetilde{\Xi})$ is found exactly. This is, however, impossible to achieve using the root finding algorithms with a finite tolerance. The scheme was implemented\footnote{Implementation of the methods in Python is available at \url{https://github.com/michaelroop96/axisymmetricMHD.git}} with fixed point iterations, and the tolerance was set to $10^{-14}$. The simulations were run with $N=5$, as the purpose of the present paper is to develop the methods and show their preservation properties, which are independent of the matrix dimension, rather than looking into the actual dynamics. For the actual simulations of turbulence, much higher spatial resolution would be needed. The initial condition $(W_{0},P_{0},Q_{0},\Xi_{0})$ is a quadruple of randomly generated $\mathfrak{su}(N)$ matrices. The final time of simulation is $T=500$, and the step-size is $h=0.01$.

Finding the eigenbasis $T_{l,m}^{N}$ of the Hoppe--Yau Laplacian given by \eqref{definitionTLMwigner} requires calculating the Wigner 3j-symbols, which is computationally costly. Since the operator $\Delta_{N}$ preserves the space of sparse matrices that have only $(\pm m)$-off non-trivial diagonals, one can reformulate the problem of finding the eigenbasis $T_{l,m}^{N}$ as a collection of $N$ tridiagonal $\mathcal{O}(N-|m|)$-dimensional eigenvalue problems, which gives the complexity of finding the entire basis as $\mathcal{O}(N^{2})$, and therefore solving the matrix Poisson equation $\widetilde W=\Delta_{N}\widetilde{\Psi}$, which is needed at every time iteration of the method \eqref{integrationMethodFor3DMHD}, can be performed in $\mathcal{O}(N^{2})$ operations. Since the matrix multiplication requires $\mathcal{O}(N^{3})$ operations, we conclude that the total computational complexity of the method \eqref{integrationMethodFor3DMHD} is $\mathcal{O}(N^{3})$ per time iteration. For detailed analysis, we refer to \cite{CiViMo2023}.

In Fig.~\ref{preservationofspec}, one observes the preservation of the spectrum of the matrix $\Xi$.
\begin{figure}[ht!]
   \includegraphics{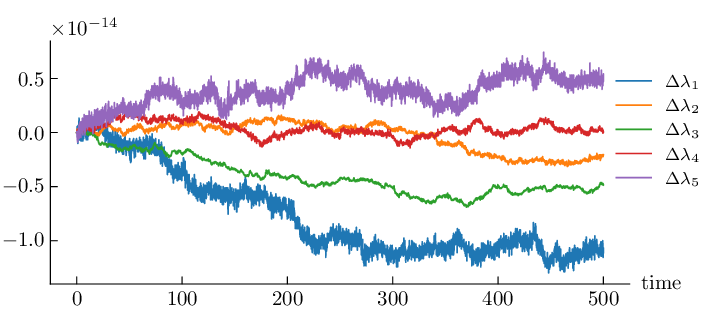}
\caption{Preservation of the eigenvalues $\lambda_{1},\ldots,\lambda_{5}$ of the matrix $\Xi\in\mathfrak{su}(N)$ with $N=5$. The step-size is $h=0.01$. The magnitude $10^{-14}$ of the variation $\Delta\lambda_{i}=\lambda_{i}(t)-\lambda_{i}(0)$ of the eigenvalues shows that the spectrum of $\Xi$ is preserved up to machine precision. 
}
\label{preservationofspec}
\end{figure}

Further, Fig.~\ref{preservationofcasimirs} shows the preservation of the Casimirs $K_{m}^{N}$ and $J_{m}^{N}$. One curious observation is that the Casimirs with an odd number of multipliers under the trace sign are preserved with a somewhat better accuracy.
\begin{figure}[ht!]
   \includegraphics{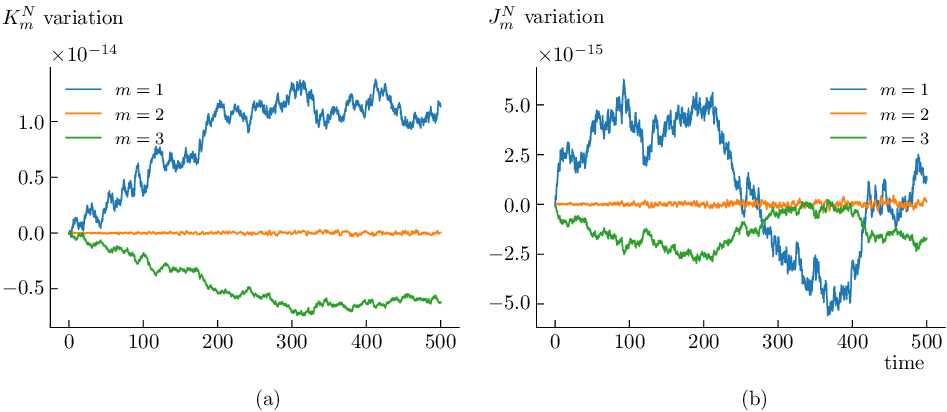}
\caption{Preservation of the Casimirs $K_{m}^{N}$ (a) and $J_{m}^{N}$ (b) for $N=5$ and $m=1,2,3$. The step-size is $h=0.01$. Casimirs $K_{2}^{N}$ and $J_{2}^{N}$ are preserved with a higher accuracy compared to $K_{3}^{N}$ and $J_{3}^{N}$. The magnitude $10^{-14}$ of the variation of the Casimirs shows preservation up to machine precision. 
}
\label{preservationofcasimirs}
\end{figure}

Finally, in Fig.~\ref{preservationofcrosshelandHam}, we show the preservation of the cross-helicity Casimir $\mathcal{I}^{N}_{m}$ and the Hamiltonian $H^{N}$. The cross-helicity $\mathcal{I}^{N}_{m}$ is preserved up to machine precision, as are all the other Casimirs. However, the Hamiltonian function $H^{N}$ is preserved up to the order $10^{-5}$, which indicates only near preservation of the Hamiltonian in the sense of backward error analysis \cite{HaiLubWan}.
\begin{figure}[ht!]
   \includegraphics{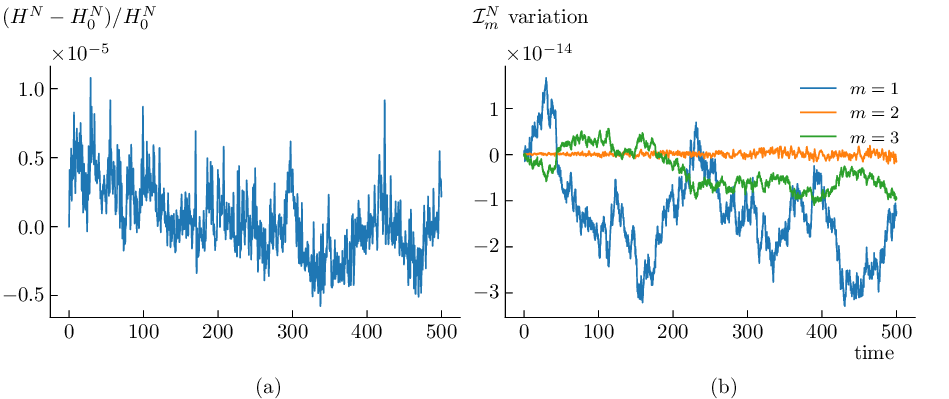}
\caption{Preservation of the Hamiltonian $H^{N}$ (a) and the cross-helicity Casimir $\mathcal{I}^{N}_{m}$ (b) for $N=5$. The step-size is $h=0.01$. The Casimir $\mathcal{I}^{N}_{m}$ is preserved up to machine precision, which is indicated by the variation magnitude $10^{-14}$. The magnitude $10^{-5}$ in the relative error of the Hamiltonian indicates its near preservation.
}
\label{preservationofcrosshelandHam}
\end{figure}
\subsection{Discrete Lie--Poisson reduction}
\label{sectdiscrLPreduction}
In this section, we derive a different numerical integration scheme for equations \eqref{LP3Dquant} that preserves the underlying Lie--Poisson structure. The method emerges from a \textit{discrete Lie--Poisson reduction} developed for Lie algebras of the type  $\mathfrak{F}=(\mathfrak{su}(N)\ltimes\mathfrak{su}(N)^{*})\ltimes(\mathfrak{su}(N)\ltimes\mathfrak{su}(N)^{*})^{*}$. 

Many approaches of constructing Poisson integrators for (Lie-)Poisson systems are based on a connection of a given (Lie-)Poisson system to a certain canonical Hamiltonian system that can be integrated by a symplectic integrator. For Lie--Poisson systems, this connection is established via the \textit{momentum map}. Starting with a $\mathcal{F}$-invariant Hamiltonian system on $T^{*}\mathcal{F}$, one can reduce the dynamics to the dual algebra $\mathfrak{F}^{*}$ of the corresponding Lie group $\mathcal{F}$. The projection map $\mu\colon T^{*}\mathcal{F}\to\mathfrak{F}^{*}\simeq T^{*}\mathcal{F}/\mathcal{F}$ is called the momentum map. The inverse process called \textit{Lie--Poisson reconstruction} allows for mapping the given Lie--Poisson system to a right-invariant Hamiltonian system by means of the momentum map associated with the lifted left action of $\mathcal{F}$ on $T^{*}\mathcal{F}$. Then, a right-invariant symplectic integrator applied to the reconstructed Hamiltonian system descends to a Lie--Poisson integrator in the sense of Definition~\ref{definitionofLPintegr} (see \cite[Theorems VII.5.11, IX.5.7]{HaiLubWan}):
\begin{theorem}
\label{thmdescendLPintegr}
If $\Phi_{h}\colon T^{*}\mathcal{F}\to T^{*}\mathcal{F}$ is a symplectic and right-invariant integrator for the reconstructed Hamiltonian system, then, under the momentum map, it descends to a Lie--Poisson integrator $\phi_{h}\colon \mathfrak{F}^{*}\to \mathfrak{F}^{*}$ for the Lie--Poisson system on $\mathfrak{F}^{*}$.
\end{theorem}
\begin{figure}[ht!]
   \includegraphics[scale=0.8]{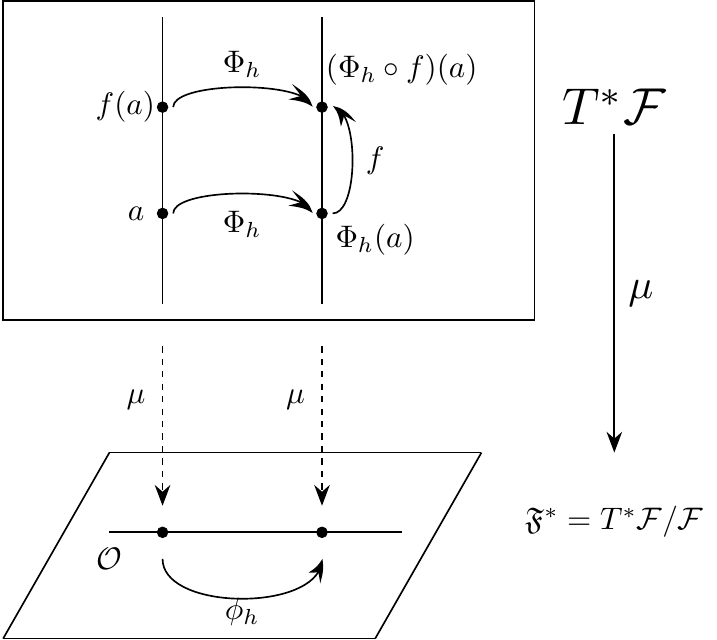}
\caption{Discrete Lie--Poisson reduction. A symplectic right-invariant method $\Phi_{h}\colon T^{*}\mathcal{F}\to T^{*}\mathcal{F}$ descends to a Lie--Poisson integrator $\phi_{h}\colon \mathfrak{F}^{*}\to \mathfrak{F}^{*}$ confined to the coadjoint orbit $\mathcal{O}$.
}
\label{diagrLP}
\end{figure}
This observation, summarized in the diagram Fig.~\ref{diagrLP}, suggests the following algorithm for finding Lie--Poisson integrators for Lie--Poisson systems:
\begin{itemize}
\item for a given $\mathfrak{F}^{*}$, find the momentum map $\mu\colon T^{*}\mathcal{F}\to\mathfrak{F}^{*}$ associated with the left action;
\item from a given $H\in C^{\infty}(\mathfrak{F}^{*})$ construct the right-invariant Hamiltonian $\mathrm{H}=\mu^{*}(H)\in C^{\infty}(T^{*}\mathcal{F})$;
\item apply a symplectic right-invariant integrator $\Phi_{h}\colon T^{*}\mathcal{F}\to T^{*}\mathcal{F}$ to Hamilton's equations on $T^{*}\mathcal{F}$ with the Hamiltonian $\mathrm{H}$;
\item using the momentum map $\mu$, project the method $\Phi_{h}$ to $\phi_{h}$.
\end{itemize}

We recall that the system of matrix equations \eqref{LP3Dquant}, for which we aim to derive the Lie--Poisson integrator, evolves on the dual $\mathfrak{F}^{*}$ of the Lie algebra $\mathfrak{F}=(\mathfrak{su}(N)\ltimes\mathfrak{su}(N)^{*})\ltimes(\mathfrak{su}(N)\ltimes\mathfrak{su}(N)^{*})^{*}$. The corresponding Lie group $\mathcal{F}$ can be identified with $\mathcal{F}=(\mathrm{SU}(N)\ltimes\mathfrak{su}(N)^{*})\ltimes(\mathfrak{su}(N)\ltimes\mathfrak{su}(N)^{*})^{*}$, or, explicitly
\begin{equation*}
\mathcal{F}=\left\{(G,U)=((g,\beta),(u,\sigma))\mid g\in \mathrm{SU}(N),\,\beta,\sigma\in\mathfrak{su}(N)^{*},\,u\in\mathfrak{su}(N)\right\},
\end{equation*}
with the group operation (see Appendix~\ref{appendix2} for details)
\begin{equation*}
(G,U)\cdot(Q,M)=\left((gq,\mathrm{Ad}^{*}_{q}\beta+\alpha),(\mathrm{Ad}_{q^{-1}}u+m,\mathrm{Ad}^{*}_{q}\sigma-\mathrm{ad}^{*}_{\mathrm{Ad}_{q^{-1}}u}\alpha+\gamma)\right),
\end{equation*}
with $(G,U)=((g,\beta),(u,\sigma)),(Q,M)=((q,\alpha),(m,\gamma))\in \mathcal{F}$, and where the identification of $\mathfrak{su}(N)^{*}$ with $\mathfrak{su}(N)$ is done via the Frobenius inner product.

Further, we construct the total space $T^{*}\mathcal{F}$ of of the cotangent bundle of the group $\mathcal{F}$ as
\begin{equation*}
T^{*}\mathcal{F}=\left\{(Q,M,P,A)\mid(Q,M)=((q,\alpha),(m,\gamma))\in \mathcal{F},\,(P,A)=((p,\nu),(a,b))\in T^{*}_{(Q,M)}\mathcal{F}\right\}
\end{equation*}
\begin{theorem}
The cotangent lift $((G,U),(Q,M,P,A))\mapsto(\widetilde Q,\widetilde M,\widetilde P,\widetilde A)$ of the left $\mathcal{F}$-action on $\mathcal{F}$ is given by
\begin{equation}
\label{cotangentlift}
\begin{split}
&\widetilde q=gq,\quad\widetilde\alpha=\mathrm{Ad}^{*}_{q}\beta+\alpha,\quad\widetilde m=\mathrm{Ad}_{q^{-1}}u+m,\quad\widetilde\gamma=\mathrm{Ad}^{*}_{q}\sigma-\mathrm{ad}^{*}_{\mathrm{Ad}_{q^{-1}}u}\alpha+\gamma{}\\&
\widetilde\nu=\nu-\mathrm{ad}_{a}\left(\mathrm{Ad}_{q^{-1}}u\right),\quad \widetilde a=a,\quad\widetilde b=b,{}\\&
\widetilde p=\left(g^{-1}\right)^{\dag}p+\left(g^{-1}\right)^{\dag}\left(q^{-1}\right)^{\dag}\left(\mathrm{ad}^{*}_{\nu-\mathrm{ad}_{a}\left(\mathrm{Ad}_{q^{-1}}u\right)}\left(\mathrm{Ad}^{*}_{q}\beta\right)-\mathrm{ad}^{*}_{\mathrm{Ad}_{q^{-1}}u}\left(\mathrm{ad}^{*}_{a}\alpha+b\right)+\mathrm{ad}^{*}_{a}\left(\mathrm{Ad}^{*}_{q}\sigma\right)\right).
\end{split}
\end{equation}
\end{theorem}
\begin{proof}
See Appendix~\ref{appendix2}.
\end{proof}
From the action \eqref{cotangentlift}, one gets the following expression for the momentum map $\mu\colon T^{*}\mathcal{F}\to\mathfrak{F}^{*},$ (see Appendix~\ref{appendix2}):
\begin{equation}
\label{mommapforsemidir}
\mu((q,\alpha),(m,\gamma),(p,\nu),(a,b))=\left(\frac{1}{2}(pq^{\dag}-qp^{\dag}),q\nu q^{\dag},q(\alpha a-a\alpha+b)q^{\dag},qaq^{\dag}\right).
\end{equation}
We define the fields $(W,Q,P,\Xi)$ as
\begin{equation}
\label{fieldsviacanonicalvar}
W=\frac{1}{2}(pq^{\dag}-qp^{\dag})+2q\nu q^{\dag},\quad Q=q\nu q^{\dag},\quad P=q(\alpha a-a\alpha+b)q^{\dag},\quad \Xi=qaq^{\dag}.
\end{equation}
Then, the reconstructed Hamiltonian $\mathrm{H}^{N}=\mu^{*}(H^{N})$ takes the form
\begin{equation}
\label{reconstrHamilt}
\begin{split}
\mathrm{H}^{N}(q,p,\alpha,\nu,a,b)=-&\frac{1}{2}\mathrm{tr}\left(\frac{1}{2}(pq^{\dag}-qp^{\dag})\Psi\right)-\mathrm{tr}(q\nu q^{\dag}\Psi)-\frac{1}{2}\mathrm{tr}(qaq^{\dag}J)+{}\\&\frac{1}{2}\mathrm{tr}(\nu^{2})+\frac{1}{2}\mathrm{tr}((\alpha a-a\alpha+b)^{2}).
\end{split}
\end{equation}
Taking the variational derivatives of the reconstructed right-invariant Hamiltonian \eqref{reconstrHamilt}
\begin{equation*}
\dot q=-\frac{\delta\mathrm{H}^{N}}{\delta p},\quad \dot p=\frac{\delta\mathrm{H}^{N}}{\delta q},\quad\dot\nu=-\frac{\delta\mathrm{H}^{N}}{\delta \alpha},\quad\dot\alpha=\frac{\delta\mathrm{H}^{N}}{\delta\nu},\quad\dot a=-\frac{\delta\mathrm{H}^{N}}{\delta \gamma}=0,\quad\dot b=-\frac{\delta\mathrm{H}^{N}}{\delta m}=0,
\end{equation*}
we get the reconstructed right-invariant Hamiltonian system
\begin{equation}
\label{reconstructCanHam}
\left\{
\begin{aligned}
\displaystyle\dot q&=-\Psi q, \\
\displaystyle\dot p&=-\Psi p-4\Psi q\nu-2Jqa,\\
\displaystyle\dot\nu&=[a,\alpha a-a\alpha+b],\\
\displaystyle\dot\alpha&=-\nu+2q^{\dag}\Psi q,
\end{aligned}
\right.
\end{equation}
where $\Psi=\Psi(W)=\Psi(q,p,\nu)$, and $J=J(\Xi)=J(q)$.
\begin{theorem}
The right-invariant reconstructed Hamiltonian system \eqref{reconstructCanHam} on $T^{*}\mathcal{F}$ with the Hamiltonian \eqref{reconstrHamilt} reduces to the Lie--Poisson system \eqref{LP3Dquant} on $\mathfrak{F}^{*}$ under the momentum map \eqref{fieldsviacanonicalvar}.
\end{theorem}
\begin{proof}
Direct computations using \eqref{fieldsviacanonicalvar} and the identity $qq^{\dag}=q^{\dag}q=I$.
\end{proof}
\begin{remark}
Strictly speaking, the complete Hamiltonian dynamics \eqref{reconstructCanHam} involves also the evolution of the fields $m$ and $\gamma$. We omit the corresponding equations for $m$ and $\gamma$, because they are not needed to construct the fields $(W,Q,P,\Xi)$ due to the structure of the momentum map \eqref{mommapforsemidir}. Moreover, equations \eqref{reconstructCanHam} are completely decoupled from the evolution of $m$ and $\gamma$, as the right hand side of \eqref{reconstructCanHam} does not depend on either $m$ or $\gamma$. We further observe that the dynamics for the fields $a$ and $b$ is trivial and these matrices serve as parameters. Finally, we conclude that the reconstructed equations \eqref{reconstructCanHam} involve the same number of fields as the reduced system \eqref{LP3Dquant}, which is a typical case for Lie algebras of semidirect product structure (see the analogous discussion in \cite{ModRoop}).
\end{remark}

We will choose the simplest example of a symplectic Runge-Kutta integrator for the system \eqref{reconstructCanHam}, which is the implicit midpoint rule. For an ODE $\dot y=f(y)$, it is formulated as follows:
\begin{equation*}
y_{n}=\widetilde y-\frac{h}{2}f(\tilde y),\quad y_{n+1}=\widetilde y+\frac{h}{2}f(\tilde y),
\end{equation*}
where $\widetilde y$ is an intermediate state found from the first equation and used to update the state by means of the second equation, and $h$ is the given time step. Applying the implicit midpoint rule to the system \eqref{reconstructCanHam} yields
\begin{equation}
\label{reconstructCanHamMidp}
\left\{
\begin{aligned}
\displaystyle q_{n}&=\widetilde q+\frac{h}{2}\widetilde\Psi\widetilde q, \\
\displaystyle\alpha_{n}&=\widetilde\alpha+\frac{h}{2}\widetilde\nu-h\widetilde q^{\dag}\widetilde\Psi\widetilde q,\\
\displaystyle p_{n}&=\widetilde p+\frac{h}{2}\widetilde\Psi\widetilde p+2h\widetilde\Psi\widetilde q\widetilde\nu+h\widetilde J\widetilde q a,\\
\displaystyle \nu_{n}&=\widetilde\nu-\frac{h}{2}[a,\widetilde\alpha a-a\widetilde\alpha+b],
\end{aligned}
\right.
\quad
\left\{
\begin{aligned}
\displaystyle q_{n+1}&=\widetilde q-\frac{h}{2}\widetilde\Psi\widetilde q, \\
\displaystyle\alpha_{n+1}&=\widetilde\alpha-\frac{h}{2}\widetilde\nu+h\widetilde q^{\dag}\widetilde\Psi\widetilde q,\\
\displaystyle p_{n+1}&=\widetilde p-\frac{h}{2}\widetilde\Psi\widetilde p-2h\widetilde\Psi\widetilde q\widetilde\nu-h\widetilde J\widetilde q a,\\
\displaystyle \nu_{n+1}&=\widetilde\nu+\frac{h}{2}[a,\widetilde\alpha a-a\widetilde\alpha+b],
\end{aligned}
\right.
\end{equation}
where $\widetilde\Psi=\Psi(\widetilde q,\widetilde p,\widetilde \nu)$ and $\widetilde J=J(\widetilde q)$.

We arrive at the final central result of the paper.
\begin{theorem}
\label{thmLPinterg2}
The implicit midpoint scheme \eqref{reconstructCanHamMidp} descends to a Lie--Poisson integrator $\phi_{h}\colon\mathfrak{F}^{*}\to\mathfrak{F}^{*}$, $(W_{n+1},Q_{n+1},P_{n+1},\Xi_{n+1})=\phi_{h}(W_{n},Q_{n},P_{n},\Xi_{n})$ defined by the following equations:
\begin{equation}
\label{LPintegrator2}
\left\{
\begin{aligned}
\displaystyle \Xi_{n}&=\left(I+\frac{h}{2}\widetilde\Psi\right)\widetilde\Xi\left(I-\frac{h}{2}\widetilde\Psi\right), \\
\displaystyle Q_{n}&=\left(I+\frac{h}{2}\widetilde\Psi\right)\left(\widetilde Q-\frac{h}{2}\widetilde N\right)\left(I-\frac{h}{2}\widetilde\Psi\right),\\
\displaystyle P_{n}&=\left(I+\frac{h}{2}\widetilde\Psi\right)\left(\widetilde P+\frac{h}{2}\widetilde M\right)\left(I-\frac{h}{2}\widetilde\Psi\right)-h\left(I-\frac{h}{2}\widetilde\Psi\right)^{-1}\widetilde\Psi\widetilde\Xi\left(I-\frac{h}{2}\widetilde\Psi\right){}\\&+h\left(I+\frac{h}{2}\widetilde\Psi\right)\widetilde\Xi\widetilde\Psi\left(I+\frac{h}{2}\widetilde\Psi\right)^{-1},\\
\displaystyle W_{n}&=\left(I+\frac{h}{2}\widetilde\Psi\right)\left(\widetilde W-h\widetilde N\right)\left(I-\frac{h}{2}\widetilde\Psi\right)-\frac{h}{2}[\widetilde\Xi,\widetilde J]-\frac{h}{2}[2\widetilde Q,\widetilde\Psi]{}\\&-\frac{h^{2}}{4}\left(\widetilde J\widetilde\Xi\widetilde\Psi+\widetilde\Psi\widetilde\Xi\widetilde J+4\widetilde\Psi\widetilde Q\widetilde\Psi\right),
\end{aligned}
\right.
\end{equation}
where $\widetilde M$ and $\widetilde N$ are
\begin{equation}
\label{MNtildes}
\begin{split}
&\widetilde N=[\widetilde\Xi,\widetilde P]+\frac{h^{2}}{4}\left(\widetilde P\widetilde\Psi^{2}\widetilde\Xi-\widetilde \Xi\widetilde\Psi^{2}\widetilde P\right){}\\&
\widetilde M=[\widetilde Q,\widetilde\Xi]+\frac{h^{2}}{4}\left(\widetilde \Xi\widetilde\Psi^{2}\widetilde Q-\widetilde  Q\widetilde\Psi^{2}\widetilde\Xi\right),
\end{split}
\end{equation}
and the updated fields $(W_{n+1},Q_{n+1},P_{n+1},\Xi_{n+1})$ are obtained by changing $+h$ to $-h$ in \eqref{LPintegrator2}.

The integrator \eqref{LPintegrator2}-\eqref{MNtildes} preserves the coadjoint orbits of $\mathfrak{F}^{*}$, and, in particular,  Casimirs \eqref{casimirsquantizedmonomials}, and is a symplectic map on the coadjoint orbits.
\end{theorem}

To prove Theorem~\ref{thmLPinterg2}, we will need the following intermediate result.
\begin{lemma}
Let $\widetilde q$ be a solution of \eqref{reconstructCanHamMidp}. Then, the quantities $\mathcal{A}=\widetilde q^{\dag}\widetilde q$ and $\widetilde{\mathcal{A}}=\widetilde q\widetilde q^{\dag}$ fulfill the following equations:
\begin{equation}
\label{formulasAandAtilde}
\mathcal{A}=I+\frac{h^{2}}{4}\widetilde q^{\dag}\widetilde\Psi^{2}\widetilde q,\quad
\left(I+\frac{h}{2}\widetilde\Psi\right)\widetilde{\mathcal{A}}\left(I-\frac{h}{2}\widetilde\Psi\right)=I.
\end{equation}
\end{lemma}
\begin{proof}
Since \eqref{reconstructCanHamMidp} is the implicit midpoint method applied to a Hamiltonian system, then it preserves quadratic invariants as any symplectic Runge-Kutta scheme. It means that $q_{n}q_{n}^{\dag}=q_{n}^{\dag}q_{n}=I$. From \eqref{reconstructCanHamMidp}, we get
\begin{equation*}
I=q_{n}^{\dag}q_{n}=\left(\widetilde q^{\dag}-\frac{h}{2}\widetilde q^{\dag}\widetilde\Psi\right)\left(\widetilde q+\frac{h}{2}\widetilde\Psi\widetilde q\right)=\widetilde q^{\dag}\widetilde q-\frac{h^{2}}{4}\widetilde q^{\dag}\widetilde\Psi^{2}\widetilde q=\mathcal{A}-\frac{h^{2}}{4}\widetilde q^{\dag}\widetilde\Psi^{2}\widetilde q,
\end{equation*}
and the first identity in \eqref{formulasAandAtilde} follows. Further,
\begin{equation*}
I=q_{n}q_{n}^{\dag}=\left(\widetilde q+\frac{h}{2}\widetilde\Psi\widetilde q\right)\left(\widetilde q^{\dag}-\frac{h}{2}\widetilde q^{\dag}\widetilde\Psi\right)=\left(I+\frac{h}{2}\widetilde\Psi\right)\widetilde q\widetilde q^{\dag}\left(I-\frac{h}{2}\widetilde\Psi\right)=\left(I+\frac{h}{2}\widetilde\Psi\right)\widetilde{\mathcal{A}}\left(I-\frac{h}{2}\widetilde\Psi\right),
\end{equation*}
and the second equality in \eqref{formulasAandAtilde} follows.
\end{proof}
This result reflects the fact that the implicit midpoint scheme \eqref{reconstructCanHamMidp} is not intrinsically formulated on the underlying Lie group, and calculations of the vector fields outside the Lie group may occur.

We proceed with the proof of Theorem~\ref{thmLPinterg2}.
\begin{proof}[Proof of Theorem~\ref{thmLPinterg2}]
We begin with the matrix $\Xi_{n}$.
\begin{equation*}
\begin{split}
\Xi_{n}&=q_{n}a q_{n}^{\dag}=\left(\widetilde q+\frac{h}{2}\widetilde\Psi\widetilde q\right)a\left(\widetilde q^{\dag}-\frac{h}{2}\widetilde q^{\dag}\widetilde\Psi\right)=\left(I+\frac{h}{2}\widetilde\Psi\right)\widetilde qa\widetilde q^{\dag}\left(I-\frac{h}{2}\widetilde\Psi\right)=\left(I+\frac{h}{2}\widetilde\Psi\right)\widetilde\Xi\left(I-\frac{h}{2}\widetilde\Psi\right),
\end{split}
\end{equation*}
where $\widetilde\Xi=qa\widetilde q^{\dag}$.

Further, for the matrix $Q_{n}$, we get
\begin{equation*}
\begin{split}
Q_{n}&=q_{n}\nu_{n}q_{n}^{\dag}=\left(\widetilde q+\frac{h}{2}\widetilde\Psi\widetilde q\right)\left(\widetilde\nu-\frac{h}{2}[a,\widetilde\alpha a-a\widetilde\alpha+b]\right)\left(\widetilde q^{\dag}-\frac{h}{2}\widetilde q^{\dag}\widetilde\Psi\right){}\\&=\left(I+\frac{h}{2}\widetilde\Psi\right)\widetilde q\widetilde\nu\widetilde q^{\dag}\left(I-\frac{h}{2}\widetilde\Psi\right)-\frac{h}{2}\left(I+\frac{h}{2}\widetilde\Psi\right)\widetilde q[a,\widetilde\alpha a-a\widetilde\alpha+b]\widetilde q^{\dag}\left(I-\frac{h}{2}\widetilde\Psi\right){}\\&=\left(I+\frac{h}{2}\widetilde\Psi\right)\left(\widetilde Q-\frac{h}{2}\widetilde N\right)\left(I-\frac{h}{2}\widetilde\Psi\right),
\end{split}
\end{equation*}
where $\widetilde Q=\widetilde q\widetilde\nu\widetilde q^{\dag}$ and $\widetilde N=q[a,\widetilde\alpha a-a\widetilde\alpha+b]\widetilde q^{\dag}$. 

Let us define $\widetilde P=\widetilde q(\widetilde\alpha a-a\widetilde\alpha+b)\widetilde q^{\dag}$, and find
\begin{equation*}
[\widetilde\Xi,\widetilde P]=\widetilde{q}a\widetilde q^{\dag}\widetilde q(\widetilde\alpha a-a\widetilde\alpha+b)\widetilde q^{\dag}-\widetilde q(\widetilde\alpha a-a\widetilde\alpha+b)\widetilde q^{\dag}\widetilde{q}a\widetilde q^{\dag}=\widetilde q(a\mathcal{A}(\widetilde\alpha a-a\widetilde\alpha+b)-(\widetilde\alpha a-a\widetilde\alpha+b)\mathcal{A}a)\widetilde q^{\dag}.
\end{equation*}
Further, using \eqref{formulasAandAtilde} for the matrix $\mathcal{A}$, we get
\begin{equation*}
\begin{split}
[\widetilde\Xi,\widetilde P]&=\widetilde q\left(a\left(I+\frac{h^{2}}{4}\widetilde q^{\dag}\widetilde\Psi^{2}\widetilde q\right)(\widetilde\alpha a-a\widetilde\alpha+b)-(\widetilde\alpha a-a\widetilde\alpha+b)\left(I+\frac{h^{2}}{4}\widetilde q^{\dag}\widetilde\Psi^{2}\widetilde q\right)a\right)\widetilde q^{\dag}{}\\&=
\widetilde N+\frac{h^{2}}{4}\widetilde\Xi\widetilde\Psi^{2}\widetilde P-\frac{h^{2}}{4}\widetilde P\widetilde\Psi^{2}\widetilde\Xi,
\end{split}
\end{equation*}
which implies \eqref{MNtildes}.

For the matrix $P_{n}$, we get
\begin{equation}
\label{intermediateforPn}
\begin{split}
P_{n}&=q_{n}(\alpha_{n}a-a\alpha_{n}+b)q_{n}^{\dag}=\left(\widetilde q+\frac{h}{2}\widetilde\Psi\widetilde q\right)\left(\left[\widetilde\alpha+\frac{h}{2}\widetilde\nu-h\widetilde q^{\dag}\widetilde\Psi\widetilde q,a\right]+b\right)\left(\widetilde q^{\dag}-\frac{h}{2}\widetilde q^{\dag}\widetilde\Psi\right){}\\&=
\left(I+\frac{h}{2}\widetilde\Psi\right)\widetilde q([\widetilde\alpha,a]+b)\widetilde q^{\dag}\left(I-\frac{h}{2}\widetilde\Psi\right)
+\frac{h}{2}\left(I+\frac{h}{2}\widetilde\Psi\right)\widetilde q[\widetilde\nu,a]\widetilde q^{\dag}\left(I-\frac{h}{2}\widetilde\Psi\right){}\\&
-h\left(I+\frac{h}{2}\widetilde\Psi\right)\widetilde q[\widetilde q^{\dag}\widetilde\Psi\widetilde q,a]\widetilde q^{\dag}\left(I-\frac{h}{2}\widetilde\Psi\right){}\\&
=\left(I+\frac{h}{2}\widetilde\Psi\right)\left(\widetilde P+\frac{h}{2}\widetilde M\right)\left(I-\frac{h}{2}\widetilde\Psi\right)-h\left(I+\frac{h}{2}\widetilde\Psi\right)\widetilde B\left(I-\frac{h}{2}\widetilde\Psi\right),
\end{split}
\end{equation}
where $\widetilde M=\widetilde q[\widetilde\nu,a]\widetilde q^{\dag}$, and $\widetilde B=\widetilde q[\widetilde q^{\dag}\widetilde\Psi\widetilde q,a]\widetilde q^{\dag}$.

We need to eliminate $\widetilde q$ and $\widetilde q^{\dag}$ from the expressions for $\widetilde M$ and $\widetilde B$. To that end, consider
\begin{equation*}
\begin{split}
[\widetilde\Xi,\widetilde Q]&=\widetilde qa\widetilde q^{\dag}\widetilde q\widetilde\nu\widetilde q^{\dag}-\widetilde q\widetilde\nu\widetilde q^{\dag}\widetilde qa\widetilde q^{\dag}=\widetilde q(a\mathcal{A}\widetilde\nu-\widetilde\nu\mathcal{A}a)\widetilde q^{\dag}{}\\&=\widetilde q\left(a\left(I+\frac{h^{2}}{4}\widetilde q^{\dag}\widetilde\Psi^{2}\widetilde q\right)\widetilde\nu-\widetilde\nu\left(I+\frac{h^{2}}{4}\widetilde q^{\dag}\widetilde\Psi^{2}\widetilde q\right)a\right)\widetilde q^{\dag}=
-\widetilde M+\frac{h^{2}}{4}\widetilde\Xi\widetilde\Psi^{2}\widetilde Q-\frac{h^{2}}{4}\widetilde Q\widetilde\Psi^{2}\widetilde\Xi,
\end{split}
\end{equation*}
where we used \eqref{formulasAandAtilde} for the matrix $\mathcal{A}$, and which implies \eqref{MNtildes}.

Further, for the matrix $\widetilde B$, we have
\begin{equation*}
\widetilde B=\widetilde q[\widetilde q^{\dag}\widetilde\Psi\widetilde q,a]\widetilde q^{\dag}=\widetilde q\widetilde q^{\dag}\widetilde\Psi\widetilde qa\widetilde q^{\dag}-\widetilde qa\widetilde q^{\dag}\widetilde\Psi\widetilde q\widetilde q^{\dag}=\widetilde{\mathcal{A}}\widetilde\Psi\widetilde\Xi-\widetilde\Xi\widetilde\Psi\widetilde{\mathcal{A}}.
\end{equation*}
Using \eqref{MNtildes} for $\widetilde{\mathcal{A}}$, we get
\begin{equation*}
\begin{split}
\left(I+\frac{h}{2}\widetilde\Psi\right)\widetilde B\left(I-\frac{h}{2}\widetilde\Psi\right)&=\left(I+\frac{h}{2}\widetilde\Psi\right)\widetilde{\mathcal{A}}\widetilde\Psi\widetilde\Xi\left(I-\frac{h}{2}\widetilde\Psi\right)-\left(I+\frac{h}{2}\widetilde\Psi\right)\widetilde\Xi\widetilde\Psi\widetilde{\mathcal{A}}\left(I-\frac{h}{2}\widetilde\Psi\right){}\\&=\left(I-\frac{h}{2}\widetilde\Psi\right)^{-1}\widetilde\Psi\widetilde\Xi\left(I-\frac{h}{2}\widetilde\Psi\right)-\left(I+\frac{h}{2}\widetilde\Psi\right)\widetilde\Xi\widetilde\Psi\left(I+\frac{h}{2}\widetilde\Psi\right)^{-1},
\end{split}
\end{equation*}
and equation \eqref{intermediateforPn} finally becomes
\begin{equation*}
\begin{split}
P_{n}&=\left(I+\frac{h}{2}\widetilde\Psi\right)\left(\widetilde P+\frac{h}{2}\widetilde M\right)\left(I-\frac{h}{2}\widetilde\Psi\right)-h\left(I-\frac{h}{2}\widetilde\Psi\right)^{-1}\widetilde\Psi\widetilde\Xi\left(I-\frac{h}{2}\widetilde\Psi\right){}\\&+h\left(I+\frac{h}{2}\widetilde\Psi\right)\widetilde\Xi\widetilde\Psi\left(I+\frac{h}{2}\widetilde\Psi\right)^{-1}
\end{split}
\end{equation*}

For the matrix $W_{n}$, we get
\begin{equation*}
\begin{split}
W_{n}&=\frac{1}{2}(p_{n}q_{n}^{\dag}-q_{n}p_{n}^{\dag})+2q_{n}\nu_{n}q_{n}^{\dag}{}\\&=\frac{1}{2}(\widetilde p\widetilde q^{\dag}-\widetilde q\widetilde p^{\dag})-\frac{h}{2}[\widetilde W,\widetilde\Psi]-\frac{h}{2}[\widetilde\Xi,\widetilde J]-\frac{h^{2}}{4}(\widetilde\Psi\widetilde W\widetilde\Psi+\widetilde J\widetilde\Xi\widetilde\Psi+\widetilde \Psi\widetilde\Xi\widetilde J+2\widetilde\Psi\widetilde Q\widetilde\Psi){}\\&+\left(I+\frac{h}{2}\widetilde\Psi\right)\left(2\widetilde Q-h\widetilde N\right)\left(I-\frac{h}{2}\widetilde\Psi\right){}\\&=\left(I+\frac{h}{2}\widetilde\Psi\right)\left(\widetilde W-h\widetilde N\right)\left(I-\frac{h}{2}\widetilde\Psi\right)-\frac{h}{2}[\widetilde\Xi,\widetilde J]-\frac{h}{2}[2\widetilde Q,\widetilde\Psi]-\frac{h^{2}}{4}\left(\widetilde J\widetilde\Xi\widetilde\Psi+\widetilde\Psi\widetilde\Xi\widetilde J+4\widetilde\Psi\widetilde Q\widetilde\Psi\right),
\end{split}
\end{equation*}
where $\widetilde W=(\widetilde p\widetilde q^{\dag}-\widetilde q\widetilde p^{\dag})/2+2\widetilde q\widetilde\nu\widetilde q^{\dag}$.

Finally, expressions for $(W_{n+1},Q_{n+1},P_{n+1},\Xi_{n+1})$ are obtained from those for $(W_{n},Q_{n},P_{n},\Xi_{n})$ by changing $+h$ to $-h$, which follows from the definition of the implicit midpoint scheme.

The Lie--Poisson property of the scheme follows directly from Theorem~\ref{thmdescendLPintegr}.
\end{proof}

The method \eqref{LPintegrator2} obtained as a Lie--Poisson reduced implicit midpoint scheme is clearly different from \eqref{integrationMethodFor3DMHD}. Firstly, it contains $\mathcal{O}(h^{3})$ terms that are not present in \eqref{integrationMethodFor3DMHD}, and secondly, it involves matrix inverse operations. This situation is different from integration of the matrix equations for the two-field MHD \eqref{qMHD}. Namely, the methods obtained by applying the isospectral integrator \eqref{LPintegrMilo} to the respective block matrix equations equivalent to \eqref{qMHD} and by the Lie--Poisson reduction of the implicit midpoint rule, are identical \cite{ModRoop}. The explanation of this fact comes from the quadratic structure of the momentum map, which makes the Lie--Poisson reduction and a more general reduction by quadratic maps developed in \cite{VivianiStern2026} equivalent. Here, the momentum map \eqref{mommapforsemidir} is cubic, which puts the results developed here outside of the framework developed in \cite{VivianiStern2026}. Furthermore, as already noted, semidirect products have a remarkable feature that the reconstructed canonical equations evolve on a submanifold of $T^{*}\mathcal{F}$ of the same dimension as the corresponding semidirect product Lie algebra $\mathfrak{F}$. In other words, the reconstructed equations operate with the same number of fields as the original reduced equations on $\mathfrak{F}^{*}$. It means that in the context of time integrators there is no additional computational price that needed to be paid to integrate the reconstructed equations compared to the original reduced equations, and therefore the possibility to reduce the integration method to $\mathfrak{F}^{*}$ is not critical. In general, reconstructed equations involve twice as many fields as reduced equations, and therefore it is important that integrators on $T^{*}\mathcal{F}$ explicitly descend to integrators on $\mathfrak{F}^{*}$.

\subsection{Numerical simulations}
Here, we show the preservation properties of the method \eqref{LPintegrator2}-\eqref{MNtildes}. The setup is the same as in the previous section~\ref{sectnumersim1}.

In Fig.~\ref{preservationofspec2}, one observes the preservation of the spectrum of the matrix $\Xi$.
\begin{figure}[ht!]
   \includegraphics{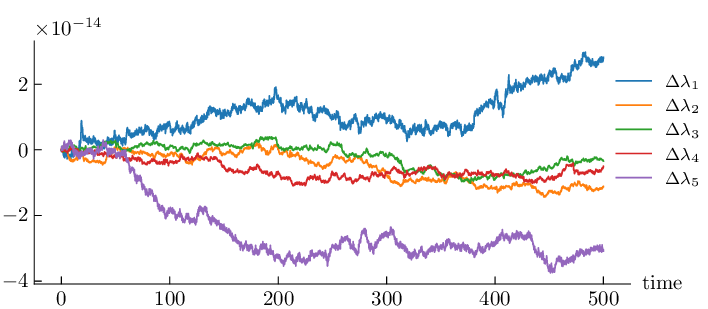}
\caption{Preservation of the eigenvalues $\lambda_{1},\ldots,\lambda_{5}$ of the matrix $\Xi\in\mathfrak{su}(N)$ with $N=5$. The step-size is $h=0.01$. The magnitude $10^{-14}$ of the variation $\Delta\lambda_{i}=\lambda_{i}(t)-\lambda_{i}(0)$ of the eigenvalues shows that the spectrum of $\Xi$ is preserved up to machine precision. 
}
\label{preservationofspec2}
\end{figure}

Further, Fig.~\ref{preservationofcasimirs2} shows the preservation of the Casimirs $K_{m}^{N}$ and $J_{m}^{N}$. One observes that the Casimirs with an odd number of multipliers under the trace sign are preserved with a better accuracy, as in the simulations with the method \eqref{integrationMethodFor3DMHD}.
\begin{figure}[ht!]
   \includegraphics{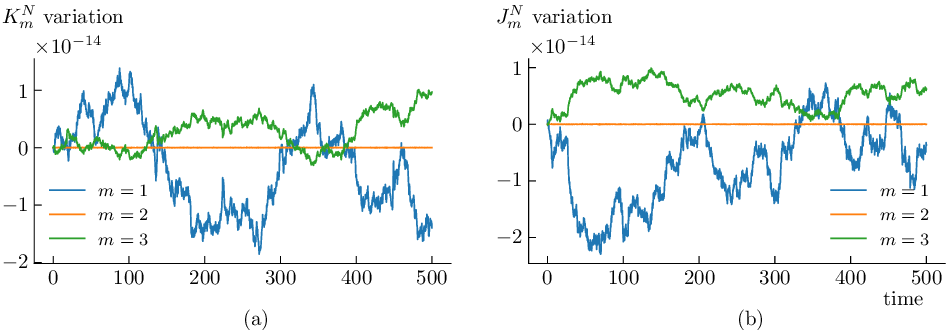}
\caption{Preservation of the Casimirs $K_{m}^{N}$ (a) and $J_{m}^{N}$ (b) for $N=5$ and $m=1,2,3$. The step-size is $h=0.01$. Casimirs $K_{2}^{N}$ and $J_{2}^{N}$ are preserved with a higher accuracy compared to $K_{3}^{N}$ and $J_{3}^{N}$. The magnitude $10^{-14}$ of the variation of the Casimirs shows preservation up to machine precision. 
}
\label{preservationofcasimirs2}
\end{figure}

Finally, in Fig.~\ref{preservationofcrosshelandHam2}, we show the preservation of the cross-helicity Casimir $\mathcal{I}^{N}_{m}$ and the Hamiltonian $H^{N}$. The cross-helicity $\mathcal{I}^{N}_{m}$ is preserved up to machine precision, as all the other Casimirs. The Hamiltonian function $H^{N}$ is preserved up to the order $10^{-5}$.
\begin{figure}[ht!]
   \includegraphics{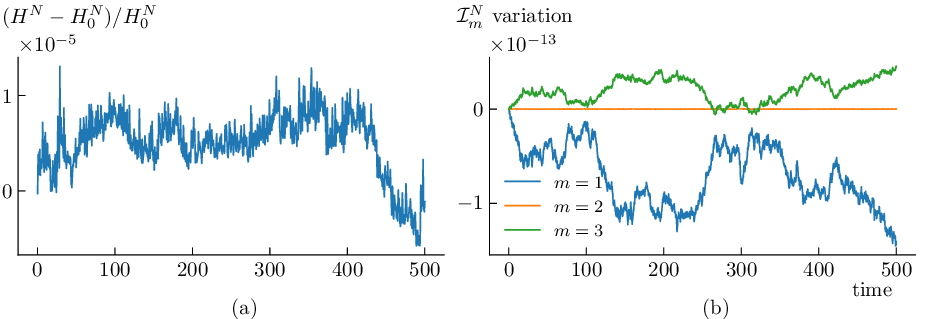}
\caption{Preservation of the Hamiltonian $H^{N}$ (a) and the cross-helicity Casimir $\mathcal{I}^{N}_{m}$ (b) for $N=5$. The step-size is $h=0.01$. The Casimir $\mathcal{I}^{N}_{m}$ is preserved up to machine precision, which is indicated by the variation magnitude $10^{-13}$. The magnitude $10^{-5}$ in the relative error of the Hamiltonian indicates its near preservation.
}
\label{preservationofcrosshelandHam2}
\end{figure}
\section{Conclusion}
The system of axisymmetric MHD equations on $\mathbb{S}^{3}$ derived in the present paper possesses a Hamiltonian formulation as a Lie--Poisson system on the dual of the semidirect product Lie algebra, which implies the existence of an infinite number of conservation laws. Its finite-dimensional approximation has been developed based on Zeitlin's matrix equations. The resulting finite-dimensional Lie--Poisson system possesses a finite number of independent Casimirs approximating those of the original equations, and their number grows up, as the dimension of the matrices increases. The finite-dimensional axisymmetric MHD--Zeitlin equations are further approximated in time using two structure preserving time integrators, whose preservation properties are shown in numerical simulations. One of them is obtained by applying the isospectral symplectic Runge--Kutta integrator to the block matrix equations equivalent to \eqref{LP3Dquant}, while the other one is derived via the discrete Lie--Poisson reduction for semidirect products of the form $\mathfrak{F}=(\mathfrak{su}(N)\ltimes\mathfrak{su}(N)^{*})\ltimes(\mathfrak{su}(N)\ltimes\mathfrak{su}(N)^{*})^{*}$.

The developed discretization provides the possibility of structure-preserving simulation of turbulence governed by the reduced version of the 3D MHD equations, thus extending the structure-preserving approach to MHD turbulence from two dimensions addressed previously to three dimensions. 
\appendix
\section{Derivation of reduced equations}
\label{appendix1}
Here, we show the details of the derivation of equations \eqref{MHDtilde} from equations \eqref{vectorfieldvortform}.

Using expressions \eqref{ubS3} and \eqref{omegaS3} for the fields $u$ and $\omega$, we get (we will use the notation $\tilde\Delta\tilde f=(E_{3}^{2}+E_{2}^{2})\tilde f$ for any $E_{1}$-invariant function $\tilde f$):
\begin{equation}
\label{wucommutatorintermsoffuncts}
\begin{split}
[\omega,u]&=[(\tilde\Delta\tilde\psi)E_{1}+(E_{3}\tilde q)E_{2}-(E_{2}\tilde q)E_{3},\tilde\sigma E_{1}-(E_{3}\tilde\psi)E_{2}+(E_{2}\tilde\psi)E_{3}]{}\\&=
[(\tilde\Delta\tilde\psi)E_{1},\tilde\sigma E_{1}]+[(E_{3}\tilde\psi)E_{2},(\tilde\Delta\tilde\psi)E_{1}]+[(\tilde\Delta\tilde\psi)E_{1},(E_{2}\tilde\psi)E_{3}]{}\\&+
[(E_{3}\tilde q)E_{2},\tilde\sigma E_{1}]+[(E_{3}\tilde\psi)E_{2},(E_{3}\tilde q)E_{2}]+[(E_{3}\tilde q)E_{2},(E_{2}\tilde\psi)E_{3}]{}\\&+
[\tilde\sigma E_{1},(E_{2}\tilde q)E_{3}]+[(E_{2}\tilde q)E_{3},(E_{3}\tilde\psi)E_{2}]+[(E_{2}\tilde \psi)E_{3},(E_{2}\tilde q)E_{3}],
\end{split}
\end{equation}
where $\tilde{q}=\tilde{\psi}+\tilde{\sigma}$.

To proceed, we will need one brief intermediate result.
\begin{lemma}
\label{lemmaaboutdelta}
If $\tilde\psi$ is $E_{1}$-invariant, then $\tilde\Delta\tilde\psi$ is also $E_{1}$-invariant.
\end{lemma}
\begin{proof}
We need to check that $E_{1}(\tilde\Delta\tilde\psi)=0$.
\begin{equation*}
\begin{split}
E_{1}(\tilde\Delta\tilde\psi)&=E_{1}E_{2}E_{2}\tilde\psi+E_{1}E_{3}E_{3}\tilde\psi=\underbrace{[E_{1},E_{2}]}_{=-E_{3}}E_{2}\tilde\psi+E_{2}E_{1}E_{2}\tilde\psi+\underbrace{[E_{1},E_{3}]}_{=E_{2}}E_{3}\tilde\psi+E_{3}E_{1}E_{3}\tilde\psi{}\\&=E_{2}\underbrace{[E_{1},E_{2}]}_{=-E_{3}}\tilde\psi+E_{2}^{2}E_{1}\tilde\psi+E_{3}^{2}E_{1}\tilde\psi+E_{3}\underbrace{[E_{1},E_{3}]}_{=E_{2}}\tilde\psi=0,
\end{split}
\end{equation*}
where we used \eqref{LieAlgStr}, $E_{1}$-invariance of $\tilde\psi$, and the fact that $E_{2}$ and $E_{3}$ commute when acting upon $E_{1}$-invariant functions, again due to \eqref{LieAlgStr}.
\end{proof}
Further, using the formula
\begin{equation*}
[\tilde fE_{i},\tilde g E_{j}]=\tilde fE_{i}(\tilde g)E_{j}-\tilde g E_{j}(\tilde f)E_{i}+\tilde f\tilde g[E_{i},E_{j}],
\end{equation*}
we transform each of the nine terms in \eqref{wucommutatorintermsoffuncts}. First, we conclude that because of $E_{1}$-invariance and the result of Lemma \ref{lemmaaboutdelta}, the first term vanishes: $[(\tilde\Delta\tilde\psi)E_{1},\tilde\sigma E_{1}]=0$. Next, we get
\begin{align}
\begin{split}
\label{allnineterms}
[(E_{3}\tilde\psi)E_{2},(\tilde\Delta\tilde\psi)E_{1}]&=(E_{3}\tilde\psi)E_{2}(\tilde\Delta\tilde\psi)E_{1}-(\tilde\Delta\tilde\psi)(E_{2}\tilde\psi)E_{2}+(E_{3}\tilde\psi)(\tilde\Delta\tilde\psi)E_{3}{}\\
[(\tilde\Delta\tilde\psi)E_{1},(E_{2}\tilde\psi)E_{3}]&=-(E_{2}\tilde\psi)E_{3}(\tilde\Delta\tilde\psi)E_{1}+(\tilde\Delta\tilde\psi)(E_{2}\tilde\psi)E_{2}-(\tilde\Delta\tilde\psi)(E_{3}\tilde\psi)E_{3}{}\\
[(E_{3}\tilde q)E_{2},\tilde\sigma E_{1}]&=(E_{3}\tilde q)(E_{2}\tilde\sigma)E_{1}-\tilde\sigma(E_{2}\tilde q)E_{2}+\tilde\sigma(E_{3}\tilde q)E_{3},{}\\
[(E_{3}\tilde\psi)E_{2},(E_{3}\tilde q)E_{2}]&=(E_{3}\tilde\psi)(E_{2}E_{3}\tilde q)E_{2}-(E_{3}\tilde q)(E_{2}E_{3}\tilde\psi)E_{2},{}\\
[(E_{3}\tilde q)E_{2},(E_{2}\tilde\psi)E_{3}]&=-(E_{3}\tilde q)(E_{2}\tilde\psi)E_{1}-(E_{2}\tilde\psi)(E_{3}^{2}\tilde q)E_{2}+(E_{3}\tilde q)(E_{2}^{2}\tilde\psi)E_{3},{}\\
[\tilde\sigma E_{1},(E_{2}\tilde q)E_{3}]&=-(E_{2}\tilde q)(E_{3}\tilde\sigma)E_{1}+\tilde\sigma(E_{2}\tilde q)E_{2}-\tilde\sigma(E_{3}\tilde q)E_{3},{}\\
[(E_{2}\tilde q)E_{3},(E_{3}\tilde\psi)E_{2}]&=(E_{2}\tilde q)(E_{3}\tilde\psi)E_{1}+(E_{2}\tilde q)(E_{3}^{2}\tilde\psi)E_{2}-(E_{3}\tilde\psi)(E_{2}^{2}\tilde q)E_{3},{}\\
[(E_{2}\tilde\psi)E_{3},(E_{2}\tilde q)E_{3}]&=(E_{2}\tilde\psi)(E_{3}E_{2}\tilde q)E_{3}-(E_{2}\tilde q)(E_{3}E_{2}\tilde\psi)E_{3}.
\end{split}
\end{align}
Summing up all the terms in \eqref{allnineterms}, we obtain
\begin{equation}
\label{finalformsforwubrack}
\begin{split}
[\omega,u]&=\left((E_{3}\tilde\psi)E_{2}(\tilde\Delta\tilde\psi)-(E_{2}\tilde\psi)E_{3}(\tilde\Delta\tilde\psi)+(E_{3}\tilde q)(E_{2}\tilde\sigma)-(E_{3}\tilde q)(E_{2}\tilde\psi)-(E_{2}\tilde q)(E_{3}\tilde\sigma)+(E_{2}\tilde q)(E_{3}\tilde\psi)\right)E_{1}{}\\&+
\left((E_{3}\tilde\psi)(E_{2}E_{3}\tilde q)-(E_{3}\tilde q)(E_{2}E_{3}\tilde\psi)-(E_{2}\tilde\psi)(E_{3}^{2}\tilde q)+(E_{2}\tilde q)(E_{3}^{2}\tilde\psi)\right)E_{2}{}\\&+\left((E_{3}\tilde q)(E_{2}^{2}\tilde\psi)-(E_{3}\tilde\psi)(E_{2}^{2}\tilde q)+(E_{2}\tilde\psi)(E_{3}E_{2}\tilde q)-(E_{2}\tilde q)(E_{3}E_{2}\tilde\psi)\right)E_{3}{}\\&=\left(\mathfrak{B}(\tilde\Delta\tilde\psi,\tilde\psi)+\mathfrak{B}(\tilde\sigma,\tilde q)+\mathfrak{B}(\tilde q,\tilde\psi)\right)E_{1}+E_{3}\left(\mathfrak{B}(\tilde q,\tilde\psi)\right)E_{2}-E_{2}\left(\mathfrak{B}(\tilde q,\tilde\psi)\right)E_{3}{}\\&=\left(\mathfrak{B}(\tilde\Delta\tilde\psi,\tilde\psi)+2\mathfrak{B}(\tilde q,\tilde\psi)\right)E_{1}+E_{3}\left(\mathfrak{B}(\tilde q,\tilde\psi)\right)E_{2}-E_{2}\left(\mathfrak{B}(\tilde q,\tilde\psi)\right)E_{3}
,
\end{split}
\end{equation}
where 
\begin{equation*}
\mathfrak{B}(\tilde f,\tilde h)=E_{2}(\tilde f)E_{3}(\tilde h)-E_{3}(\tilde f)E_{2}(\tilde h),
\end{equation*}
for arbitrary $E_{1}$-invariant functions $\tilde f$ and $\tilde h$, and where in the last equality we used the definition of the new field $\tilde q=\tilde\sigma+\tilde\psi$.

Analogously, for the bracket $[B,j]$, we get
\begin{equation}
\label{finalformforbjbrack}
[B,j]=\left(\mathfrak{B}(\tilde\xi,\tilde\Delta\tilde\xi)+2\mathfrak{B}(\tilde\xi,\tilde\rho)\right)E_{1}+E_{3}\left(\mathfrak{B}(\tilde\xi,\tilde\rho)\right)E_{2}-E_{2}\left(\mathfrak{B}(\tilde\xi,\tilde\rho)\right)E_{3}.
\end{equation}
Since $\dot\omega=[\omega,u]+[B,j]$ and $\dot\omega=(\tilde\Delta\dot{\tilde\psi})E_{1}+(E_{3}(\dot{\tilde q}))E_{2}-(E_{2}\dot{\tilde q})E_{3}$, from \eqref{finalformsforwubrack}-\eqref{finalformforbjbrack}, we find the following equations for the pair $(\tilde\psi,\tilde q)$:
\begin{equation*}
\left\{
\begin{aligned}
&\frac{\partial}{\partial t}\tilde\Delta\tilde\psi=\mathfrak{B}(\tilde\Delta\tilde\psi,\tilde\psi)+\mathfrak{B}(\tilde\xi,\tilde\Delta\tilde\xi)+2\mathfrak{B}(\tilde q,\tilde\psi)+2\mathfrak{B}(\tilde \xi,\tilde\rho), \\
&\frac{\partial\tilde q}{\partial t}=\mathfrak{B}(\tilde q,\tilde\psi)+\mathfrak{B}(\tilde\xi,\tilde\rho).\\
\end{aligned}
\right.
\end{equation*}

Repeating the same steps for the equation $\dot B=[B,u]$, we also find the corresponding equations for the pair $(\tilde\rho,\tilde\xi)$:
\begin{equation*}
\left\{
\begin{aligned}
&\frac{\partial\tilde\rho}{\partial t}=\mathfrak{B}(\tilde \rho,\tilde\psi)+\mathfrak{B}(\tilde \xi,\tilde q)+2\mathfrak{B}(\tilde \psi,\tilde\xi),\\
&\frac{\partial\tilde\xi}{\partial t}=\mathfrak{B}(\tilde\xi,\tilde\psi).\\
\end{aligned}
\right.
\end{equation*}
\section{Reduction theory for semidirect product extensions}
\label{appendix2}
Here, we develop the reduction theory for semidirect product Lie algebras of the form $\mathfrak{F}=(\mathfrak{g}\ltimes\mathfrak{g}^{*})\ltimes(\mathfrak{g}\ltimes\mathfrak{g}^{*})^{*}$, with $\mathfrak{g}$ being a matrix Lie algebra. The main goal is to derive formula \eqref{mommapforsemidir} for the momentum map. The exposition here holds for an arbitrary matrix Lie algebra $\mathfrak{g}$, with $\mathfrak{g}=\mathfrak{su}(N)$ being a primary example.

A Lie group $\mathcal{F}$ corresponding to the Lie algebra $\mathfrak{F}=\mathfrak{f}\ltimes\mathfrak{f}^{*}$, with $\mathfrak{f}=\mathfrak{g}\ltimes\mathfrak{g}^{*}$, has the form $\mathcal{F}=F\ltimes\mathfrak{f}^{*}$, where $F$ is the Lie group for $\mathfrak{f}=\mathfrak{g}\ltimes\mathfrak{g}^{*}$, which is $F=\mathcal{G}\ltimes\mathfrak{g}^{*}$, with $\mathcal{G}$ being the Lie group for $\mathfrak{g}$. 

Identifying $\mathfrak{g}^{*}$ with $\mathfrak{g}$ via the Frobenius inner product $\langle A,B\rangle_{F}=\mathrm{tr}(A^{\dag}B)$, for $A\in\mathfrak{g}^{*},\,B\in\mathfrak{g}$, we will use, in what follows, the following formulas for the adjoint and coadjoint operators:
\begin{equation}
\label{adjointandcadjonG}
\mathrm{Ad}_{q}\alpha=q\alpha q^{-1},\quad\mathrm{Ad}_{q}^{*}\alpha^{\prime}=q^{\dag}\alpha^{\prime} (q^{-1})^{\dag},\quad \mathrm{ad}_{v}\alpha=[v,\alpha],\quad\mathrm{ad}_{v}^{*}\alpha^{\prime}=[v^{\dag},\alpha^{\prime}],
\end{equation}
for $q\in\mathcal{G}$, $\alpha,v\in\mathfrak{g}$, $\alpha^{\prime}\in\mathfrak{g}^{*}$.

Let now $(G,U),\,(Q,M)\in\mathcal{F}$, i.e. $G=(g,\beta)\in F,\,U=(u,\sigma)\in\mathfrak{f}^{*}$, and $Q=(q,\alpha)\in F,\,M=(m,\gamma)\in\mathfrak{f}^{*}$, which in turn means that $g,q\in\mathcal{G}$, $m,u\in\mathfrak{g}$, and $\alpha,\beta,\gamma,\sigma\in\mathfrak{g}^{*}$. Then, the group operation on $\mathcal{F}$ has the form of the standard group operation on a magnetic extension Lie group (see, for example, \cite[Ch. I.10.B]{ArnKh})
\begin{equation}
\label{groupoperationonsemidirextension}
(G,U)\cdot(Q,M)=(GQ,\mathrm{Ad}^{*}_{Q}U+M),
\end{equation}
where $GQ=(g,\beta)\cdot(q,\alpha)=(gq,\mathrm{Ad}^{*}_{q}\beta+\alpha)$ is the group operation on $F=\mathcal{G}\ltimes\mathfrak{g}^{*}$, and $\mathrm{Ad}^{*}_{Q}U$ is the operator of the coadjoint action of $F$ on $\mathfrak{f}^{*}$.
\begin{proposition}
The coadjoint operator of the $F$-action on $\mathfrak{f}^{*}$ is 
\begin{equation*}
\mathrm{Ad}^{*}_{Q}U=\left(\mathrm{Ad}_{q^{-1}}u,\mathrm{Ad}^{*}_{q}\sigma-\mathrm{ad}^{*}_{\mathrm{Ad}_{q^{-1}}u}\alpha\right).
\end{equation*}
\end{proposition}
\begin{proof}
We recall that the inner automorphism on $F$ is $A_{Q}G=QGQ^{-1}$, where $Q^{-1}=(q,\alpha)^{-1}=(q^{-1},-\mathrm{Ad}^{*}_{q^{-1}}\alpha)$. This gives
\begin{equation*}
A_{Q}G=A_{(q,\alpha)}(g,\beta)=(q,\alpha)\cdot(g,\beta)\cdot(q,\alpha)^{-1}=\left(qgq^{-1},\mathrm{Ad}^{*}_{gq^{-1}}\alpha+\mathrm{Ad}^{*}_{q^{-1}}\beta-\mathrm{Ad}^{*}_{q^{-1}}\alpha\right).
\end{equation*}
Let now $(g,\beta)$ be close to the identity, i.e.
\begin{equation*}
g(t)=I+tv+\mathcal{O}(t^{2}),\quad\beta(t)=0+tw+\mathcal{O}(t^{2}),\quad v\in\mathfrak{g},\,w\in\mathfrak{g}^{*}.
\end{equation*}
Then, we get
\begin{equation*}
\mathrm{Ad}_{(q,\alpha)}(v,w)=\left.\frac{\mathrm{d}}{\mathrm{d}t}\right|_{t=0}A_{(q,\alpha)}(g(t),\beta(t))=\left.\frac{\mathrm{d}}{\mathrm{d}t}\right|_{t=0}\left(A_{q}g(t),\mathrm{Ad}^{*}_{g(t)q^{-1}}\alpha+\mathrm{Ad}^{*}_{q^{-1}}\beta(t)-\mathrm{Ad}^{*}_{q^{-1}}\alpha\right).
\end{equation*}
The first component becomes the adjoint operator $\mathrm{Ad}_{q}v$ of $\mathcal{G}$-action on $\mathfrak{g}$, and for the second component we get
\begin{equation*}
\begin{split}
&\left.\frac{\mathrm{d}}{\mathrm{d}t}\right|_{t=0}\left(\mathrm{Ad}^{*}_{g(t)q^{-1}}\alpha+\mathrm{Ad}^{*}_{q^{-1}}\beta(t)-\mathrm{Ad}^{*}_{q^{-1}}\alpha\right)=\left.\frac{\mathrm{d}}{\mathrm{d}t}\right|_{t=0}\left(\mathrm{Ad}^{*}_{g(t)q^{-1}}\alpha\right)+\left.\frac{\mathrm{d}}{\mathrm{d}t}\right|_{t=0}\left(\mathrm{Ad}^{*}_{q^{-1}}\beta(t)\right)={}\\&\left.\frac{\mathrm{d}}{\mathrm{d}t}\right|_{t=0}\left(\mathrm{Ad}^{*}_{q^{-1}}\mathrm{Ad}^{*}_{g(t)}\alpha\right)+\mathrm{Ad}^{*}_{q^{-1}}w=\mathrm{Ad}^{*}_{q^{-1}}\left(w+\mathrm{ad}^{*}_{v}\alpha\right).
\end{split}
\end{equation*}
Finally, we arrive at the formula for the group adjoint operator on $\mathfrak{f}=\mathfrak{g}\ltimes\mathfrak{g}^{*}$:
\begin{equation}
\label{groupadjointonf}
\mathrm{Ad}_{(q,\alpha)}(v,w)=\left(\mathrm{Ad}_{q}v,\mathrm{Ad}^{*}_{q^{-1}}\left(w+\mathrm{ad}^{*}_{v}\alpha\right)\right).
\end{equation}

Let now $U^{\prime}=(v,w)\in\mathfrak{f}$. For the group coadjoint operator $\mathrm{Ad}^{*}_{Q}U$, we get
\begin{equation*}
\begin{split}
\langle\mathrm{Ad}^{*}_{Q}U,U^{\prime}\rangle&=\langle(u,\sigma),\mathrm{Ad}_{(q,\alpha)}(v,w)\rangle=\left\langle(u,\sigma),\left(\mathrm{Ad}_{q}v,\mathrm{Ad}^{*}_{q^{-1}}\left(w+\mathrm{ad}^{*}_{v}\alpha\right)\right)\right\rangle{}\\&=\left\langle\mathrm{Ad}^{*}_{q^{-1}}\left(w+\mathrm{ad}^{*}_{v}\alpha\right),u\right\rangle+\left\langle\sigma,\mathrm{Ad}_{q}v\right\rangle=\left\langle\mathrm{Ad}_{q}^{*}\sigma,v\right\rangle+\left\langle w,\mathrm{Ad}_{q^{-1}}u\right\rangle+\left\langle\mathrm{ad}^{*}_{v}\alpha,\mathrm{Ad}_{q^{-1}}u\right\rangle,
\end{split}
\end{equation*}
where we used \eqref{groupadjointonf}. The last component yields
\begin{equation*}
\left\langle\mathrm{ad}^{*}_{v}\alpha,\mathrm{Ad}_{q^{-1}}u\right\rangle=\left\langle\alpha,\mathrm{ad}_{v}\left(\mathrm{Ad}_{q^{-1}}u\right)\right\rangle=-\left\langle\alpha,\mathrm{ad}_{\mathrm{Ad}_{q^{-1}}u}v\right\rangle=-\left\langle\mathrm{ad}^{*}_{\mathrm{Ad}_{q^{-1}}u}\alpha,v\right\rangle.
\end{equation*}
Finally, we get
\begin{equation*}
\langle\mathrm{Ad}^{*}_{Q}U,U^{\prime}\rangle=\left\langle\mathrm{Ad}_{q}^{*}\sigma,v\right\rangle+\left\langle w,\mathrm{Ad}_{q^{-1}}u\right\rangle-\left\langle\mathrm{ad}^{*}_{\mathrm{Ad}_{q^{-1}}u}\alpha,v\right\rangle=\left\langle\left(\mathrm{Ad}_{q^{-1}}u,\mathrm{Ad}_{q}^{*}\sigma-\mathrm{ad}^{*}_{\mathrm{Ad}_{q^{-1}}u}\alpha\right),(v,w)\right\rangle,
\end{equation*}
which implies the assertion.
\end{proof}

Summarizing, we write out the group operation \eqref{groupoperationonsemidirextension} on $\mathcal{F}=(\mathcal{G}\ltimes\mathfrak{g}^{*})\ltimes(\mathfrak{g}\ltimes\mathfrak{g}^{*})^{*}$ in terms of its ``elementary" building blocks:
\begin{equation}
\label{completeformulaforgroupoperation}
(G,U)\cdot(Q,M)=\left((gq,\mathrm{Ad}^{*}_{q}\beta+\alpha),(\mathrm{Ad}_{q^{-1}}u+m,\mathrm{Ad}^{*}_{q}\sigma-\mathrm{ad}^{*}_{\mathrm{Ad}_{q^{-1}}u}\alpha+\gamma)\right),
\end{equation}
for $(G,U)=((g,\beta),(u,\sigma))\in\mathcal{F}$, $(Q,M)=((q,\alpha),(m,\gamma))\in\mathcal{F}$.

The next step is to obtain a formula for the cotangent lift of the left $\mathcal{F}$-action on itself. Let $\Phi_{(G,U)}\colon(Q,M)\mapsto(\widetilde Q,\widetilde M)=(G,U)\cdot(Q,M)$ be the left action of $\mathcal{F}$ on itself. Further, let $(V,B)=((v,c),(w,d))\in T_{(Q,M)}\mathcal{F}$ be a tangent vector at the point $(Q,M)\in\mathcal{F}$. 
\begin{proposition}
The differential 
\begin{equation*}
\left(\Phi_{(G,U)}\right)_{*}\colon T_{(Q,M)}\mathcal{F}\to T_{(\widetilde{Q},\widetilde{M})}\mathcal{F},\quad\left(\Phi_{(G,U)}\right)_{*}\colon(V,B)\mapsto(\widetilde{V},\widetilde{B})=((\widetilde{v},\widetilde{c}),(\widetilde{w},\widetilde{d}))
\end{equation*}
of the left $\mathcal{F}$-action on itself is given by the formulas
\begin{equation}
\label{tangentliftformula}
\begin{split}
&\widetilde{v}=gv,\quad\widetilde{c}=c+\mathrm{ad}^{*}_{q^{-1}v}\left(\mathrm{Ad}^{*}_{q}\beta\right),\quad\widetilde{w}=w-\mathrm{ad}_{q^{-1}v}\left(\mathrm{Ad}_{q^{-1}}u\right){}\\&
\widetilde{d}=d+\mathrm{ad}^{*}_{q^{-1}v}\left(\mathrm{Ad}^{*}_{q}\sigma\right)+\mathrm{ad}^{*}_{\mathrm{ad}_{q^{-1}v}\left(\mathrm{Ad}_{q^{-1}}u\right)}(\alpha)-\mathrm{ad}^{*}_{\mathrm{Ad}_{q^{-1}}u}(c).
\end{split}
\end{equation}
\end{proposition}
\begin{proof}
Consider a path $(q(t),\alpha(t),m(t),\gamma(t))\in\mathcal{F}$ on the group $\mathcal{F}$ generated by the vector $(V,B)=((v,c),(w,d))$ and passing through the point $((q,\alpha),(m,\gamma))\in\mathcal{F}$, i.e.
\begin{equation*}
q(t)=q+tv+\mathcal{O}(t^{2}),\,\alpha(t)=\alpha+tc+\mathcal{O}(t^{2}),\, m(t)=m+tw+\mathcal{O}(t^{2}),\,\gamma(t)=\gamma+td+\mathcal{O}(t^{2}).
\end{equation*}
The corresponding path $\left(\widetilde{Q}(t),\widetilde{M}(t)\right)$ passing through the point $\left(\widetilde{Q},\widetilde{M}\right)=((\widetilde{q},\widetilde{\alpha}),(\widetilde{m},\widetilde{\gamma}))\in\mathcal{F}$ is generated by the vector $(\widetilde{V},\widetilde{B})\in T_{\left(\widetilde{Q},\widetilde{M}\right)}\mathcal{F}$:
\begin{equation*}
\begin{split}
\left(\Phi_{(G,U)}\right)_{*}(V,B)&=(\widetilde{V},\widetilde{B})=\left.\frac{\mathrm{d}}{\mathrm{d}t}\right|_{t=0}\left(\widetilde{Q}(t),\widetilde{M}(t)\right)=\left.\frac{\mathrm{d}}{\mathrm{d}t}\right|_{t=0}(G,U)\cdot(Q(t),M(t)){}\\&=\left.\frac{\mathrm{d}}{\mathrm{d}t}\right|_{t=0}\left((gq(t),\mathrm{Ad}^{*}_{q(t)}\beta+\alpha(t)),(\mathrm{Ad}_{q^{-1}(t)}u+m(t),\mathrm{Ad}^{*}_{q(t)}\sigma-\mathrm{ad}^{*}_{\mathrm{Ad}_{q^{-1}(t)}u}\alpha(t)+\gamma(t))\right),
\end{split}
\end{equation*}
where we used \eqref{completeformulaforgroupoperation}. We calculate each component separately. Firstly, for the $\widetilde{v}$ component, we get
\begin{equation*}
\widetilde{v}=\left.\frac{\mathrm{d}}{\mathrm{d}t}\right|_{t=0}gq(t)=gv.
\end{equation*}
Further, for the $\widetilde{c}$ component we find
\begin{equation*}
\begin{split}
\widetilde{c}&=\left.\frac{\mathrm{d}}{\mathrm{d}t}\right|_{t=0}\left(\mathrm{Ad}^{*}_{q(t)}\beta+\alpha(t)\right)=c+\left.\frac{\mathrm{d}}{\mathrm{d}t}\right|_{t=0}\mathrm{Ad}^{*}_{q(I+tq^{-1}v+\mathcal{O}(t^{2}))}\beta{}\\&=c+\left.\frac{\mathrm{d}}{\mathrm{d}t}\right|_{t=0}\mathrm{Ad}^{*}_{I+tq^{-1}v+\mathcal{O}(t^{2})}\left(\mathrm{Ad}^{*}_{q}\beta\right)=c+\mathrm{ad}^{*}_{q^{-1}v}\left(\mathrm{Ad}^{*}_{q}\beta\right).
\end{split}
\end{equation*}
For the $\widetilde{w}$ component, we get
\begin{equation*}
\begin{split}
\widetilde{w}&=\left.\frac{\mathrm{d}}{\mathrm{d}t}\right|_{t=0}\left(\mathrm{Ad}_{q^{-1}(t)}u+m(t)\right)=w+\left.\frac{\mathrm{d}}{\mathrm{d}t}\right|_{t=0}\mathrm{Ad}_{(I+tq^{-1}v+\mathcal{O}(t^{2}))^{-1}q^{-1}}u{}\\&=w+\left.\frac{\mathrm{d}}{\mathrm{d}t}\right|_{t=0}\mathrm{Ad}_{I-tq^{-1}v+\mathcal{O}(t^{2})}\mathrm{Ad}_{q^{-1}}u=w-\mathrm{ad}_{q^{-1}v}\left(\mathrm{Ad}_{q^{-1}}u\right).
\end{split}
\end{equation*}
Finally, for the $\widetilde{d}$, we find
\begin{equation}
\label{intermediatefortilded}
\widetilde{d}=\left.\frac{\mathrm{d}}{\mathrm{d}t}\right|_{t=0}\left(\mathrm{Ad}^{*}_{q(t)}\sigma-\mathrm{ad}^{*}_{\mathrm{Ad}_{q^{-1}(t)}u}\alpha(t)+\gamma(t)\right)=d+\mathrm{ad}^{*}_{q^{-1}v}\left(\mathrm{Ad}^{*}_{q}\sigma\right)-\left.\frac{\mathrm{d}}{\mathrm{d}t}\right|_{t=0}\left(\mathrm{ad}^{*}_{\mathrm{Ad}_{q^{-1}(t)}u}\alpha(t)\right).
\end{equation}
Let us calculate the last term separately:
\begin{equation*}
\left.\frac{\mathrm{d}}{\mathrm{d}t}\right|_{t=0}\left(\mathrm{ad}^{*}_{\mathrm{Ad}_{q^{-1}(t)}u}\alpha(t)\right)=\mathrm{ad}^{*}_{\mathrm{Ad}_{q^{-1}}u}(c)+\left.\frac{\mathrm{d}}{\mathrm{d}t}\right|_{t=0}\left(\mathrm{ad}^{*}_{\mathrm{Ad}_{q^{-1}(t)}u}\alpha\right)=\mathrm{ad}^{*}_{\mathrm{Ad}_{q^{-1}}u}(c)-\mathrm{ad}^{*}_{\mathrm{ad}_{q^{-1}v}\left(\mathrm{Ad}_{q^{-1}}u\right)}(\alpha),
\end{equation*}
and therefore \eqref{intermediatefortilded} becomes
\begin{equation*}
\widetilde{d}=d+\mathrm{ad}^{*}_{q^{-1}v}\left(\mathrm{Ad}^{*}_{q}\sigma\right)+\mathrm{ad}^{*}_{\mathrm{ad}_{q^{-1}v}\left(\mathrm{Ad}_{q^{-1}}u\right)}(\alpha)-\mathrm{ad}^{*}_{\mathrm{Ad}_{q^{-1}}u}(c),
\end{equation*}
which finalizes the proof.
\end{proof}

The cotangent lift of the action \eqref{completeformulaforgroupoperation} can be obtained from \eqref{tangentliftformula}. Let $(P,A)\in T^{*}_{(Q,M)}\mathcal{F}$ be an element in the cotangent space at the point $(Q,M)\in\mathcal{F}$, i.e. $(P,A)=((p,\nu),(a,b))$, with $(p,\nu)\in T^{*}_{(q,\alpha)}F$ being an element in the cotangent space of the group $F$ at the point $(q,\alpha)\in F$, and $(a,b)\in\mathfrak{f}^{*}$. Further, let $(\Phi_{(G,U)})^{*}$ be the co-differential of the left $\mathcal{F}$-action
\begin{equation*}
(\Phi_{(G,U)})^{*}\colon T^{*}_{(\widetilde{Q},\widetilde{M})}\mathcal{F}\to T^{*}_{(Q,M)}\mathcal{F},\quad \left(\Phi_{(G,U)}\right)^{*}\left(\widetilde{P},\widetilde{A}\right)=\left(P,A\right),
\end{equation*}
defined as follows:
\begin{equation}
\begin{split}
\label{defofcodiff}
\left\langle\left(P,A\right),\left(V,B\right)\right\rangle&=\left\langle(\Phi_{(G,U)})^{*}\left(\widetilde{P},\widetilde{A}\right),(V,B)\right\rangle=\left\langle\left(\widetilde{P},\widetilde{A}\right),\left(\Phi_{(G,U)}\right)_{*}\left(V,B\right)\right\rangle{}\\&=\left\langle\left(\widetilde{P},\widetilde{A}\right),\left(\widetilde{V},\widetilde{B}\right)\right\rangle.
\end{split}
\end{equation}
Together with \eqref{tangentliftformula}, equation \eqref{defofcodiff} defines $\left(\widetilde{P},\widetilde{A}\right)=((\widetilde{p},\widetilde{\nu}),(\widetilde{a},\widetilde{b}))$ via $(P,A)=((p,\nu),(a,b))$.
\begin{proposition}
The co-differential
\begin{equation*}
\left(\Phi_{(G,U)}^{-1}\right)^{*}\colon T_{(Q,M)}^{*}\mathcal{F}\to T_{(\widetilde{Q},\widetilde{M})}^{*}\mathcal{F},\quad\left(\Phi_{(G,U)}^{-1}\right)^{*}\colon(P,A)\mapsto(\widetilde{P},\widetilde{A})=((\widetilde{p},\widetilde{\nu}),(\widetilde{a},\widetilde{b}))
\end{equation*}
of the left $\mathcal{F}$-action on itself is given by the formulas
\begin{equation}
\label{cotangentliftformula}
\begin{split}
&\widetilde{a}=a,\quad\widetilde{b}=b,\quad\widetilde{\nu}=\nu-\mathrm{ad}_{a}\left(\mathrm{Ad}_{q^{-1}}u\right),\\&
\widetilde{p}=\left(g^{-1}\right)^{\dag}p+\left(g^{-1}\right)^{\dag}\left(q^{-1}\right)^{\dag}\left(\mathrm{ad}^{*}_{\nu-\mathrm{ad}_{a}\left(\mathrm{Ad}_{q^{-1}}u\right)}\left(\mathrm{Ad}^{*}_{q}\beta\right)-\mathrm{ad}^{*}_{\mathrm{Ad}_{q^{-1}}u}\left(\mathrm{ad}^{*}_{a}\alpha+b\right)+\mathrm{ad}^{*}_{a}\left(\mathrm{Ad}^{*}_{q}\sigma\right)\right).
\end{split}
\end{equation}
\end{proposition}
\begin{proof}
Equation \eqref{defofcodiff} yields
\begin{equation}
\label{fullpairing}
\langle p,v\rangle+\langle c,\nu\rangle+\langle d,a\rangle+\langle b,w\rangle=\langle \widetilde{p},\widetilde{v}\rangle+\langle \widetilde{c},\widetilde{\nu}\rangle+\langle \widetilde{d},\widetilde{a}\rangle+\langle \widetilde{b},\widetilde{w}\rangle.
\end{equation}
Using \eqref{tangentliftformula} and isolating $(v,c,d,w)$, we obtain
\begin{equation}
\label{eqfortildep}
\langle \widetilde{p},\widetilde{v}\rangle=\langle \widetilde{p},gqq^{-1}v\rangle=\langle (gq)^{\dag}\widetilde{p},q^{-1}v\rangle.
\end{equation}
Further, for the $\langle \widetilde{c},\widetilde{\nu}\rangle$ component, we get
\begin{equation}
\label{eqfortildenu}
\langle \widetilde{c},\widetilde{\nu}\rangle=\left\langle c+\mathrm{ad}^{*}_{q^{-1}v}\left(\mathrm{Ad}^{*}_{q}\beta\right),\widetilde{\nu}\right\rangle=\langle c,\widetilde{\nu}\rangle-\left\langle\mathrm{ad}^{*}_{\widetilde{\nu}}\left(\mathrm{Ad}^{*}_{q}\beta\right),q^{-1}v\right\rangle.
\end{equation}
For the $\langle \widetilde{d},\widetilde{a}\rangle$ component, we find
\begin{equation}
\label{eqfortildea}
\begin{split}
\langle \widetilde{d},\widetilde{a}\rangle&=\langle d,\widetilde{a}\rangle+\left\langle\mathrm{ad}^{*}_{q^{-1}v}\left(\mathrm{Ad}^{*}_{q}\sigma\right)+\mathrm{ad}^{*}_{\mathrm{ad}_{q^{-1}v}\left(\mathrm{Ad}_{q^{-1}}u\right)}(\alpha)-\mathrm{ad}^{*}_{\mathrm{Ad}_{q^{-1}}u}(c),\widetilde{a}\right\rangle{}\\&=\langle d,\widetilde{a}\rangle+\left\langle c,\mathrm{ad}_{\widetilde{a}}\left(\mathrm{Ad}_{q^{-1}}u\right)\right\rangle+\left\langle\mathrm{ad}^{*}_{\mathrm{Ad}_{q^{-1}}u}\left(\mathrm{ad}^{*}_{\widetilde{a}}\alpha\right),q^{-1}v\right\rangle-\left\langle\mathrm{ad}^{*}_{\widetilde{a}}\left(\mathrm{Ad}^{*}_{q}\sigma\right),q^{-1}v\right\rangle.
\end{split}
\end{equation}
Finally, for the term $\langle \widetilde{b},\widetilde{w}\rangle$, we get
\begin{equation}
\label{eqfortildeb}
\langle \widetilde{b},\widetilde{w}\rangle=\left\langle\widetilde{b},w-\mathrm{ad}_{q^{-1}v}\left(\mathrm{Ad}_{q^{-1}}u\right)\right\rangle=\langle\widetilde{b},w\rangle+\left\langle\mathrm{ad}^{*}_{\mathrm{Ad}_{q^{-1}}u}\widetilde{b},q^{-1}v\right\rangle.
\end{equation}
Collecting the terms in \eqref{eqfortildep}-\eqref{eqfortildeb} containing $(v,c,d,w)$ and equalizing them with the left hand side of \eqref{fullpairing} yields the following system of equations:
\begin{equation*}
\begin{split}
&\langle d,\widetilde{a}\rangle=\langle d,a\rangle,\quad\langle \widetilde{b},w\rangle=\langle b,w\rangle,\quad\langle c,\nu\rangle=\left\langle c,\widetilde{\nu}+\mathrm{ad}_{\widetilde{a}}\left(\mathrm{Ad}_{q^{-1}}u\right)\right\rangle,{}\\&
\langle q^{\dag}p,q^{-1}v\rangle=\left\langle q^{\dag}g^{\dag}\widetilde{p}-\mathrm{ad}^{*}_{\widetilde{\nu}}\left(\mathrm{Ad}^{*}_{q}\beta\right)+\mathrm{ad}^{*}_{\mathrm{Ad}_{q^{-1}}u}\left(\mathrm{ad}^{*}_{\widetilde{a}}\alpha+\widetilde{b}\right)-\mathrm{ad}^{*}_{\widetilde{a}}\left(\mathrm{Ad}^{*}_{q}\sigma\right),q^{-1}v\right\rangle,
\end{split}
\end{equation*}
and since $(v,c,d,w)$ are arbitrary, one deduces \eqref{cotangentliftformula}.
\end{proof}

Together with the base variables transformation \eqref{completeformulaforgroupoperation}, equations \eqref{cotangentliftformula} define the cotangent lift $((G,U),(Q,M,P,A))\mapsto(\widetilde Q,\widetilde M,\widetilde P,\widetilde A)$ of the left $\mathcal{F}$-action on itself:
\begin{equation}
\label{fullformulaforthecotangentlift}
\begin{split}
&\widetilde q=gq,\quad\widetilde\alpha=\mathrm{Ad}^{*}_{q}\beta+\alpha,\quad\widetilde m=\mathrm{Ad}_{q^{-1}}u+m,\quad\widetilde\gamma=\mathrm{Ad}^{*}_{q}\sigma-\mathrm{ad}^{*}_{\mathrm{Ad}_{q^{-1}}u}\alpha+\gamma{},\\&
\widetilde\nu=\nu-\mathrm{ad}_{a}\left(\mathrm{Ad}_{q^{-1}}u\right),\quad \widetilde a=a,\quad\widetilde b=b,{}\\&
\widetilde p=\left(g^{-1}\right)^{\dag}p+\left(g^{-1}\right)^{\dag}\left(q^{-1}\right)^{\dag}\left(\mathrm{ad}^{*}_{\nu-\mathrm{ad}_{a}\left(\mathrm{Ad}_{q^{-1}}u\right)}\left(\mathrm{Ad}^{*}_{q}\beta\right)-\mathrm{ad}^{*}_{\mathrm{Ad}_{q^{-1}}u}\left(\mathrm{ad}^{*}_{a}\alpha+b\right)+\mathrm{ad}^{*}_{a}\left(\mathrm{Ad}^{*}_{q}\sigma\right)\right).
\end{split}
\end{equation}

Consider the infinitesimal version of the $\mathcal{F}$-action \eqref{fullformulaforthecotangentlift} on $T^{*}\mathcal{F}$. To that end, we take
\begin{equation}
\label{groupinfinitesgenerator}
\begin{split}
&G=I+t\mathrm{T}+\mathcal{O}(t^{2})\Longleftrightarrow\left\{
\begin{aligned}
&g=I+t\xi+\mathcal{O}(t^{2}),\quad\xi\in\mathfrak{g}\\
&\beta=0+t\eta+\mathcal{O}(t^{2}),\quad\eta\in\mathfrak{g}^{*},\\
\end{aligned}
\right.
{}\\&
U=O+t\Lambda+\mathcal{O}(t^{2})\Longleftrightarrow\left\{
\begin{aligned}
&u=0+t\lambda+\mathcal{O}(t^{2}),\quad\lambda\in\mathfrak{g}\\
&\sigma=0+t\tau+\mathcal{O}(t^{2}),\quad\tau\in\mathfrak{g}^{*},\\
\end{aligned}
\right.
\end{split}
\end{equation}
where $(\Lambda,\mathrm{T})=((\xi,\eta),(\lambda,\tau))\in\mathfrak{F}$. The infinitesimal generator of transformation \eqref{fullformulaforthecotangentlift} is globally Hamiltonian with the Hamiltonian $\mathcal{H}=\mathcal{H}(q,\alpha,m,\gamma,p,\nu,a,b)$ satisfying
\begin{equation}
\label{systemforhamiltoniancotangentaction}
\begin{split}
&\frac{\delta\mathcal{H}}{\delta p}=\left.\frac{\mathrm{d}\widetilde{q}}{\mathrm{d}t}\right|_{t=0},\quad-\frac{\delta\mathcal{H}}{\delta q}=\left.\frac{\mathrm{d}\widetilde{p}}{\mathrm{d}t}\right|_{t=0},\quad\frac{\delta\mathcal{H}}{\delta \nu}=\left.\frac{\mathrm{d}\widetilde{\alpha}}{\mathrm{d}t}\right|_{t=0},\quad-\frac{\delta\mathcal{H}}{\delta \alpha}=\left.\frac{\mathrm{d}\widetilde{\nu}}{\mathrm{d}t}\right|_{t=0},{}\\&\frac{\delta\mathcal{H}}{\delta a}=\left.\frac{\mathrm{d}\widetilde{\gamma}}{\mathrm{d}t}\right|_{t=0},\quad-\frac{\delta\mathcal{H}}{\delta \gamma}=\left.\frac{\mathrm{d}\widetilde{a}}{\mathrm{d}t}\right|_{t=0}=0,\quad\frac{\delta\mathcal{H}}{\delta b}=\left.\frac{\mathrm{d}\widetilde{m}}{\mathrm{d}t}\right|_{t=0},\quad-\frac{\delta\mathcal{H}}{\delta m}=\left.\frac{\mathrm{d}\widetilde{b}}{\mathrm{d}t}\right|_{t=0}=0.
\end{split}
\end{equation}

Inserting \eqref{groupinfinitesgenerator} into \eqref{fullformulaforthecotangentlift}, the time derivatives in \eqref{systemforhamiltoniancotangentaction} are found as follows:
\begin{equation}
\label{exprforderivatives}
\begin{split}
&\left.\frac{\mathrm{d}\widetilde{q}}{\mathrm{d}t}\right|_{t=0}=\xi q,\quad\left.\frac{\mathrm{d}\widetilde{\alpha}}{\mathrm{d}t}\right|_{t=0}=\mathrm{Ad}^{*}_{q}\eta,\quad
\left.\frac{\mathrm{d}\widetilde{m}}{\mathrm{d}t}\right|_{t=0}=\mathrm{Ad}_{q^{-1}}\lambda,{}\\&\left.\frac{\mathrm{d}\widetilde{\gamma}}{\mathrm{d}t}\right|_{t=0}=\mathrm{Ad}^{*}_{q}\tau-\mathrm{ad}^{*}_{\mathrm{Ad}_{q^{-1}}\lambda}(\alpha),\quad\left.\frac{\mathrm{d}\widetilde{\nu}}{\mathrm{d}t}\right|_{t=0}=-\mathrm{ad}_{a}\left(\mathrm{Ad}_{q^{-1}}\lambda\right),{}\\&
\left.\frac{\mathrm{d}\widetilde{p}}{\mathrm{d}t}\right|_{t=0}=-\xi^{\dag}p+(q^{-1})^{\dag}\left(\mathrm{ad}^{*}_{\nu}\left(\mathrm{Ad}^{*}_{q}\eta\right)-\mathrm{ad}^{*}_{\mathrm{Ad}_{q^{-1}}\lambda}\left(\mathrm{ad}^{*}_{a}\alpha+b\right)+\mathrm{ad}^{*}_{a}\left(\mathrm{Ad}^{*}_{q}\tau\right)\right).
\end{split}
\end{equation}
Solving \eqref{systemforhamiltoniancotangentaction}-\eqref{exprforderivatives} for the Hamiltonian $\mathcal{H}$, one gets
\begin{equation*}
\mathcal{H}=\langle b,\mathrm{Ad}_{q^{-1}}\lambda\rangle+\langle\mathrm{Ad}^{*}_{q}\tau-\mathrm{ad}^{*}_{\mathrm{Ad}_{q^{-1}}\lambda}(\alpha),a\rangle+\langle\mathrm{Ad}^{*}_{q}\eta,\nu\rangle+\langle p,\xi q\rangle.
\end{equation*}
The momentum map $\mu\colon T^{*}\mathcal{F}\to\mathfrak{F}^{*}$ by its definition fulfills
\begin{equation*}
\langle\mu,((\xi,\eta),(\lambda,\tau))\rangle=\mathcal{H}=\langle pq^{\dag},\xi\rangle+\langle\mathrm{\eta},\mathrm{Ad}_{q}\nu\rangle+\left\langle\mathrm{Ad}^{*}_{q^{-1}}\left(b+\mathrm{ad}_{a}^{*}\alpha\right),\lambda\right\rangle+\langle\tau,\mathrm{Ad}_{q}a\rangle,
\end{equation*}
which implies that the momentum map $\mu\colon T^{*}\mathcal{F}\to\mathfrak{F}^{*}$ is
\begin{equation}
\label{momantummapgeneralG}
\begin{split}
\mu((q,\alpha),(m,\gamma),(p,\mu),(a,b))&=\left(pq^{\dag},\mathrm{Ad}_{q}\nu,\mathrm{Ad}^{*}_{q^{-1}}\left(b+\mathrm{ad}_{a}^{*}\alpha\right),\mathrm{Ad}_{q}a\right){}\\&=\left(pq^{\dag},q\nu q^{-1},(q^{-1})^{\dag}(b+[a^{\dag},\alpha])q^{\dag},qaq^{-1}\right),
\end{split}
\end{equation}
where we used \eqref{adjointandcadjonG}. We arrive at the following result:
\begin{theorem}
The momentum map $\mu\colon T^{*}\mathcal{F}\to\mathfrak{F}^{*}$ associated with the left $\mathcal{F}$-action on itself is given by \eqref{momantummapgeneralG}.
\end{theorem}

Formula \eqref{momantummapgeneralG} is obtained in the generality of an arbitrary matrix Lie algebra $\mathfrak{g}$ and the corresponding Lie group $\mathcal{G}$. Specifying $\mathfrak{g}=\mathfrak{su}(N)$ and $\mathcal{G}=\mathrm{SU}(N)$, i.e. assuming that $qq^{\dag}=I$ and $a^{\dag}=-a$, \eqref{momantummapgeneralG} becomes
\begin{equation*}
\mu((q,\alpha),(m,\gamma),(p,\nu),(a,b))=\left(\frac{1}{2}(pq^{\dag}-qp^{\dag}),q\nu q^{\dag},q(b+[\alpha,a])q^{\dag},qaq^{\dag}\right)\in\mathfrak{F}^{*}.
\end{equation*}
\bibliographystyle{spmpsci}
\bibliography{3D_MHD.bib}
\end{document}